\documentclass[dvipsnames,a4paper]{article}

\usepackage[a4paper, margin=1.1in]{geometry}
\usepackage[T1]{fontenc}
\usepackage{lmodern}
\usepackage{microtype}
\usepackage{amsmath,amssymb,amsthm,mathtools,mathrsfs,bm}
\usepackage{enumitem}
\usepackage{xcolor}
\usepackage{hyperref}
\usepackage{pdflscape}
\usepackage{tikz}
\usetikzlibrary{arrows.meta,positioning,fit,calc,fadings}
\usepackage{tikz-cd}
\usepackage{tabularx}
\usepackage{booktabs}
\usepackage{authblk}
\usepackage{braket}
\usepackage{subcaption}
\numberwithin{equation}{section}
\usepackage[en-MT]{datetime2}
\DTMsetstyle{en-MT-numeric}
\usepackage[
backend=biber,
style=phys,
maxbibnames=99,
doi=true,
eprint=true,
giveninits=true
]{biblatex}
\definecolor{myblue}{HTML}{0251D9}
\definecolor{mygreen}{HTML}{10C902}
\definecolor{myred}{HTML}{D90251}
\hypersetup{
	colorlinks=true,
	linkcolor=NavyBlue,
	citecolor=Green,
	urlcolor=OrangeRed
}

\newtheorem{theorem}{Theorem}[section]
\newtheorem{prop}[theorem]{Proposition}

\newtheorem{cor}[theorem]{Corollary}
\newtheorem{corollary}[theorem]{Corollary}
\theoremstyle{definition}
\newtheorem{definition}[theorem]{Definition}

\theoremstyle{remark}
\newtheorem{remark}[theorem]{Remark}

\title{An exact spacetime polymer gas for finite-temperature $\mathbb Z_N$ homological quantum code}

\author[1,2]{Nafiz Ishtiaque\thanks{\texttt{nafiz@simis.cn}}}
\author[3]{Shanto Chakroborty}

\affil[1]{Center for Mathematics and Interdisciplinary Sciences, Fudan University, Shanghai, China}
\affil[2]{Shanghai Institute for Mathematics and Interdisciplinary Sciences (SIMIS), Shanghai, China}
\affil[3]{Department of Theoretical Physics, University of Dhaka, Dhaka, Bangladesh}

\date{}

\newcommand{\Hom}{\mathrm{Hom}}
\newcommand{\Lk}{\mathrm{Lk}}
\newcommand{\tr}{\mathrm{tr}}
\newcommand{\Supp}{\mathrm{Supp}}
\newcommand{\matrixel}[3]{\left\langle #1 \middle| #2 \middle| #3 \right\rangle}
\newcommand{\wt}{\mathrm{wt}}
\newcommand{\Sus}{\operatorname{Sus}}

\begin{document}
	\maketitle

\begin{abstract}
	We study finite-temperature $P$-form $\mathbb Z_N$ homological codes via an exact finite-Trotter quantum-to-classical map to a $(d+1)$-dimensional spacetime model with electric and magnetic topological background charges. The resulting background-resolved partition functions admit an exact reformulation in terms of closed magnetic and electric defect polymers, with opposite-species interactions governed by linking phases. By bounding this complex polymer gas by positive same-species hard-core majorant gases, we obtain an explicit low-activity criterion under which all background-dependent partition functions are uniformly controlled and homologically nontrivial polymers are exponentially suppressed on the scale of the spacetime systole. We also derive an exact higher-form Kramers-Wannier duality exchanging electric and magnetic backgrounds, Wilson and 't Hooft operators, and $P$-form and $(d-P)$-form theories. Finally, for prime $N$, we identify an exact source-free gauge-theory specialization coupled to the plaquette random-cluster model, which imports sharp phase-transition results on special geometries into the spacetime framework.
\end{abstract}
	
	\tableofcontents

\section{Introduction}
\label{sec:introduction}

Homological quantum codes are quantum many-body systems whose ground-state
degeneracy is organized by topology. Physical degrees of freedom are assigned to
cells of a lattice, stabilizers are constructed from cellular boundary and
coboundary operators, and logical operators are represented by nontrivial
homology classes
\cite{Kitaev:1997wr,freedman2002,DennisKitaevLandahlPreskill2002,
	BombinMartinDelgado2007,Terhal:2013vbm}. At zero temperature this gives a
topological encoding of quantum memory. At finite temperature, however, the same
topological structure becomes intertwined with statistical mechanics: thermal
excitations form extended defects, and defect histories that are homologically
nontrivial can alter the encoded data. This interplay between topological
quantum memory, lattice gauge theory, duality, and defect expansions has long
been central to the study of finite-temperature stability
\cite{Wegner:1971app,Savit1980,DennisKitaevLandahlPreskill2002,
	WangHarringtonPreskill2003,Nussinov:2007srh,Castelnovo:2007byc,
	Castelnovo:2008jyy,Bravyi:2008rvr,Bravyi:2009ert,Hastings:2011iig,
	Landon-Cardinal:2012nio,Hastings:2013jxa,Brown:2014idi,Alicki:2008nak,Chesi:2009woi}.

In this paper we develop an exact finite-temperature spacetime framework for
\(P\)-form \(\mathbb Z_N\) homological codes on a finite cell decomposition
\(\Lambda\) of a closed oriented \(d\)-manifold. We allow general local \(X\)- and \(Z\)-type source terms and study the thermal
ensemble itself, rather than the effective disordered statistical-mechanical
models used in decoding-threshold analyses and quenched disorder averages
\cite{DennisKitaevLandahlPreskill2002,WangHarringtonPreskill2003,Nishimori1981}. Our central objects are
background-resolved spacetime partition functions obtained from the thermal
trace by inserting Wilson and 't Hooft operators together with Euclidean twists
around the thermal circle. At fixed Trotter number \(M\), this decorated trace
admits an exact quantum-to-classical reformulation
(Prop.~\ref{prop:ZtwistedSpatial}) as a classical model on the spacetime lattice
\begin{equation}
	\underline\Lambda=\Lambda\times (\mathbb Z/M\mathbb Z),
\end{equation}
with electric and magnetic spacetime backgrounds determined by the chosen
insertions. The finite-temperature problem is therefore organized not by a
single partition function but by a family of background-dependent spacetime
partition functions indexed by spacetime homology classes, simultaneously
encoding ordinary spatial logical operators and twists winding around Euclidean
time. This part of the construction holds for arbitrary \(N\ge 2\).

The main structural development of the paper is an exact reformulation of these
spacetime partition functions in terms of extended defects. Expanding the local weights
around the flat-field regime produces local magnetic and electric defect
variables, and summing over the spacetime field forces them to assemble into
closed defects whose total homology classes are fixed by the chosen electric and
magnetic backgrounds. This closed-defect expansion can be
expressed in terms of a generalized electric-magnetic linking pairing, and for prime $N$,
decomposing the defects into support-connected components yields an exact
two-species polymer gas of connected closed electric and magnetic polymers
(Prop.~\ref{prop:closed-defect-gas-fixed-charge},
Prop.~\ref{prop:fourier-resolved-polymer-gas}). This polymer description is the
main framework developed here. It resolves the finite-temperature ensemble into
spacetime topological sectors, rewrites thermal fluctuations in terms of closed
extended objects, and provides a common language for perturbative analysis and
exact duality. In this picture, the thermally dangerous excitations are exactly
the homologically nontrivial connected polymers.

A first application is perturbative control near the flat-field, or code, limit
(Thm.~\ref{thm:low-activity-region}). The exact two-species polymer gas is
generally complex: opposite-species polymers interact through pure linking
phases rather than positive Boltzmann weights. We control it by comparing its
absolute value to a product of two positive same-species hard-core polymer gases
and applying the criteria of Koteck\'y-Preiss,
Fern\'andez-Procacci, and Bissacot-Fern\'andez-Procacci
\cite{KoteckyPreiss1986,FernandezProcacci2007,
	BissacotFernandezProcacci2010,Ueltschi2004}. This yields an explicit
low-activity region in which all background-dependent partition functions are
uniformly controlled. In that regime, large connected polymers are
exponentially suppressed, and homologically nontrivial polymers are suppressed
on the scale of the corresponding spacetime systole
(Cor.~\ref{cor:physical-occupation-tails}). The weakest point in the resulting
region boundary is a purely combinatorial counting problem for rooted connected closed
polymers, so improvements in that counting problem would directly sharpen the
rigorous low-activity region.

A second general output is an exact higher-form Kramers-Wannier duality
(Prop.~\ref{prop:KW-fixed-charge-duality},
Thm.~\ref{thm:KW-Trotter-duality}). The derivation is again exact and local: one
Fourier-transforms the cell weights and reinterprets the Fourier variables as
fields on the dual spacetime lattice. The result exchanges the \(P\)-form
theory on \(\underline\Lambda\) with a \((d-P)\)-form theory on the dual lattice
\(\underline\Lambda^\vee\), exchanges electric and magnetic backgrounds, and at
the level of the quantum code interchanges \(X\)- and \(Z\)-type stabilizers,
local Weyl operators, and Wilson and 't Hooft logical operators. On self-dual
lattices this becomes an exact \(\mathbb Z_2\)-symmetry of the full spacetime
theory-space. This places the construction in the
higher-form Kramers-Wannier/Wegner duality tradition
\cite{KramersWannier1941,Wegner:1971app,Savit1980,
	ElitzurPearsonShigemitsu1979,UkawaWindeyGuth1980,Cardy1982,
	CardyRabinovici1982,CobaneraOrtizNussinov2011,FreedTeleman2022,
	DelcampIshtiaque2024}. Unlike the linking-based polymer reformulation, this duality holds for
arbitrary \(N\ge 2\).

The general spacetime framework is broader than a positive lattice model: it
allows arbitrary local weights, complex phases, and fixed electric and magnetic
backgrounds. It nevertheless contains natural exact positive slices. The most
important is obtained by removing the \(P\)-cell weights and retaining flatness
weights on \((P+1)\)-cells, which yields \(P\)-form \(N\)-state Potts lattice
gauge theory \cite{Wegner:1971app,FradkinShenker1979,BanksRabinovici1979}. For
prime \(N\), this gauge theory is exactly coupled to the plaquette random-cluster
model by a higher-form Fortuin-Kasteleyn/Edwards-Sokal construction
\cite{FortuinKasteleyn1972,EdwardsSokal1988,HiraokaShirai2016,
	DuncanSchweinhartTopological,DuncanSchweinhartDeconfinement}. This
gauge/PRCM slice provides a sharp application of the general framework. In
particular, for the middle-dimensional self-dual toric geometry
\(\underline\Lambda=T_L^{\,2(P+1)}\), Duncan-Schweinhart proved for odd prime
\(N\) that the associated \((P+1)\)-dimensional PRCM undergoes a sharp
homological phase transition at the self-dual probability
\(p_{\mathrm{sd}}(N)=\frac{\sqrt N}{1+\sqrt N}\)
\cite{DuncanSchweinhartTopological}. Through the exact gauge/PRCM coupling, this
gives a corresponding benchmark on the gauge-theory slice and, in particular, a
gauge-theoretic interpretation of the critical coupling in terms of macroscopic
toric sectors and value-locking of suitably separated toric observables
(Prop.~\ref{prop:toric-gauge-consequences}).

Taken together, these results provide an exact finite-temperature sector
decomposition for homological codes and related lattice models. The thermal
problem is reorganized in terms of background-resolved spacetime partition functions,
closed defect gases, polymer expansions, and exact duality. We expect this
framework to be useful beyond the specific gauge/PRCM application developed
here, both for sharper perturbative questions and for further nonperturbative
specializations.

The paper is organized as follows. In \S\ref{sec:quantum-classical} we define
the quantum model, the Wilson/'t Hooft and Euclidean twist insertions, and the
exact decorated quantum-to-classical map. In
\S\ref{sec:fixed-charge-polymer-gas} we derive the local defect expansion, the
closed-defect gas, and the connected polymer representation. In
\S\ref{sec:low-activity} we identify the low-activity region and prove the
exponential suppression of large and homologically nontrivial polymers. In
\S\ref{sec:KW-duality} we derive the exact higher-form Kramers-Wannier duality
and its action on the quantum code. Finally, \S\ref{sec:specializations}
identifies the exact gauge/PRCM specialization and develops the
middle-dimensional toric application. The appendices collect the finite Fourier
conventions, lattice geometry, generalized linking pairing, and Trotter-kernel
computations used in the main text.

\paragraph{Acknowledgment.}
We thank the ICTP Physics Without Frontiers Bangladesh Program for initiating the collaboration that resulted in this paper. The work of NI was supported by the Research Start-up Fund of the Shanghai Institute for Mathematics and Interdisciplinary Sciences (SIMIS).

\section{Quantum model and decorated quantum-to-classical map}
\label{sec:quantum-classical}

In this section we pass from the quantum \(\mathbb Z_N\) homological code on the spatial
lattice \(\Lambda\) to an exact finite-\(M\) spacetime model on
\begin{equation}
\underline\Lambda=\Lambda\times S^1_M \qquad \text{where,} \qquad S_M^1 := \mathbb Z/M\mathbb Z.
\end{equation}
We define the quantum model, the spatial
Wilson and 't Hooft operators, the Euclidean twist operators, and the resulting spacetime
background partition functions.

Throughout, \(\Lambda\) is a finite cell decomposition of a closed
oriented \(d\)-manifold, and the code degrees of freedom live on the \(P\)-cells of
\(\Lambda\). For a cell complex \(X\), we write \(C_p(X)\) for its oriented
\(p\)-cells, \(\mathcal C_p(X)\) for its \(p\)-chains, and \(\Omega^p(X)\) for its
\(p\)-cochains, as reviewed in Appendix~\ref{app:geometryBasics}. All computations of this section, culminating in Prop.~\ref{prop:ZtwistedSpatial}, work for any integer $N \ge 2$.

\subsection{Code Hamiltonian with generic sources}

To each \(P\)-cell \(b\in C_P(\Lambda)\) we assign an \(N\)-dimensional Hilbert space
\(\mathcal H_b\cong\mathbb C^N\), and define
\begin{equation}
	\mathcal H
	:=
	\bigotimes_{b\in C_P(\Lambda)}\mathcal H_b .
	\label{eq:H-total}
\end{equation}
Fix an orthonormal basis \(\{\ket{j}\}_{j\in\mathbb Z_N}\) of \(\mathbb C^N\), and set
\begin{equation}
\omega:=e^{2\pi i/N}.
\end{equation}
On \(\mathbb C^N\), define the shift and clock operators by
\begin{equation}
	\hat X\ket{j}=\ket{j+1},
	\qquad
	\hat Z\ket{j}=\omega^j\ket{j},
	\qquad
	j\in\mathbb Z_N.
	\label{eq:XZ}
\end{equation}
They obey
\begin{equation}
	\hat Z^r\hat X^s=\omega^{rs}\hat X^s\hat Z^r,
	\qquad
	\hat X^N=\hat Z^N=\hat 1,
	\qquad
	\hat X^\dagger=\hat X^{-1},
	\qquad
	\hat Z^\dagger=\hat Z^{-1}.
	\label{eq:Weyl}
\end{equation}
For any operator \(\hat O\) on \(\mathbb C^N\), we write \(\hat O_b\) for the operator on
\(\mathcal H\) acting as \(\hat O\) on \(\mathcal H_b\) and as the identity on all other
tensor factors. We also define the Kronecker delta and the fundamental character of $\mathbb Z_N$:
\begin{equation}
	\delta_N(x)
	:=
	\begin{cases}
		1,& x=0\in\mathbb Z_N,\\
		0,& x\neq0\in\mathbb Z_N,
	\end{cases}
	\qquad
	\chi(x):=\omega^x .
	\label{eq:delta-and-character}
\end{equation}

For \(a\in C_{P-1}(\Lambda)\) and \(c\in C_{P+1}(\Lambda)\), define the stabilizer operators
\begin{equation}
	\hat A_a
	:=
	\prod_{b\supset a}\hat X_b^{\,(-1)^P\varepsilon(b,a)},
	\qquad
	\hat B_c
	:=
	\prod_{b\subset c}\hat Z_b^{\,\varepsilon(c,b)} .
	\label{eq:AB}
\end{equation}
The sign \((-1)^P\) in \(\hat A_a\) is the convention compatible with the vertical
component of the spacetime coboundary \(D\) on \(\underline\Lambda=\Lambda\times S^1_M\) (see \eqref{eq:D-components}). The incidence identity \eqref{eq:incidence-identity} (equivalent to \(\partial^2=0\)) implies
that all \(\hat A_a\) commute with all \(\hat B_c\), and the two families are separately
commuting.

Define the stabilizer projectors
\begin{equation}
	\hat{\mathcal A}_a
	:=
	\frac1N\sum_{m\in\mathbb Z_N}\hat A_a^m,
	\qquad
	\hat{\mathcal B}_c
	:=
	\frac1N\sum_{m\in\mathbb Z_N}\hat B_c^m.
	\label{eq:projectors}
\end{equation}
Let \(J_a>0\) and \(K_c>0\) be couplings on \((P-1)\)- and \((P+1)\)-cells, and let
\(g_b^{(n)}\), \(h_b^{(n)}\) be local \(X\)- and \(Z\)-source coefficients on \(P\)-cells.
We define
\begin{equation}
	\hat H_0
	=
	-\sum_{a\in C_{P-1}(\Lambda)}J_a\hat{\mathcal A}_a
	-\sum_{c\in C_{P+1}(\Lambda)}K_c\hat{\mathcal B}_c,
	\label{eq:H0}
\end{equation}
\begin{equation}
	\hat H_1
	=
	-\sum_{b\in C_P(\Lambda)}
	\sum_{n\in\mathbb Z_N}
	\Bigl(
	g_b^{(n)}\hat X_b^n+h_b^{(n)}\hat Z_b^n
	\Bigr),
	\label{eq:H1}
\end{equation}
and the total Hamiltonian
\begin{equation}
	\hat H:=\hat H_0+\hat H_1 .
	\label{eq:H}
\end{equation}
Another separation of the Hamiltonian is into its $X$- and $Z$-dependent parts: 
\begin{equation}
	\hat H=\hat H_x+\hat H_z,
\end{equation}
with
\begin{equation}
	\hat H_x
	=
	-\sum_{a\in C_{P-1}(\Lambda)}J_a\hat{\mathcal A}_a
	-\sum_{b\in C_P(\Lambda)}
	\sum_{n\in\mathbb Z_N}g_b^{(n)}\hat X_b^n,
	\label{eq:Hx}
\end{equation}
and
\begin{equation}
	\hat H_z
	=
	-\sum_{c\in C_{P+1}(\Lambda)}K_c\hat{\mathcal B}_c
	-\sum_{b\in C_P(\Lambda)}
	\sum_{n\in\mathbb Z_N}h_b^{(n)}\hat Z_b^n.
	\label{eq:Hz}
\end{equation}
All terms within \(\hat H_x\) commute with one another, and all terms within \(\hat H_z\)
commute with one another. The formulas below make sense for complex source parameters. If
one wants \(\hat H\) Hermitian, one should impose
\begin{equation}
g_b^{(n)}=\overline{g_b^{(-n)}},
\qquad
h_b^{(n)}=\overline{h_b^{(-n)}}.
\end{equation}

\subsection{Spatial Wilson and 't Hooft operators}
\label{sec:spatial-topological-operators}

In the absence of sources, i.e., with $g=h=0$, the code subspace is
\begin{equation}
	\mathcal H_{\mathrm{code}}
	:=
	\bigcap_{a\in C_{P-1}(\Lambda)}\ker(\hat A_a-\hat 1)
	\cap
	\bigcap_{c\in C_{P+1}(\Lambda)}\ker(\hat B_c-\hat 1).
	\label{eq:code-subspace}
\end{equation}

Let $\nu=\sum_{b\in C_P(\Lambda)}\nu_b\,b\in Z_P(\Lambda)$ be a spatial \(P\)-cycle and $\mu^\vee\in Z_{d-P}(\Lambda^\vee)$ be a $(d-P)$-cycle of the dual lattice. Then Wilson and 't Hooft operators supported on $\nu$ and $\mu^\vee$ respectively can be defined as\footnote{The sign in the 't Hooft operator was chosen so that it will result in a clean cocycle shift in the classical spacetime model based on our conventions for lattice differentials (see \eqref{eq:Zsp}).}
\begin{align}
	\hat{\mathcal W}_\nu
	:=
	\prod_{b\in C_P(\Lambda)}\hat Z_b^{\,\nu_b},
	\qquad
	\hat{\mathcal T}_{\mu^\vee}
	:=
	\prod_{b\in C_P(\Lambda)}
	\hat X_b^{\,(-1)^P I_\Lambda(\mu^\vee,b)}.
	\label{eq:quantumWT}
\end{align}
Some well-known algebraic properties of these operators are
\begin{prop}
	\label{prop:spatial-WT}
	Let \(\nu\in Z_P(\Lambda)\) and
	\(\mu^\vee\in Z_{d-P}(\Lambda^\vee)\).
	\begin{enumerate}[label=(\roman*)]
		\item\label{prop:WTpreserves} The operators \(\hat{\mathcal W}_\nu\) and
		\(\hat{\mathcal T}_{\mu^\vee}\) preserve the code subspace
		\(\mathcal H_{\mathrm{code}}\).
		
		\item\label{prop:Wtrivial} If \(\nu\in B_P(\Lambda)\), then \(\hat{\mathcal W}_\nu\) is a product of
		\(\hat B_c\)-stabilizers and hence acts trivially on
		\(\mathcal H_{\mathrm{code}}\). Therefore \(\hat{\mathcal W}_\nu\) depends only on
		\([\nu]\in H_P(\Lambda)\) when restricted to the code subspace.
		
		\item\label{prop:Ttrivial} If \(\mu^\vee\in B_{d-P}(\Lambda^\vee)\), then
		\(\hat{\mathcal T}_{\mu^\vee}\) is a product of \(\hat A_a\)-stabilizers and hence
		acts trivially on \(\mathcal H_{\mathrm{code}}\). Therefore
		\(\hat{\mathcal T}_{\mu^\vee}\) depends only on
		\([\mu^\vee]\in H_{d-P}(\Lambda^\vee)\) when restricted to the code subspace.
		
		\item\label{prop:WTcommutator} The mixed commutator is
		\begin{equation}
			\hat{\mathcal T}_{\mu^\vee}\hat{\mathcal W}_\nu
			=
			\omega^{-(-1)^P I_\Lambda(\mu^\vee,\nu)}
			\hat{\mathcal W}_\nu\hat{\mathcal T}_{\mu^\vee}.
		\end{equation}
	\end{enumerate}
\end{prop}
These follow from elementary algebra, we reproduce the proof in Appendix~\ref{app:WTalgebra} for completion.

\subsection{Trotter weights and neutral spacetime model}
\label{sec:trotter-weights}

The thermal partition function at inverse temperature \(\beta\) is
\begin{equation}
	Z
	:=
	\tr_{\mathcal H}\left(e^{-\beta\hat H}\right).
	\label{eq:Z-quantum-def}
\end{equation}
For \(M\ge1\), define the Trotterized partition function
\begin{equation}
	Z_M
	:=
	\tr_{\mathcal H}
	\left[
	\left(
	e^{-\beta\hat H_z/M}e^{-\beta\hat H_x/M}
	\right)^M
	\right],
	\qquad
	Z=\lim_{M\to\infty}Z_M.
	\label{eq:ZM-def}
\end{equation}
The limit equation is an application of the Lie/Trotter product formula. All identities below are exact at fixed \(M\).

We define local weight functions which we shall use to write traces such as \eqref{eq:ZM-def}. For horizontal
\((P+1)\)-cells \(c(i)=c\times\{i\}\) and horizontal \(P\)-cells \(b(i)=b\times\{i\}\), set
\begin{subequations}
	\begin{align}
		W_{c(i)}^{\parallel}:\mathbb Z_N \to&\; \mathbb C,
		&\qquad
		W_{c(i)}^{\parallel}(x)
		&:=
		\exp\left(\frac{\beta K_c}{M}\delta_N(x)\right)
		=
		1+\left(e^{\beta K_c/M}-1\right)\delta_N(x),
		\label{eq:W-parallel}
		\\
		V_{b(i)}^{\parallel}:\mathbb Z_N \to&\; \mathbb C,
		&
		V_{b(i)}^{\parallel}(x)
		&:=
		\exp\left(
		\frac{\beta}{M}
		\sum_{n\in\mathbb Z_N}h_b^{(n)}\omega^{nx}
		\right).
		\label{eq:V-parallel}
	\end{align}
	\label{eq:parallel-weights}
\end{subequations}
For vertical \(P\)-cells \(\underline a(i)=a\times[i,i+1]\) and vertical
\((P+1)\)-cells \(\underline b(i)=b\times[i,i+1]\), set
\begin{subequations}
	\begin{align}
		V_{\underline a(i)}^{\perp}:\mathbb Z_N \to&\; \mathbb C,
		&\qquad
		V_{\underline a(i)}^{\perp}(x)
		&:=
		\delta_N(x)+\frac{e^{\beta J_a/M}-1}{N},
		\label{eq:V-perp}
		\\
		W_{\underline b(i)}^{\perp}:\mathbb Z_N \to&\; \mathbb C,
		&
		W_{\underline b(i)}^{\perp}(x)
		&:=
		\sum_{j\in\mathbb Z_N}
		\exp\left(
		\frac{\beta}{M}
		\sum_{n\in\mathbb Z_N}g_b^{(n)}\omega^{nj}
		\right)
		\frac{\omega^{(-1)^Pjx}}{N}.
		\label{eq:W-perp}
	\end{align}
	\label{eq:perp-weights}
\end{subequations}
When no confusion is possible, we write simply \(W_c\) for either \(W_c^\parallel\) or
\(W_c^\perp\), depending on whether \(c\in C_{P+1}(\underline\Lambda)\) is horizontal or vertical,
and similarly \(V_u\) for \(u\in C_P(\underline\Lambda)\). Fig.~\ref{fig:space2spacetime} shows how cells of the spatial lattice transform into horizontal and vertical cells of the spacetime lattice.

\begin{prop}[Neutral thermal partition function] The finite $M$ partition function $Z_M$ defined in \eqref{eq:ZM-def} is given by
	\begin{equation}
		Z_M
		=
		\sum_{\phi\in\Omega^P(\underline\Lambda)}
		\prod_{c\in C_{P+1}(\underline\Lambda)}
		W_c\left((D\phi)_c\right)
		\prod_{u\in C_P(\underline\Lambda)}
		V_u(\phi_u).
		\label{eq:Z-0}
	\end{equation}
\end{prop}
\begin{proof}
	This is \eqref{eq:untwisted-spacetime}, derived via the kernel computations in Appendix~\ref{app:vanillaTr}.
\end{proof}
Forms and differential on the spacetime lattice are defined in Appendix~\ref{app:geometryBasics}.

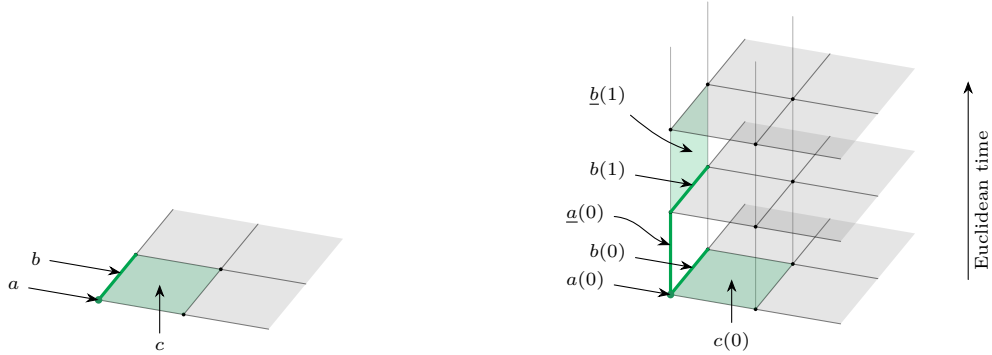
\begin{figure}[h]
	\centering
	
	\begin{subfigure}[t]{0.48\textwidth}
		\centering
		\begin{tikzpicture}[
			x={(3.2em,-.55em)},
			y={(1.4em,1.7em)},
			z={(0em,3.1em)},
			line cap=round,
			line join=round,
			dot/.style={circle,fill,inner sep=.5pt},
			spatedge/.style={opacity=.5}
			]
			
			\def\N{1}
			
			\pgfmathtruncatemacro{\Vmax}{\N}
			\pgfmathtruncatemacro{\Vmaxx}{\N+1}
			\pgfmathtruncatemacro{\Emax}{\N-1}
			
			\fill[gray, opacity=0.2] (0,0,0) -- (\Vmaxx,0,0) -- (\Vmaxx,\Vmaxx,0) -- (0,\Vmaxx,0) -- cycle;
			
			
			\foreach \j in {0,...,\Vmax}{
				\foreach \i in {0,...,\numexpr\Emax+1\relax}{
					\draw[spatedge] (\i,\j,0) -- ++(1,0,0);
				}
			}
			
			\foreach \i in {0,...,\Vmax}{
				\foreach \j in {0,...,\numexpr\Emax+1\relax}{
					\draw[spatedge] (\i,\j,0) -- ++(0,1,0);
				}
			}
			
			\foreach \i in {0,...,\Vmax}{
				\foreach \j in {0,...,\Vmax}{
					\node[dot] at (\i,\j,0) {};
				}
			}
			
			\coordinate (a) at (0,0,0);
			\coordinate (c) at (.5,.5,0);
			\coordinate (b) at (0,.6,0);
			
			\fill[Green, opacity=0.2] (0,0,0) -- (1,0,0) -- (1,1,0) -- (0,1,0) -- cycle;
			\draw[Green, very thick] (0,0,0) -- (0,1,0);
			\node[circle,fill=Green!75!black,inner sep=1pt] at (0,0,0) {};
			
			\node (name_a) at ($(a)+(-1,0,0)$) {{\scriptsize $a$}};
			\draw[-Stealth] (name_a) to (a);
			\node (name_b) at ($(b)+(-1,0,0)$) {{\scriptsize $b$}};
			\draw[-Stealth] (name_b) to (b);
			\node (name_c) at ($(c)+(0,0,-.75)$) {{\scriptsize $c$}};
			\draw[-Stealth] (name_c) to (c);
		\end{tikzpicture}
		\caption{The spatial lattice \(\Lambda\) in the case \(d=2\), with a distinguished vertex \(a\in C_0(\Lambda)\), edge \(b\in C_1(\Lambda)\), and plaquette \(c\in C_2(\Lambda)\).}
		\label{fig:spatialL}
	\end{subfigure}
	\hfill
	\begin{subfigure}[t]{0.48\textwidth}
		\centering
		\begin{tikzpicture}[
			x={(3.2em,-.55em)},
			y={(1.4em,1.7em)},
			z={(0em,3.1em)},
			line cap=round,
			line join=round,
			dot/.style={circle,fill,inner sep=.5pt},
			spatedge/.style={opacity=.5}
			]
			
			\def\N{1}
			
			\pgfmathtruncatemacro{\Vmax}{\N}
			\pgfmathtruncatemacro{\Vmaxx}{\N+1}
			\pgfmathtruncatemacro{\Emax}{\N-1}
			\pgfmathtruncatemacro{\Lmax}{\N+1}
			
			\foreach \k in {0,...,\Lmax}{
				\fill[gray, opacity=0.2] (0,0,\k) -- (\Vmaxx,0,\k) -- (\Vmaxx,\Vmaxx,\k) -- (0,\Vmaxx,\k) -- cycle;
				
				
				\foreach \j in {0,...,\Vmax}{
					\foreach \i in {0,...,\numexpr\Emax+1\relax}{
						\draw[spatedge] (\i,\j,\k) -- ++(1,0,0);
					}
				}
				
				\foreach \i in {0,...,\Vmax}{
					\foreach \j in {0,...,\numexpr\Emax+1\relax}{
						\draw[spatedge] (\i,\j,\k) -- ++(0,1,0);
					}
				}
				
				\ifnum\k<\numexpr\Lmax+1\relax
				\foreach \i in {0,...,\Vmax}{
					\foreach \j in {0,...,\Vmax}{
						\draw[thin,opacity=.3] (\i,\j,\k) -- ++(0,0,1);
					}
				}
				\fi
				
				\foreach \i in {0,...,\Vmax}{
					\foreach \j in {0,...,\Vmax}{
						\node[dot] at (\i,\j,\k) {};
					}
				}
				
				\ifnum\k=0
				\coordinate (a) at (0,0,0);
				\coordinate (c) at (.5,.5,0);
				\coordinate (b0) at (0,.6,0);
				\coordinate (avert) at (0,0,.6);
				
				\fill[Green, opacity=0.2] (0,0,0) -- (1,0,0) -- (1,1,0) -- (0,1,0) -- cycle;
				\draw[Green, very thick] (0,0,0) -- (0,0,1);
				\draw[Green, very thick] (0,0,0) -- (0,1,0);
				\node[circle,fill=Green!75!black,inner sep=1pt] at (0,0,0) {};
				\fi
				\ifnum\k=1
				\coordinate (b) at (0,.6,1);
				\coordinate (bvert) at (0,.6,1.4);
				
				\fill[Green, opacity=0.2] (0,0,1) -- (0,1,1) -- (0,1,2) -- (0,0,2) -- cycle;
				\draw[Green, very thick] (0,0,1) -- (0,1,1);
				\fi
			}
			\node (name_a) at ($(a)+(-1,0,0)$) {{\scriptsize $a(0)$}};
			\draw[-Stealth] (name_a) to (a);
			\node (name_b) at ($(b)+(-1,0,0)$) {{\scriptsize $b(1)$}};
			\draw[-Stealth] (name_b) to (b);
			\node (name_b0) at ($(b0)+(-1,0,0)$) {{\scriptsize $b(0)$}};
			\draw[-Stealth] (name_b0) to (b0);
			\node (name_c) at ($(c)+(0,0,-.75)$) {{\scriptsize $c(0)$}};
			\draw[-Stealth] (name_c) to (c);
			\node (name_bvert) at ($(bvert)+(-1,0,.5)$) {{\scriptsize $\underline b(1)$}};
			\draw[-Stealth] (name_bvert) to[out=-40,in=160] (bvert);
			\node (name_avert) at ($(avert)+(-1,0,.2)$) {{\scriptsize $\underline a(0)$}};
			\draw[-Stealth] (name_avert) to[out=-5,in=160] (avert);
			
			\draw[-Stealth]
			([xshift=2em, yshift=3em]current bounding box.south east) --
			([yshift=-3em]current bounding box.north east)
			node[pos=0, right=.5em, rotate=90] {\scriptsize Euclidean time};
		\end{tikzpicture}
		\caption{The spacetime lattice \(\underline{\Lambda}\) obtained by stacking copies of \(\Lambda\) along the Euclidean time direction. The spatial cells \(a\), \(b\), and \(c\) give rise to horizontal cells \(a(i)\), \(b(i)\), and \(c(i)\) in the \(i\)-th time slice, and to vertical cells \(\underline a(i)\), \(\underline b(i)\), and \(\underline c(i)\) connecting adjacent slices.}
		\label{fig:spacetimeL}
	\end{subfigure}
	
	\caption{Passage from the spatial square lattice \(\Lambda\) to the spacetime lattice \(\underline{\Lambda}\) in the example \(d=2\), \(P=1\). All bold/green-shaded cells are emphasized only for visual clarity.}
	\label{fig:space2spacetime}
\end{figure}

\subsection{Spatial background}
\label{sec:trotter-with-insertions}

We first decorate the Trotterized trace by a spatial Wilson and 't Hooft insertion. Given
\begin{equation}
\nu\in Z_P(\Lambda),
\qquad
\mu^\vee\in Z_{d-P}(\Lambda^\vee),
\end{equation}
define
\begin{equation}
	Z_M^{\mathrm{sp}}(\mu^\vee,\nu)
	:=
	\tr_{\mathcal H}
	\left[
	e^{-\beta\hat H_z/M}\,
	\hat{\mathcal W}_\nu\,
	e^{-\beta\hat H_x/M}\,
	\hat{\mathcal T}_{\mu^\vee}\,
	\left(e^{-\beta\hat H_z/M}e^{-\beta\hat H_x/M}\right)^{M-1}
	\right].
	\label{eq:WTtrotter}
\end{equation}
The Wilson operator is placed on the initial time slice, while the 't Hooft operator acts
on the first slab.

Let \(\nu(0)\in Z_P(\underline\Lambda)\) denote the horizontal lift of \(\nu\) to the
initial time slice. Let \(\mu^\vee(0)\in Z_{d-P}(\underline\Lambda^\vee)\) denote the dual
cycle supported on the first dual slab and characterized by
\begin{equation}
	\left(\vartheta_{\underline\Lambda}(\mu^\vee(0))\right)_{\underline b(0)}
	=
	(\vartheta_\Lambda(\mu^\vee))_b,
	\qquad
	\left(\vartheta_{\underline\Lambda}(\mu^\vee(0))\right)_{c(i)}=0,
	\qquad
	\left(\vartheta_{\underline\Lambda}(\mu^\vee(0))\right)_{\underline b(i)}=0
	\quad(i\neq0).
	\label{eq:eta-mu0}
\end{equation}
Fig.~\ref{fig:WT} illustrates how $\nu$ and $\mu^\vee$ are represented on the spacetime lattice.

\begin{prop}[Spatially decorated thermal partition function]
	\label{prop:spatially-decorated-spacetime}
	For every \(M\), \(\nu\in Z_P(\Lambda)\), and
	\(\mu^\vee\in Z_{d-P}(\Lambda^\vee)\),
	\begin{equation}
		Z_M^{\mathrm{sp}}(\mu^\vee,\nu)
		=
		\sum_{\phi\in\Omega^P(\underline\Lambda)}
		\chi\left(\int_{\nu(0)}\phi\right)
		\prod_{c\in C_{P+1}(\underline\Lambda)}
		W_c\left((D\phi+\vartheta_{\underline\Lambda}(\mu^\vee(0)))_c\right)
		\prod_{u\in C_P(\underline\Lambda)}
		V_u(\phi_u).
		\label{eq:Zsp}
	\end{equation}
\end{prop}

\begin{proof}
	This is \eqref{eq:spatial-insertions-spacetime}. The Wilson insertion gives the character
	\(\chi(\int_{\nu(0)}\phi)\), while the 't Hooft insertion produces the vertical
	cocycle shift \(\vartheta_{\underline\Lambda}(\mu^\vee(0))\).
\end{proof}

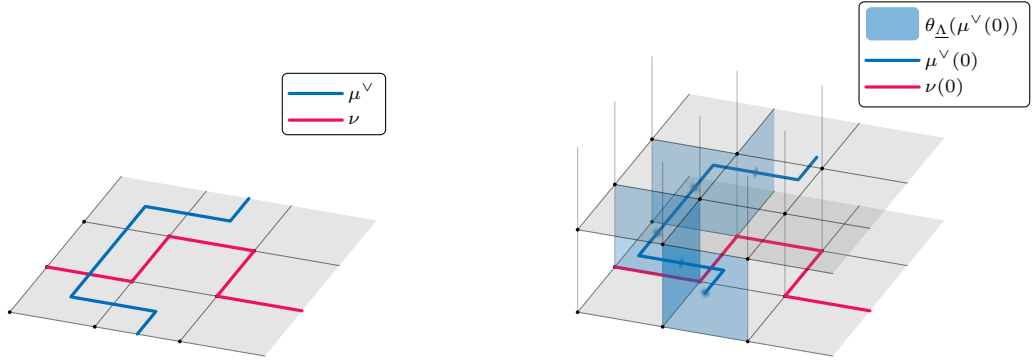
\begin{figure}[h]
	\centering
	
	\begin{subfigure}[t]{0.48\textwidth}
		\centering
		\begin{tikzpicture}[
			x={(3.2em,-.55em)},
			y={(1.4em,1.7em)},
			z={(0em,3.1em)},
			line cap=round,
			line join=round,
			dot/.style={circle,fill,inner sep=.5pt},
			spatedge/.style={opacity=.5}
			]
			
			\def\N{2}
			
			\pgfmathtruncatemacro{\Vmax}{\N}
			\pgfmathtruncatemacro{\Vmaxx}{\N+1}
			\pgfmathtruncatemacro{\Emax}{\N-1}
			
			\fill[gray, opacity=0.2] (0,0,0) -- (\Vmaxx,0,0) -- (\Vmaxx,\Vmaxx,0) -- (0,\Vmaxx,0) -- cycle;
			
			
			\foreach \j in {0,...,\Vmax}{
				\foreach \i in {0,...,\numexpr\Emax+1\relax}{
					\draw[spatedge] (\i,\j,0) -- ++(1,0,0);
				}
			}
			
			\foreach \i in {0,...,\Vmax}{
				\foreach \j in {0,...,\numexpr\Emax+1\relax}{
					\draw[spatedge] (\i,\j,0) -- ++(0,1,0);
				}
			}
			
			\foreach \i in {0,...,\Vmax}{
				\foreach \j in {0,...,\Vmax}{
					\node[dot] at (\i,\j,0) {};
				}
			}
			
			\draw[OrangeRed, very thick] (0,1,0) -- (1,1,0) -- (1,2,0) -- (2,2,0) -- (2,1,0) -- (3,1,0);
			\draw[NavyBlue, very thick] (1.5,0,0) -- (1.5,.5,0) -- (0.5,.5,0) -- (0.5,2.5,0) -- (1.5,2.5,0) -- (1.5,3,0);
			
			\node[
			draw=black,
			fill=white,
			rounded corners=2pt,
			inner sep=2pt,
			anchor=south
			] at ($(current bounding box.east)+(-.5,0,1.5)$) {%
				\begin{tikzpicture}[baseline={(0,0)},x={(1em,0em)},y={(0em,1em)},]
					\draw[OrangeRed, very thick] (0,0) -- (2,0);
					\node[right] at (2,0) {\scriptsize$\nu$};
					
					\draw[NavyBlue, very thick] (0,1) -- (2,1);
					\node[right] at (2,1) {\scriptsize$\mu^\vee$};
				\end{tikzpicture}%
			};
		\end{tikzpicture}
		\caption{The spatial lattice \(\Lambda\) for the case \(d=2\), \(P=1\), with the supports
			\(\nu\subset \Lambda\) and \(\mu^\vee\subset \Lambda^\vee\) of the Wilson and 't~Hooft
			operators inserted in the decorated spatial trace. The support \(\nu\) is drawn on the
			primal lattice \(\Lambda\), while the support \(\mu^\vee\) consists of 1-cells of the dual lattice
			\(\Lambda^\vee\).}
		\label{fig:WTspatial}
	\end{subfigure}
	\hfill
	\begin{subfigure}[t]{0.48\textwidth}
		\centering
		\newcommand{\fadediskxz}[4]{%
			\foreach \rr/\op in {#4/0.08,0.85*#4/0.12,0.70*#4/0.18,0.55*#4/0.28,0.40*#4/0.42,0.25*#4/0.65,0.12*#4/1.0}{
				\fill[NavyBlue!80!black, opacity=\op]
				plot[domain=0:360, samples=50]
				({#1 + \rr*cos(\x)}, {#2}, {#3 + \rr*sin(\x)}) -- cycle;
			}%
		}
		
		\newcommand{\fadediskyz}[4]{%
			\foreach \rr/\op in {#4/0.08,0.85*#4/0.12,0.70*#4/0.18,0.55*#4/0.28,0.40*#4/0.42,0.25*#4/0.65,0.12*#4/1.0}{
				\fill[NavyBlue!70!black, opacity=\op]
				plot[domain=0:360, samples=50]
				({#1}, {#2 + \rr*cos(\x)}, {#3 + \rr*sin(\x)}) -- cycle;
			}%
		}
		\begin{tikzpicture}[
			x={(3.2em,-.55em)},
			y={(1.4em,1.7em)},
			z={(0em,3.1em)},
			line cap=round,
			line join=round,
			dot/.style={circle,fill,inner sep=.5pt},
			spatedge/.style={opacity=.5}
			]
			
			\def\N{2}
			
			\pgfmathtruncatemacro{\Vmax}{\N}
			\pgfmathtruncatemacro{\Vmaxx}{\N+1}
			\pgfmathtruncatemacro{\Emax}{\N-1}
			\pgfmathtruncatemacro{\Lmax}{\N-1}
			
			\foreach \k in {0,...,\Lmax}{
				\fill[gray, opacity=0.2] (0,0,\k) -- (\Vmaxx,0,\k) -- (\Vmaxx,\Vmaxx,\k) -- (0,\Vmaxx,\k) -- cycle;
				
				
				\foreach \j in {0,...,\Vmax}{
					\foreach \i in {0,...,\numexpr\Emax+1\relax}{
						\draw[spatedge] (\i,\j,\k) -- ++(1,0,0);
					}
				}
				
				\foreach \i in {0,...,\Vmax}{
					\foreach \j in {0,...,\numexpr\Emax+1\relax}{
						\draw[spatedge] (\i,\j,\k) -- ++(0,1,0);
					}
				}
				
				\ifnum\k<\numexpr\Lmax+1\relax
				\foreach \i in {0,...,\Vmax}{
					\foreach \j in {0,...,\Vmax}{
						\draw[thin,opacity=.3] (\i,\j,\k) -- ++(0,0,1);
					}
				}
				\fi
				
				\foreach \i in {0,...,\Vmax}{
					\foreach \j in {0,...,\Vmax}{
						\node[dot] at (\i,\j,\k) {};
					}
				}
				
				\ifnum\k=0
					\draw[OrangeRed, very thick] (0,1,0) -- (1,1,0) -- (1,2,0) -- (2,2,0) -- (2,1,0) -- (3,1,0);
					\draw[NavyBlue, very thick] (1.5,0,.5) -- (1.5,.5,.5) -- (0.5,.5,.5) -- (0.5,2.5,.5) -- (1.5,2.5,.5) -- (1.5,3,.5);
					
					
					\fadediskxz{1.5}{0}{0.5}{0.07}
					\fadediskyz{1}{0.5}{0.5}{0.08}
					\fadediskxz{0.5}{1}{0.5}{0.07}
					\fadediskxz{0.5}{2}{0.5}{0.07}
					\fadediskyz{1}{2.5}{0.5}{0.08}
					
					\fill[NavyBlue, opacity=.3] (1,0,0) -- (2,0,0) -- (2,0,1) -- (1,0,1) -- cycle;
					\fill[NavyBlue, opacity=.3] (1,0,0) -- (1,1,0) -- (1,1,1) -- (1,0,1) -- cycle;
					\fill[NavyBlue, opacity=.3] (1,1,0) -- (0,1,0) -- (0,1,1) -- (1,1,1) -- cycle;
					\fill[NavyBlue, opacity=.3] (1,2,0) -- (0,2,0) -- (0,2,1) -- (1,2,1) -- cycle;
					\fill[NavyBlue, opacity=.3] (1,2,0) -- (1,3,0) -- (1,3,1) -- (1,2,1) -- cycle;
				\fi
			}
			
			\node[
			draw=black,
			fill=white,
			rounded corners=2pt,
			inner sep=2pt,
			anchor=south
			] at ($(current bounding box.east)+(0,0,1.2)$) {%
				\begin{tikzpicture}[baseline={(0,0)},x={(1em,0em)},y={(0em,1em)},]
					\draw[OrangeRed, very thick] (0,0) -- (2,0);
					\node[right] at (2,0) {\scriptsize$\nu(0)$};
					
					\draw[NavyBlue, very thick] (0,1) -- (2,1);
					\node[right] at (2,1) {\scriptsize$\mu^\vee(0)$};
					
					\fill[NavyBlue, opacity=.5] (0,1.75) rectangle (2,3);
					\node[right] at (2,2.25) {\scriptsize$\theta_{\underline\Lambda}(\mu^\vee(0))$};
				\end{tikzpicture}%
			};
		\end{tikzpicture}
		\caption{The corresponding spacetime supports in \(\underline{\Lambda}\). The support
			\(\nu(0)\subset \underline{\Lambda}\) of the Wilson insertion remains on the \(0\)-th time
			slice, whereas the support \(\mu^\vee(0)\subset \underline{\Lambda}^\vee\) of the
			't~Hooft insertion lies between the \(0\)-th and \(1\)-st slices. The blue-shaded plaquettes form the support of \(\theta_{\underline{\Lambda}}(\mu^\vee(0))\), a \(2\)-chain on the primal spacetime lattice defined by \eqref{eq:theta}; equivalently, it consists of the primal plaquettes dual to the edges of \(\mu^\vee(0)\). The 2-chain $\theta_{\underline\Lambda}(\mu^\vee(0))$ is the support of the 2-form $\vartheta_{\underline\Lambda}(\mu^\vee(0))$.}
		\label{fig:WTspacetime}
	\end{subfigure}
	
	\caption{From spatial insertions to their spacetime representatives in the case
		\(d=2\), \(P=1\). The dual lattices \(\Lambda^\vee\) and
		\(\underline{\Lambda}^\vee\) are not shown.}
	\label{fig:WT}
\end{figure}

\subsection{Euclidean twists and generic background}
\label{sec:euclidean-twists}

We now introduce twists inserted uniformly around the Euclidean time circle. These generate
spacetime backgrounds wrapping the thermal direction.

\subsubsection{Electric twist.}
Let
\begin{equation}
\alpha=\sum_{a\in C_{P-1}(\Lambda)}\alpha_a\,a
\in Z_{P-1}(\Lambda).
\end{equation}
For each \(a\in C_{P-1}(\Lambda)\), define the spectral projectors of \(\hat A_a\) by
\begin{equation}
	\hat\Pi_{a,j}^{(A)}
	:=
	\frac1N\sum_{m=0}^{N-1}\omega^{-jm}\hat A_a^{\,m},
	\qquad
	j\in\mathbb Z_N.
\end{equation}
They satisfy
\begin{equation}
\hat A_a\hat\Pi_{a,j}^{(A)}
=
\omega^j\hat\Pi_{a,j}^{(A)},
\qquad
\sum_{j\in\mathbb Z_N}\hat\Pi_{a,j}^{(A)}=\hat 1,
\qquad
\hat{\mathcal A}_a=\hat\Pi_{a,0}^{(A)}.
\label{eq:Aproj-relations}
\end{equation}
Set
\begin{equation}
	\lambda_{a,j}
	:=
	1+\left(e^{\beta J_a/M}-1\right)\delta_N(j).
\end{equation}
The electric twist operator is
\begin{equation}
	\hat U_e(\alpha)
	:=
	\prod_{a\in C_{P-1}(\Lambda)}\hat U_{e,a}(\alpha_a),
	\qquad
	\hat U_{e,a}(\alpha_a)
	:=
	\sum_{j\in\mathbb Z_N}
	\frac{\lambda_{a,j+\alpha_a}}{\lambda_{a,j}}\,
	\hat\Pi_{a,j}^{(A)} .
	\label{eq:electric-twist-operator}
\end{equation}
Its effect on the Trotter kernel is computed in Appendix~\ref{app:electric-twist}. It
multiplies the \(i\)-th slab by
\begin{equation}
\prod_{a\in C_{P-1}(\Lambda)}\omega^{\alpha_a n(i)_a}.
\end{equation}
Equivalently, after packaging into \(\phi\in\Omega^P(\underline\Lambda)\), the product over
all slabs is the spacetime Wilson factor
\begin{equation}
	\chi\left(\int_{\Sus(\alpha)}\phi\right),
	\qquad
	\Sus(\alpha)
	:=
	\sum_{i=0}^{M-1}\underline\alpha(i)\in Z_P(\underline\Lambda),
	\qquad
	\underline\alpha(i)=\alpha\times[i,i+1].
	\label{eq:electric-suspension}
\end{equation}
The total electric spacetime cycle is therefore
\begin{equation}
	q_e(\nu,\alpha)
	:=
	\nu(0)+\Sus(\alpha)
	\in Z_P(\underline\Lambda).
	\label{eq:qe}
\end{equation}

\subsubsection{Magnetic twist.}
Let
\begin{equation}
\beta^\vee\in Z_{d-P-1}(\Lambda^\vee),
\qquad
\beta_c:=(\vartheta_\Lambda(\beta^\vee))_c,
\qquad
c\in C_{P+1}(\Lambda).
\end{equation}
For each \(c\in C_{P+1}(\Lambda)\), define the spectral projectors of \(\hat B_c\) by
\begin{equation}
	\hat\Pi_{c,j}^{(B)}
	:=
	\frac1N\sum_{m=0}^{N-1}\omega^{-jm}\hat B_c^{\,m},
	\qquad
	j\in\mathbb Z_N.
\end{equation}
They satisfy
\begin{equation}
\hat B_c\hat\Pi_{c,j}^{(B)}
=
\omega^j\hat\Pi_{c,j}^{(B)},
\qquad
\sum_{j\in\mathbb Z_N}\hat\Pi_{c,j}^{(B)}=\hat 1,
\qquad
\hat{\mathcal B}_c=\hat\Pi_{c,0}^{(B)}.
\end{equation}
Set
\begin{equation}
	\rho_{c,j}
	:=
	1+\left(e^{\beta K_c/M}-1\right)\delta_N(j).
\end{equation}
The magnetic twist operator is
\begin{equation}
	\hat U_m(\beta^\vee)
	:=
	\prod_{c\in C_{P+1}(\Lambda)}\hat U_{m,c}(\beta_c),
	\qquad
	\hat U_{m,c}(\beta_c)
	:=
	\sum_{j\in\mathbb Z_N}
	\frac{\rho_{c,j+\beta_c}}{\rho_{c,j}}\,
	\hat\Pi_{c,j}^{(B)} .
	\label{eq:magnetic-twist-operator}
\end{equation}
As shown in Appendix~\ref{app:magnetic-twist}, this shifts the horizontal curvature
argument by \(\vartheta_\Lambda(\beta^\vee)\). In spacetime terms, it produces the cocycle
\(\vartheta_{\underline\Lambda}(\Sus(\beta^\vee))\), where
\begin{equation}
	\Sus(\beta^\vee)
	\in Z_{d-P}(\underline\Lambda^\vee)
	\label{eq:magnetic-suspension}
\end{equation}
is characterized by
\begin{equation}
	\left(\vartheta_{\underline\Lambda}(\Sus(\beta^\vee))\right)_{c(i)}
	=
	(\vartheta_\Lambda(\beta^\vee))_c,
	\qquad
	\left(\vartheta_{\underline\Lambda}(\Sus(\beta^\vee))\right)_{\underline b(i)}
	=
	0.
	\label{eq:eta-hor-beta-components}
\end{equation}
Thus the total magnetic spacetime cycle is
\begin{equation}
	q_m(\mu^\vee,\beta^\vee)
	:=
	\mu^\vee(0)+\Sus(\beta^\vee)
	\in Z_{d-P}(\underline\Lambda^\vee).
	\label{eq:qm}
\end{equation}

\subsubsection{General background partition function}
\label{sec:general-background-partition}

Define the fully decorated Trotterized trace by
\begin{equation}
	\begin{aligned}
		Z_M(\mu^\vee,\nu;\beta^\vee,\alpha)
		:=
		\tr_{\mathcal H}\biggl[
		&\hat{\mathcal W}_\nu\,
		\hat U_m(\beta^\vee)e^{-\beta\hat H_z/M}\,
		\hat{\mathcal T}_{\mu^\vee}\hat U_e(\alpha)e^{-\beta\hat H_x/M}
		\\
		&\qquad\times
		\left(
		\hat U_m(\beta^\vee)e^{-\beta\hat H_z/M}\,
		\hat U_e(\alpha)e^{-\beta\hat H_x/M}
		\right)^{M-1}
		\biggr].
	\end{aligned}
	\label{eq:decoratedTr}
\end{equation}

\begin{prop}[General background partition function]
	\label{prop:ZtwistedSpatial}
	With \(q_e(\nu,\alpha)\) and \(q_m(\mu^\vee,\beta^\vee)\) defined by
	\eqref{eq:qe} and \eqref{eq:qm}, the decorated
	trace \eqref{eq:decoratedTr} is exactly
	\begin{equation}
		Z_M(\mu^\vee,\nu;\beta^\vee,\alpha)
		=
		\sum_{\phi\in\Omega^P(\underline\Lambda)}
		\chi\left(\int_{q_e(\nu,\alpha)}\phi\right)
		\prod_{c\in C_{P+1}(\underline\Lambda)}
		W_c\left(
		(D\phi+\vartheta_{\underline\Lambda}(q_m(\mu^\vee,\beta^\vee)))_c
		\right)
		\prod_{u\in C_P(\underline\Lambda)}
		V_u(\phi_u).
		\label{eq:Z-background}
	\end{equation}
\end{prop}

\begin{proof}
	The spatial Wilson insertion and the electric Euclidean twist combine into the single
	character $\chi\left(\int_{q_e(\nu,\alpha)}\phi\right)$.
	The spatial 't Hooft insertion and the magnetic Euclidean twist combine into the single
	cocycle shift $\vartheta_{\underline\Lambda}(q_m(\mu^\vee,\beta^\vee))$.
	The detailed verification of these two statements is the content of
	Appendix~\ref{app:TrTwists}, \ref{app:TrFull}. Substituting the resulting electric factor and
	magnetic cocycle shift into the neutral spacetime model
	\eqref{eq:Z-0} gives \eqref{eq:Z-background}.
\end{proof}

Proposition~\ref{prop:ZtwistedSpatial} is exact at fixed Trotter number \(M\). The physical thermal partition function of the quantum Hamiltonian is recovered only after taking the Trotter limit \(M\to\infty\).

\begin{remark}[The trace in a different basis]
	Instead of computing the Trotterized trace in the $Z$-basis as we do, one can also compute it in the $X$-basis, obtaining a dual description of the classical theory with interaction defined by the transposed differential $D^\mathsf{T}$. For the special case of the $\mathbb Z_2$ toric code on a regular 2d square lattice with spins on edges, and without any decorations, this route was taken in \cite{Timmerman:2022uin}.
\end{remark}

\section{Polymer gas in topological backgrounds}
\label{sec:fixed-charge-polymer-gas}

In this section we rewrite the exact spacetime partition function
\eqref{eq:Z-background} as a gas of extended magnetic and electric defects. The resulting
representation is exact in finite volume and is naturally viewed as an expansion around the
flat-field vacuum. The section has three main outputs.

First, we recast the decorated spacetime partition function in terms of arbitrary magnetic
and electric spacetime cycles. Second, after expanding the local weights, we show that the
sum over the spacetime field forces the local defects to assemble into closed magnetic and
electric defects whose total homology classes are fixed by the chosen background. Third, we
resolve the resulting homology constraints by finite Fourier transform, obtaining a
sectorwise polymer gas indexed by topological charges.

In \S\ref{sec:quantum-classical}, the electric and magnetic backgrounds were produced
from a combination of spatial cycles and cycles wrapping the Euclidean-time direction. Since topologically
\(\underline\Lambda=\Lambda\times S^1\) and
\(\underline\Lambda^\vee=\Lambda^\vee\times S^1\), the K\"unneth decomposition gives
\begin{equation}
	H_P(\underline\Lambda)
	\cong
	H_P(\Lambda)\oplus H_{P-1}(\Lambda),
	\qquad
	H_{d-P}(\underline\Lambda^\vee)
	\cong
	H_{d-P}(\Lambda^\vee)\oplus H_{d-P-1}(\Lambda^\vee).
	\label{eq:kunneth}
\end{equation}
The cycles $q_e(\nu,\alpha)$ \eqref{eq:qe} and $q_m(\mu^\vee,\beta^\vee)$ \eqref{eq:qm} are specific representatives of the homology classes of the primal and the dual lattice compatible with the above decomposition. While working with the classical theory, it will be more convenient to work with arbitrary spacetime representatives
\begin{equation}
	q_m\in Z_{d-P}(\underline\Lambda^\vee),
	\qquad
	q_e\in Z_P(\underline\Lambda).
\end{equation}
We therefore define the fixed-charge partition function by
\begin{equation}
	Z_M(q_m,q_e)
	:=
	\sum_{\phi\in \Omega^P(\underline\Lambda)}
	\chi \left(\int_{q_e}\phi\right)
	\prod_{c\in C_{P+1}(\underline\Lambda)}
	W_c \left((D\phi+\vartheta_{\underline\Lambda}(q_m))_c\right)
	\prod_{u\in C_P(\underline\Lambda)}
	V_u(\phi_u),
	\label{eq:Zme}
\end{equation}
where \(q_e\) is a primal \(P\)-cycle, \(q_m\) is a dual \((d-P)\)-cycle, and
\(\vartheta_{\underline\Lambda}(q_m)\in Z^{P+1}(\underline\Lambda)\) (see \eqref{eq:vartheta}) is the corresponding
primal cocycle.

\subsection{Code limit and the flat vacuum}
\label{sec:fixed-charge-code-limit}

For a perturbative expansion of Trotter weights \eqref{eq:parallel-weights}, \eqref{eq:perp-weights}, it is convenient at fixed
\(M\) to introduce
\begin{equation}
	\epsilon_g
	:=
	\max_{b\in C_P(\Lambda)}
	\sum_{n\in \mathbb Z_N}\left|\frac{\beta g_b^{(n)}}{M}\right|,
	\qquad
	\epsilon_h
	:=
	\max_{b\in C_P(\Lambda)}
	\sum_{n\in \mathbb Z_N}\left|\frac{\beta h_b^{(n)}}{M}\right|,
\end{equation}
and
\begin{equation}
	\kappa_J
	:=
	\min_{a\in C_{P-1}(\Lambda)}\frac{\beta J_a}{M},
	\qquad
	\kappa_K
	:=
	\min_{c\in C_{P+1}(\Lambda)}\frac{\beta K_c}{M}.
\end{equation}
The \emph{code limit} is the regime
\begin{equation}
	\epsilon_g\to 0,
	\qquad
	\epsilon_h\to 0,
	\qquad
	\kappa_J\to \infty,
	\qquad
	\kappa_K\to \infty.
	\label{eq:codeLimit}
\end{equation}

At the level of the explicit local weights \eqref{eq:parallel-weights}, \eqref{eq:perp-weights},
this means that, after dividing out the trivial local zero-mode factors,
\begin{equation}
	\frac{W_c(x)}{W_c(0)}\rightarrow \delta_N(x),
	\qquad
	\frac{V_u(x)}{V_u(0)}\rightarrow 1,
	\label{eq:codeLimitWV}
\end{equation}
uniformly in \(x\in\mathbb Z_N\). Accordingly, define
\begin{equation}
	\widetilde Z_M(q_m,q_e)
	:=
	\frac{Z_M(q_m,q_e)}
	{\prod_{c\in C_{P+1}(\underline\Lambda)}W_c(0)\,
		\prod_{u\in C_P(\underline\Lambda)}V_u(0)}.
	\label{eq:ZmeNormal}
\end{equation}
Then in the code limit
\begin{equation}
	\widetilde Z_M(q_m,q_e)
	\rightarrow
	\sum_{\phi\in \Omega^P(\underline\Lambda)}
	\chi \left(\int_{q_e}\phi\right)
	\prod_{c\in C_{P+1}(\underline\Lambda)}
	\delta_N \left((D\phi+\vartheta_{\underline\Lambda}(q_m))_c\right).
	\label{eq:codeLimitZme}
\end{equation}

In the neutral magnetic sector \(q_m=0\), the limiting theory is therefore supported on the
flat sector
\begin{equation}
	D\phi=0.
\end{equation}
This is the \emph{flat vacuum} around which the polymer expansion will be organized. The defect variables
introduced below measure departures from this flat-field configuration, and in the neutral
sector the dominant defect configuration is the empty one.

\subsection{Local defect expansion}

We now expand \eqref{eq:Zme} around the flat vacuum by separating the
local zero modes from the nontrivial local excitations.

For each \((P+1)\)-cell \(c\in C_{P+1}(\underline\Lambda)\), define
\begin{equation}
	\varpi_c(x)
	:=
	\frac{W_c(x)}{W_c(0)},
	\qquad
	x\in \mathbb Z_N.
	\label{eq:magnetic-local-amplitude}
\end{equation}
Then \(\varpi_c(0)=1\), and
\begin{equation}
	W_c(x)
	=
	W_c(0)\sum_{m\in \mathbb Z_N}\varpi_c(m)\,\delta_N(x-m).
	\label{eq:W-cell-expansion}
\end{equation}
A nonzero \(m\) will be interpreted as a local magnetic defect on the cell \(c\).

For each \(P\)-cell \(u\in C_P(\underline\Lambda)\), let
\begin{equation}
	\widehat V_u(e)
	:=
	\frac{1}{\sqrt N}
	\sum_{x\in \mathbb Z_N}
	V_u(x)\,\chi_e(-x),
	\qquad
	e\in \widehat{\mathbb Z}_N,
	\label{eq:Vu-fourier-transform}
\end{equation}
be the discrete Fourier transform in the convention of \eqref{eq:invFT}. Define the normalized electric Fourier
amplitudes by
\begin{equation}
	\upsilon_u(e)
	:=
	\frac{\widehat V_u(-e)}{\widehat V_u(0)},
	\qquad
	e\in \mathbb Z_N.
	\label{eq:electric-local-amplitude}
\end{equation}
Then \(\upsilon_u(0)=1\), and Fourier inversion gives
\begin{equation}
	V_u(x)
	=
	\frac{\widehat V_u(0)}{\sqrt N}
	\sum_{e\in \mathbb Z_N}
	\upsilon_u(e)\,\chi_e(-x).
	\label{eq:V-cell-expansion}
\end{equation}
A nonzero \(e\) will be interpreted as a local electric defect label on the cell \(u\).

In the code limit \eqref{eq:codeLimit}, the nontrivial amplitudes are
small:
\begin{equation}
	\varpi_c(m)\rightarrow 0 \quad (m\neq 0),
	\qquad
	\upsilon_u(e)\rightarrow 0 \quad (e\neq 0).
	\label{eq:nontrivial-local-amplitudes-small}
\end{equation}
This is the precise sense in which the polymer gas below is a perturbation series around
the flat vacuum.

Introduce magnetic and electric defect chains
\begin{equation}
	\mathfrak m^\vee
	=
	\sum_{c^\vee\in C_{d-P}(\underline\Lambda^\vee)}
	\mathfrak m_{c^\vee}^\vee\,c^\vee
	\in \mathcal C_{d-P}(\underline\Lambda^\vee),
	\qquad
	\mathfrak e
	=
	\sum_{u\in C_P(\underline\Lambda)}
	\mathfrak e_u\,u
	\in \mathcal C_P(\underline\Lambda).
\end{equation}
Let
\begin{equation}
	\mathfrak m
	:=
	\vartheta_{\underline\Lambda}(\mathfrak m^\vee)\in \Omega^{P+1}(\underline\Lambda),
	\qquad
	\mathfrak m_c
	=
	\mathfrak m^\vee_{\theta_{\underline\Lambda}^{-1}(c)}
	\quad
	(c\in C_{P+1}(\underline\Lambda)),
\end{equation}
and define the corresponding defect weights by
\begin{equation}
	\wt_m(\mathfrak m^\vee)
	:=
	\prod_{c\in C_{P+1}(\underline\Lambda)}\varpi_c(\mathfrak m_c),
	\qquad
	\wt_e(\mathfrak e)
	:=
	\prod_{u\in C_P(\underline\Lambda)}\upsilon_u(\mathfrak e_u).
	\label{eq:local-defect-weights}
\end{equation}
Finally, set
\begin{equation}
	\mathcal N_M
	:=
	\prod_{c\in C_{P+1}(\underline\Lambda)}W_c(0)\,
	\prod_{u\in C_P(\underline\Lambda)}\frac{\widehat V_u(0)}{\sqrt N}.
	\label{eq:NM-fixed-charge}
\end{equation}

\begin{prop}[Exact local-defect expansion]
	\label{prop:exact-local-defect-expansion}
	For every magnetic background \(q_m\in Z_{d-P}(\underline\Lambda^\vee)\) and electric
	background \(q_e\in Z_P(\underline\Lambda)\), the partition function \eqref{eq:Zme} is given by
	\begin{equation}
		Z_M(q_m,q_e)
		=
		\mathcal N_M
		\sum_{\mathfrak m^\vee\in \mathcal C_{d-P}(\underline\Lambda^\vee)}
		\sum_{\mathfrak e\in \mathcal C_P(\underline\Lambda)}
		\wt_m(\mathfrak m^\vee)\,\wt_e(\mathfrak e)\,
		\Xi(\mathfrak m^\vee,\mathfrak e;q_m,q_e),
		\label{eq:exact-local-defect-expansion}
	\end{equation}
	where
	\begin{equation}
		\Xi(\mathfrak m^\vee,\mathfrak e;q_m,q_e)
		:=
		\sum_{\phi\in \Omega^P(\underline\Lambda)}
		\chi \left(\int_{q_e-\mathfrak e}\phi\right)
		\prod_{c\in C_{P+1}(\underline\Lambda)}
		\delta_N \left((D\phi+\vartheta_{\underline\Lambda}(q_m-\mathfrak m^\vee))_c\right).
		\label{eq:Xi-fixed-charge}
	\end{equation}
\end{prop}

\begin{proof}
	Insert \eqref{eq:W-cell-expansion} and \eqref{eq:V-cell-expansion} into
	\eqref{eq:Zme}. The factors \(W_c(0)\) and
	\(\widehat V_u(0)/\sqrt N\) assemble into \(\mathcal N_M\), while the remaining
	cellwise products produce the defect weights \eqref{eq:local-defect-weights}. The
	magnetic defect variables are naturally assembled into the dual chain
	\(\mathfrak m^\vee\), whose primal image is the cochain \(\mathfrak m\). The electric
	Fourier labels assemble into the primal chain \(\mathfrak e\). Reordering the finite sums
	then gives \eqref{eq:exact-local-defect-expansion} with compatibility factor
	\eqref{eq:Xi-fixed-charge}.
\end{proof}

Equation \eqref{eq:Xi-fixed-charge} exhibits the basic structure of the expansion. The
chains \(\mathfrak m^\vee\) and \(\mathfrak e\) are local magnetic and electric defect
fields, while \(\Xi(\mathfrak m^\vee,\mathfrak e;q_m,q_e)\) is the compatibility factor
obtained by summing over all spacetime \(P\)-cochains consistent with those local defects in
the prescribed background.

\subsection{Closed defect expansion and homology constraints}
\label{sec:closedDefectGas}

We now determine when the compatibility factor
\(\Xi(\mathfrak m^\vee,\mathfrak e;q_m,q_e)\) is nonzero.

Suppose first that the delta constraints in \eqref{eq:Xi-fixed-charge} are solvable. Then
there exists \(\phi\in \Omega^P(\underline\Lambda)\) such that
\begin{equation}
	D\phi=\vartheta_{\underline\Lambda}(\mathfrak m^\vee-q_m).
	\label{eq:Dphi-background-m}
\end{equation}
Since the left-hand side is a coboundary and $\vartheta_{\underline\Lambda}$ is an isomorphism (Corr.~\ref{cor:complexIso}), it follows that
\begin{equation}
	\mathfrak m^\vee\in Z_{d-P}(\underline\Lambda^\vee),
	\qquad
	[\mathfrak m^\vee]=[q_m]\in H_{d-P}(\underline\Lambda^\vee),
	\label{eq:magnetic-homology-class}
\end{equation}
because \eqref{eq:Dphi-background-m} is precisely the statement that
\(\mathfrak m^\vee-q_m\in B_{d-P}(\underline\Lambda^\vee)\).

Next fix one solution \(\lambda\) of \eqref{eq:Dphi-background-m}. Every other solution is
of the form \(\lambda+\psi\) with \(\psi\in Z^P(\underline\Lambda)\). Therefore
\begin{equation}
	\Xi(\mathfrak m^\vee,\mathfrak e;q_m,q_e)
	=
	\chi \left(\int_{q_e-\mathfrak e}\lambda\right)
	\sum_{\psi\in Z^P(\underline\Lambda)}
	\chi \left(\int_{q_e-\mathfrak e}\psi\right).
	\label{eq:Xi-affine-flat-sum}
\end{equation}
By character orthogonality on the finite group of cocycles, this vanishes unless
\(q_e-\mathfrak e\) pairs trivially with every cocycle, that is, unless
\begin{equation}
	q_e-\mathfrak e\in B_P(\underline\Lambda).
\end{equation}
Equivalently,
\begin{equation}
	\mathfrak e\in Z_P(\underline\Lambda),
	\qquad
	[\mathfrak e]=[q_e]\in H_P(\underline\Lambda).
	\label{eq:electric-homology-class}
\end{equation}

Thus only closed magnetic and electric defects survive, and their total homology classes are
forced to agree with those of the prescribed background. Since $q_e-\mathfrak e$ is a boundary, the character $\chi\left(\int_{q_e-\mathfrak e}\psi\right)$ with closed $\psi$ evaluates to $1$ and the sum over $\psi$ contributes the factor \(|Z^P(\underline\Lambda)|\). It is therefore convenient to define the normalized partition function
\begin{equation}
	\mathcal Z_M(q_m,q_e)
	:=
	\frac{Z_M(q_m,q_e)}
	{\mathcal N_M\,|Z^P(\underline\Lambda)|}.
	\label{eq:normalized-closed-defect-gas}
\end{equation}

The remaining phase is determined by the two boundary differences
\begin{equation}
\delta_m:=\mathfrak m^\vee-q_m\in B_{d-P}(\underline\Lambda^\vee),
\qquad
\delta_e:=q_e-\mathfrak e\in B_P(\underline\Lambda).
\end{equation}
Let \(\Sigma\in \mathcal C_{P+1}(\underline\Lambda)\)
such that \(\partial\Sigma=\delta_e\), then the Stokes formula gives\footnote{In fact, we define the lattice differential to satisfy this condition, see \eqref{eq:d-def} for the spatial lattice, applies more generally.}
\begin{equation}
\int_{\delta_e}\lambda
=
\int_{\partial\Sigma}\lambda
=
\int_\Sigma D\lambda
=
I_{\underline\Lambda}(\delta_m,\Sigma).
\end{equation}
This is the first term in the definition of the linking number between $\delta_m$ and $\delta_e$ \eqref{eq:generalized-linking}, and since both of these are boundaries, the remaining two terms in \eqref{eq:generalized-linking} vanish. Hence
\begin{equation}
	\chi \left(\int_{q_e-\mathfrak e}\lambda\right)
	=
	\omega^{
		\Lk_{\underline\Lambda}
		(\mathfrak m^\vee-q_m,\;q_e-\mathfrak e)
	}.
	\label{eq:Xi-linking-phase}
\end{equation}

\begin{prop}[Closed-defect gas at fixed background]
	\label{prop:closed-defect-gas-fixed-charge}
	The normalized partition function \eqref{eq:normalized-closed-defect-gas} is given by
	\begin{equation}
		\mathcal Z_M(q_m,q_e)
		=
		\sum_{\substack{\mathfrak m^\vee\in Z_{d-P}(\underline\Lambda^\vee)\\
				[\mathfrak m^\vee]=[q_m]}}
		\sum_{\substack{\mathfrak e\in Z_P(\underline\Lambda)\\
				[\mathfrak e]=[q_e]}}
		\wt_m(\mathfrak m^\vee)\,
		\wt_e(\mathfrak e)\,
		\omega^{-
			\Lk_{\underline\Lambda}
			(\mathfrak m^\vee-q_m,\mathfrak e-q_e)
		}.
		\label{eq:closed-defect-gas-fixed-charge}
	\end{equation}
\end{prop}

\begin{proof}
	The preceding discussion shows that
	\(\Xi(\mathfrak m^\vee,\mathfrak e;q_m,q_e)\) vanishes unless $[\mathfrak m^\vee]=[q_m]$ and $[\mathfrak e]=[q_e]$.
	In the nonvanishing case, the cocycle sum in
	\eqref{eq:Xi-affine-flat-sum} contributes
	\(|Z^P(\underline\Lambda)|\), and the remaining character is \eqref{eq:Xi-linking-phase}. Substituting this character and $|Z^P(\underline\Lambda)|$ into \eqref{eq:exact-local-defect-expansion} and normalizing as in
	\eqref{eq:normalized-closed-defect-gas} gives \eqref{eq:closed-defect-gas-fixed-charge}.
\end{proof}
The above proposition holds for every integer \(N\ge 2\) since the phase in
\eqref{eq:closed-defect-gas-fixed-charge} involves only the boundary
differences so that it is the ordinary canonical linking number of two boundaries and does not require \(N\) to be prime (Prop.~\ref{prop:choice-dependence}\ref{propItem:bdryLink}). In the rest of \S\ref{sec:fixed-charge-polymer-gas} we will use a chosen bilinear extension of
that boundary-linking to arbitrary cycles, first in the bilinear expansion
\eqref{eq:LbiExpansion} and then in the Fourier resolution of the homology
constraints. For those steps we assume \(N\) is prime, so that
\(\mathbb Z_N=\mathbb F_N\) and the homology groups are finite-dimensional
\(\mathbb F_N\)-vector spaces.

\subsection{Connected polymer expansion}
\label{sec:fixed-charge-connected-polymers}

A nonzero closed defect need not be connected. We therefore decompose each surviving
magnetic and electric defect into its support-connected components:
\begin{equation}
	\mathfrak m^\vee=M_1+\cdots+M_r,
	\qquad
	\mathfrak e=E_1+\cdots+E_s,
	\label{eq:connected-defect-decomposition}
\end{equation}
where the supports of the \(M_i\) are pairwise disjoint, the supports of the \(E_j\) are
pairwise disjoint, and
\begin{equation}
	M_i\in Z_{d-P}(\underline\Lambda^\vee),
	\qquad
	E_j\in Z_P(\underline\Lambda).
\end{equation}

Define
\begin{equation}
\begin{aligned}
	\mathbb M 
	:=&\;
	\{\text{connected }(d-P)\text{-cycles in }\underline\Lambda^\vee\},
	\\
	\mathbb E 
	:=&\;
	\{\text{connected }P\text{-cycles in }\underline\Lambda\}.
\end{aligned}
	\label{eq:polymer-sets-fixed-charge}
\end{equation}
We refer to elements of \(\mathbb M \) as magnetic polymers and to elements of
\(\mathbb E \) as electric polymers.

Because the defect weights factor over cells, they factor over support-disjoint unions.
Thus there exist one-polymer activities
\begin{equation}
	\rho_m:\mathbb M \to \mathbb C,
	\qquad
	\rho_e:\mathbb E \to \mathbb C
\end{equation}
such that
\begin{equation}
	\wt_m(\mathfrak m^\vee)=\prod_{i=1}^r \rho_m(M_i),
	\qquad
	\wt_e(\mathfrak e)=\prod_{j=1}^s \rho_e(E_j).
	\label{eq:one-polymer-factorization}
\end{equation}
Explicitly,
\begin{equation}
	\rho_m(M)
	=
	\prod_{c\in C_{P+1}(\underline\Lambda)}
	\varpi_c \left((\vartheta_{\underline\Lambda}(M))_c\right),
	\qquad
	\rho_e(E)
	=
	\prod_{u\in C_P(\underline\Lambda)}\upsilon_u(E_u).
	\label{eq:one-polymer-activities-fixed-charge}
\end{equation}

Introduce the homology groups
\begin{equation}
	H_m:=H_{d-P}(\underline\Lambda^\vee),
	\qquad
	H_e:=H_P(\underline\Lambda).
	\label{eq:fixed-charge-homology-groups}
\end{equation}
The groups $H_m$ and $H_e$ label the topological classes of the magnetic and electric background charges respectively, i.e., they are the groups in which the homology classes \([q_m]\) and \([q_e]\) take values.

Let $\delta_H(h)$ be the Kronecker delta on any abelian group $H$. Substituting \eqref{eq:connected-defect-decomposition} and
\eqref{eq:one-polymer-factorization} into
\eqref{eq:closed-defect-gas-fixed-charge} gives the fixed-background polymer gas:
\begin{equation}
\begin{aligned}
	\mathcal Z_M(q_m,q_e)
	&=
	\sum_{r,s\ge 0}
	\frac{1}{r!\,s!}
	\sum_{\substack{M_1,\dots,M_r\in \mathbb M \\
			\text{pairwise disjoint}}}
	\sum_{\substack{E_1,\dots,E_s\in \mathbb E \\
			\text{pairwise disjoint}}}
	\Bigl(\prod_{i=1}^r \rho_m(M_i)\Bigr)
	\Bigl(\prod_{j=1}^s \rho_e(E_j)\Bigr)
	\\
	&\quad\times
	\delta_{H_m} \Bigl(\sum_{i=1}^r [M_i]-[q_m]\Bigr)\,
	\delta_{H_e} \Bigl(\sum_{j=1}^s [E_j]-[q_e]\Bigr)\,
	\omega^{-
		\Lk_{\underline\Lambda}
		\left(
		\sum_{i=1}^r M_i-q_m,\;
		\sum_{j=1}^s E_j-q_e
		\right)
	}.
\end{aligned}
\label{eq:fixed-charge-polymer-gas}
\end{equation}

This is an exact perturbation series around the code limit. The local defect amplitudes
\(\varpi_c(m)\) and \(\upsilon_u(e)\) measure departures from the flat vacuum; the
\(\phi\)-sum forces these local excitations to assemble into closed extended objects; and
their connected components are the polymers. The background data enters in two distinct
ways. First, it imposes the homology constraints
\begin{equation}
\sum_i[M_i]=[q_m],
\qquad
\sum_j[E_j]=[q_e].
\end{equation}
Second, after these constraints hold, the surviving character is the linking phase of the
boundary differences
\begin{equation}
\sum_i M_i-q_m\in B_{d-P}(\underline\Lambda^\vee),
\qquad
\sum_j E_j-q_e\in B_P(\underline\Lambda).
\end{equation}
Equivalently, by bilinearity,
\begin{equation}
\begin{aligned}
	&-\Lk_{\underline\Lambda}
	\left(
	\sum_i M_i-q_m,\;
	\sum_j E_j-q_e
	\right)
	\\
	&\qquad
	=
	-\sum_{i,j} \Lk_{\underline\Lambda}
	\left(
	M_i,E_j
	\right)
	+
	\sum_i \Lk_{\underline\Lambda}
	\left(
	M_i,q_e
	\right)
	+
	\sum_j \Lk_{\underline\Lambda}
	\left(
	q_m,E_j
	\right)
	-
	\Lk_{\underline\Lambda}(q_m,q_e).
\end{aligned}
\label{eq:LbiExpansion}
\end{equation}
Thus the exact phase contains the polymer-polymer interaction, the interaction of polymers
with the chosen background representatives, and a background-background constant.\footnote{Note that the individual terms on the right-hand side of
	\eqref{eq:LbiExpansion} involve the chosen generalized linking pairing on
	arbitrary cycles and therefore depend on the auxiliary conventions fixed in
	Appendix~\ref{app:linking}. Only the total combination, which is the linking of
	the two boundary differences appearing on the left-hand side, is canonical.} In the
neutral sector, the empty polymer configuration is the vacuum term. In a nontrivial
fixed-background sector, the leading contribution comes from the smallest polymer
configurations carrying the required homology.

\subsection{Fourier resolution of the homology constraints}
\label{sec:fixed-charge-fourier-resolution}
We can resolve the delta constraint in the partition function \eqref{eq:fixed-charge-polymer-gas} and write it as a sum over two-species polymer gas partition functions. Each two-species partition function will be labeled by the \emph{topological sector} data $(\ell_m,\ell_e) \in \widehat H_m \times \widehat H_e$ where $\widehat H_m$ and $\widehat H_e$ are the linear duals of the $\mathbb F_N$-vector spaces $H_m$ and $H_e$ respectively. Given the background charges $(q_m,q_e) \in Z_{d-P}(\underline\Lambda^\vee) \times Z_P(\underline\Lambda)$ and for any $(\ell_m,\ell_e) \in \widehat H_m \times \widehat H_e$ define the 1-polymer activity
\begin{equation}
	\rho^{\ell_m,\ell_e;q_m,q_e}(A)
	:=
	\begin{cases}
		\rho_m(A)\,
		\omega^{\ell_m([A])+\Lk_{\underline\Lambda}(M,q_e)},
		& A\in \mathbb M ,
		\\
		\rho_e(A)\,
		\omega^{\ell_e([A])+\Lk_{\underline\Lambda}(q_m,E)},
		& A\in \mathbb E .
	\end{cases}
	\label{eq:fixed-background-unified-polymer-activity}
\end{equation}
This is a modification of the activity in \eqref{eq:fixed-charge-polymer-gas} by the $(\ell_m,\ell_e)$-dependent phase.

Define an interaction between electric and magnetic polymer pairs as follows. Same-species polymers have hard-core exclusion:
two magnetic polymers are compatible only when their supports are disjoint, and similarly
for two electric polymers. Opposite-species polymers have no hard-core exclusion; instead,
a magnetic polymer \(M\) and an electric polymer \(E\) interact through the linking phase
\(\omega^{-\Lk_{\underline\Lambda}(M,E)}\). Equivalently, set
\begin{equation}
	\zeta(A,B)
	:=
	\begin{cases}
		-1,
		& A,B \text{ same species with overlapping supports},
		\\
		0,
		& A,B \text{ same species with disjoint supports},
		\\
		\omega^{-\Lk_{\underline\Lambda}(A,B)}-1,
		& A\in \mathbb M ,\ B\in \mathbb E ,
		\\
		\omega^{-\Lk_{\underline\Lambda}(B,A)}-1,
		& A\in \mathbb E ,\ B\in \mathbb M .
	\end{cases}
	\label{eq:zeta}
\end{equation}
Then
\begin{equation}
	1+\zeta(M,E)=\omega^{-\Lk_{\underline\Lambda}(M,E)}
	\label{eq:MEphase}
\end{equation}
for opposite-species pairs, while \(1+\zeta(A,B)\) enforces hard-core exclusion for
overlapping same-species pairs.

\begin{prop}[Two-species polymer gas]
	\label{prop:fourier-resolved-polymer-gas}
	The normalized fixed-background partition function \eqref{eq:fixed-charge-polymer-gas} decomposes as a sum over magnetic and electric topological sectors
	\begin{equation}
		\mathcal Z_M(q_m,q_e)
		=
		\frac{1}{|\widehat H_m|\,|\widehat H_e|}
		\sum_{\ell_m\in \widehat H_m}
		\sum_{\ell_e\in \widehat H_e}
		\omega^{-\ell_m([q_m])-\ell_e([q_e])}\,
		\widehat{\mathcal Z}_M(\ell_m,\ell_e;q_m,q_e).
		\label{eq:Zme-deltaResolved}
	\end{equation}
	For each fixed \((\ell_m,\ell_e;q_m,q_e)\), the sector amplitude is the abstract
	two-species polymer partition function
	\begin{equation}
		\widehat{\mathcal Z}_M(\ell_m,\ell_e;q_m,q_e)
		=
		\omega^{-\Lk_{\underline\Lambda}(q_m,q_e)}
		\sum_{n\ge 0}
		\frac{1}{n!}
		\sum_{A_1,\dots,A_n\in\mathbb A}
		\prod_{k=1}^n \rho^{\ell_m,\ell_e;q_m,q_e}(A_k)
		\prod_{1\le k<l\le n}\left(1+\zeta(A_k,A_l)\right).
		\label{eq:Zlmle}
	\end{equation}
\end{prop}

\begin{proof}
	The delta functions in \eqref{eq:fixed-charge-polymer-gas} can be resolved by using
	\begin{equation}
		\delta_H(h) = \frac{1}{|\widehat H|} \sum_{\ell \in \widehat H} \omega^{\ell(h)}, \label{eq:delta-H-fourier}
	\end{equation}
	The resolution holds for any finite $\mathbb F_N$-vector space $H$, and in particular, for $H_m$ and $H_e$. Using \eqref{eq:delta-H-fourier} in
	\eqref{eq:fixed-charge-polymer-gas} produces the factor $|\widehat H_m|^{-1}|\widehat H_e|^{-1}
	\omega^{-\ell_m([q_m])-\ell_e([q_e])}$
	and inserts
	\[
	\prod_i\omega^{\ell_m([M_i])},
	\qquad
	\prod_j\omega^{\ell_e([E_j])}
	\]
	into the magnetic and electric polymer weights. These are the $(\ell_m,\ell_e)$-dependent phases in \eqref{eq:fixed-background-unified-polymer-activity}.
	
	Next expand the linking phase in \eqref{eq:fixed-charge-polymer-gas} by bilinearity, as in \eqref{eq:LbiExpansion}. The polymer-background terms are precisely those absorbed into \eqref{eq:fixed-background-unified-polymer-activity}. The background-background term gives
	the overall phase
	\(\omega^{-\Lk_{\underline\Lambda}(q_m,q_e)}\). The polymer-polymer terms
	give the opposite-species pair factors encoded by \eqref{eq:MEphase}.
	
	It remains only to pass from the separated magnetic/electric sum to the abstract
	one-set notation. Same-species overlap is exactly the hard-core exclusion encoded by the
	first two cases of \eqref{eq:zeta}, while opposite-species pairs contribute the linking phase
	\eqref{eq:MEphase}. Finally, summing over ordered \(n\)-tuples in
	\(\mathbb A=\mathbb M \sqcup \mathbb E\) with the factor
	\(1/n!\) is equivalent to summing over \(r\) magnetic and \(s\) electric polymers with
	the factor \(1/(r!\,s!)\), since
	\[
	\frac{1}{n!}\binom{n}{r}
	=
	\frac{1}{r!\,s!}.
	\]
	This proves \eqref{eq:Zme-deltaResolved} and \eqref{eq:Zlmle}.
\end{proof}

Thus the Fourier transform over homology separates the fixed-background polymer gas into
sector amplitudes labelled by \((\ell_m,\ell_e)\). In each topological sector, the dependence on the
background representatives is split into a constant background-background phase and
background-twisted one-polymer activities, while the genuine two-body interaction is the
same electric-magnetic linking phase as in the neutral theory.

\section{A rigorous low-activity region}
\label{sec:low-activity}

The purpose of this section is to obtain rigorous control of the background-dependent
fixed-charge partition functions
\(\mathcal Z_M(q_m,q_e)\) in the low-activity regime. The key idea is to compare the
absolute value of the complex polymer expansion to positive same-species hard-core polymer
gases and then apply abstract polymer criteria due to Koteck\'y-Preiss,
Fern\'andez-Procacci, Bissacot-Fern\'andez-Procacci, and related cluster-expansion
methods
\cite{KoteckyPreiss1986,FernandezProcacci2007,BissacotFernandezProcacci2010,Ueltschi2004}.

When the Fourier-resolved representation \eqref{eq:Zlmle} is available, the corresponding sector amplitudes \(\widehat{\mathcal Z}_M(\ell_m,\ell_e;q_m,q_e)\) have generally complex one-polymer activities \eqref{eq:fixed-background-unified-polymer-activity}, and opposite-species polymers interact through the pure phase \eqref{eq:MEphase}. In our present formulation this sector decomposition is available for prime \(N\), and in that case we first bound the sector amplitudes and then deduce bounds for the background-dependent partition functions \(\mathcal Z_M(q_m,q_e)\). However, as explained in Remark~\ref{rmk:directZmeBound}, the same fixed-background partition functions can also be bounded directly, without passing through the sector amplitudes. The outcome is a background-uniform low-activity region in which the partition functions are absolutely controlled uniformly in finite volume, large connected polymers are exponentially suppressed, and homologically nontrivial polymers are exponentially suppressed on the scale of the corresponding spacetime systole, i.e.\ the minimal support size of a nontrivial homology cycle. The main limitation in the resulting region comes not from the abstract polymer criterion itself, but from the available combinatorial bound on the number of rooted connected closed polymers.

The prime-\(N\) assumption enters this section only through the sector
decomposition \eqref{eq:Zlmle}. The positive majorant gases and the resulting low-activity
bounds for the fixed-background partition functions \(\mathcal Z_M(q_m,q_e)\) apply for
every integer \(N\ge 2\).

\subsection{Sectorwise gas and absolute majorization}
\label{sec:low-activity-abstract-polymer}

Fix background representatives
\begin{equation}
q_m\in Z_{d-P}(\underline\Lambda^\vee),
\qquad
q_e\in Z_P(\underline\Lambda),
\end{equation}
and topological sector charges
\begin{equation}
(\ell_m,\ell_e)\in \widehat H_m\times \widehat H_e.
\end{equation}
We use the two-species polymer representation
\eqref{eq:Zlmle}. Thus the polymer set is
$\mathbb A=\mathbb M \sqcup\mathbb E $,
the background-twisted one-polymer activity is
\(\rho^{\ell_m,\ell_e;q_m,q_e}\), and the pair interaction is \(\zeta\) from
\eqref{eq:zeta}.

The important observation is that all background and character dependences in the
one-polymer activities are phases. Hence
\begin{equation}
	\left|\rho^{\ell_m,\ell_e;q_m,q_e}(M)\right|
	=
	|\rho_m(M)|,
	\qquad
	\left|\rho^{\ell_m,\ell_e;q_m,q_e}(E)\right|
	=
	|\rho_e(E)|.
	\label{eq:background-activity-moduli}
\end{equation}
Likewise, the overall background phase
\(\omega^{-\Lk_{\underline\Lambda}(q_m,q_e)}\) has modulus one, and for
opposite-species pairs
\begin{equation}
	|1+\zeta(M,E)|
	=
	\left|\omega^{-\Lk_{\underline\Lambda}(M,E)}\right|
	=
	1.
	\label{eq:opposite-phase-modulus}
\end{equation}
For same-species pairs, \(1+\zeta(A,B)\) is either \(0\) or \(1\), depending on whether the
two polymers overlap or are support-disjoint. Thus taking absolute values removes all
electric-magnetic phases and retains only same-species hard-core exclusion. This allows us to use the Fern\'andez-Procacci criterion \cite{FernandezProcacci2007} for a hard-core polymer gas. Instead of applying this criterion to the complex polymer gas, we use the fact that the absolute value of the two-species expansion is bounded by a product of two positive one-species hard-core gases.

For \(s\in\{m,e\}\), write
\begin{equation}
\mathbb A_m:=\mathbb M ,
\qquad
\mathbb A_e:=\mathbb E .
\end{equation}
The associated positive hard-core polymer gas of species \(s\) is
\begin{equation}
	\mathcal Z_s^{\mathrm{maj}}
	:=
	\sum_{n\ge 0}
	\frac{1}{n!}
	\sum_{A_1,\dots,A_n\in \mathbb A_s}
	\left(\prod_{k=1}^n |\rho_s(A_k)|\right)
	\prod_{1\le k<l\le n}\mathbf 1\{A_k\sim A_l\},
	\label{eq:majorant-gas}
\end{equation}
where same-species compatibility means support-disjointness.\footnote{Here $\mathbf 1\{A_k\sim A_l\}$ is the compatibility indicator, equal to 1 if $A_k$ and $A_l$ have disjoint supports and 0 otherwise.}

For every fixed \((\ell_m,\ell_e;q_m,q_e)\), the absolute value of the sectorwise
expansion satisfies
\begin{equation}
	\left|\widehat{\mathcal Z}_M(\ell_m,\ell_e;q_m,q_e)\right|
	\le
	\mathcal Z_m^{\mathrm{maj}}\,
	\mathcal Z_e^{\mathrm{maj}}.
	\label{eq:twisted-sector-majorant-bound}
\end{equation}
Indeed, after absolute values are taken, magnetic and electric polymers no longer interact
with one another, while same-species hard-core constraints remain. Consequently, absolute
convergence of the two positive gases \(\mathcal Z_m^{\mathrm{maj}}\) and
\(\mathcal Z_e^{\mathrm{maj}}\) implies absolute convergence, uniformly in
\((\ell_m,\ell_e;q_m,q_e)\), of the sectorwise amplitudes
\(\widehat{\mathcal Z}_M(\ell_m,\ell_e;q_m,q_e)\).

Since the fixed-background amplitudes are recovered by the finite sum \eqref{eq:Zme-deltaResolved}, the same majorant controls
\(\mathcal Z_M(q_m,q_e)\). Explicitly,
\begin{equation}
	|\mathcal Z_M(q_m,q_e)|
	\le
	\frac{1}{|\widehat H_m|\,|\widehat H_e|}
	\sum_{\ell_m,\ell_e}
	\left|\widehat{\mathcal Z}_M(\ell_m,\ell_e;q_m,q_e)\right|
	\le
	\mathcal Z_m^{\mathrm{maj}}\,
	\mathcal Z_e^{\mathrm{maj}}.
	\label{eq:fixed-charge-majorant-bound}
\end{equation}

Whenever \(\mathcal Z_s^{\mathrm{maj}}\) is finite and nonzero, define the pinned
one-polymer amplitude of the positive gas by
\begin{equation}
	\Pi_s^{\mathrm{maj}}(A)
	:=
	\frac{\partial}{\partial |\rho_s(A)|}
	\log \mathcal Z_s^{\mathrm{maj}},
	\qquad
	A\in\mathbb A_s.
	\label{eq:maj-pinned}
\end{equation}
Equivalently,
\begin{equation}
	|\rho_s(A)|\,\Pi_s^{\mathrm{maj}}(A)
	=
	\frac{|\rho_s(A)|}{\mathcal Z_s^{\mathrm{maj}}}
	\frac{\partial}{\partial |\rho_s(A)|} \mathcal Z_s^{\mathrm{maj}}.
	\label{eq:maj-pinned-identity}
\end{equation}
This is the normalized positive weight of configurations containing \(A\). In the original
complex sectorwise gas, the analogous quantity is not a probability. The majorant
quantity \eqref{eq:maj-pinned-identity} is the positive envelope that will control polymer
tails.

\subsection{Activity and counting bounds}
\label{sec:activity-counting-hypotheses}

We now formulate the two model-dependent inputs needed for convergence in large volume.

For magnetic polymers, define the carrier support on primal \((P+1)\)-cells by
\begin{equation}
	\operatorname{car}(M)
	:=
	\theta_{\underline\Lambda}(\Supp(M))
	\subset C_{P+1}(\underline\Lambda),
	\qquad
	|M|:=|\operatorname{car}(M)|,
	\qquad
	M\in \mathbb M .
	\label{eq:magnetic-carrier-size}
\end{equation}
For electric polymers, set
\begin{equation}
	\operatorname{car}(E):=\Supp(E)\subset C_P(\underline\Lambda),
	\qquad
	|E|:=|\Supp(E)|,
	\qquad
	E\in \mathbb E .
	\label{eq:electric-carrier-size}
\end{equation}

Define the carrier-cell sets
\begin{equation}
	\mathfrak C_m:=C_{P+1}(\underline\Lambda),
	\qquad
	\mathfrak C_e:=C_P(\underline\Lambda).
	\label{eq:carrier-cell-sets}
\end{equation}
Thus, for \(A\in\mathbb A_s\), the statement \(u\in\operatorname{car}(A)\) is meaningful
for \(u\in\mathfrak C_s\).

Define the uniform nontrivial local amplitudes by
\begin{equation}
	t_m
	:=
	\max_{\substack{c\in C_{P+1}(\underline\Lambda)\\ m\in \mathbb Z_N\setminus\{0\}}}
	|\varpi_c(m)|,
	\qquad
	t_e
	:=
	\max_{\substack{u\in C_P(\underline\Lambda)\\ e\in \mathbb Z_N\setminus\{0\}}}
	|\upsilon_u(e)|.
	\label{eq:tm-te}
\end{equation}
Since the one-polymer activities
factor over occupied carrier cells, see \eqref{eq:one-polymer-activities-fixed-charge},
we have
\begin{equation}
	|\rho_m(M)|\le t_m^{|M|},
	\qquad
	|\rho_e(E)|\le t_e^{|E|}.
	\label{eq:activity-hypothesis}
\end{equation}
The parameters \(t_m,t_e\) tend to zero in the code limit
\eqref{eq:codeLimit}.

For \(u\in\mathfrak C_s\) and \(n\ge 1\), define the number of polymers of size $n$ containing $u$
\begin{equation}
	\mathcal N_s(u;n)
	:=
	\#\left\{
	A\in\mathbb A_s:
	u\in\operatorname{car}(A),\ |A|=n
	\right\}.
	\label{eq:Ns-def}
\end{equation}
We assume that this number grows at most exponentially with increasing polymer size, i.e., there exists a constant \(C_s\ge 1\) such that
\begin{equation}
	\mathcal N_s(u;n)\le C_s^n
	\qquad
	\text{for every }u\in\mathfrak C_s,\ n\ge 1.
	\label{eq:counting-hypothesis}
\end{equation}

This counting hypothesis is automatic on families of carrier adjacency graphs of uniformly
bounded degree. Indeed, let \(\Delta_s\) be the maximal degree of the carrier graph for
species \(s\). A connected support of size \(n\) containing a fixed carrier cell \(u\) has a
spanning tree. A depth-first traversal of such a spanning tree produces a closed walk of
length \(2n-2\), and hence the number of possible connected supports is at most
\(\Delta_s^{2n-2}\). Since each occupied carrier cell carries a nonzero coefficient in
\(\mathbb Z_N\setminus\{0\}\), the number of coefficient assignments is at most \((N-1)^n\).
Therefore
\begin{equation}
	\mathcal N_s(u;n)
	\le
	(N-1)^n\Delta_s^{2n-2}
	\le
	\left((N-1)\max\{1,\Delta_s^2\}\right)^n.
	\label{eq:counting-hypothesis-crude}
\end{equation}
Thus one may always take
\begin{equation}
	C_s=(N-1)\max\{1,\Delta_s^2\}.
	\label{eq:Cs-crude}
\end{equation}

A sharper explicit choice follows from the rooted connected-set bound: in a graph of
maximum degree \(\Delta_s\), the number of connected vertex sets with \(k+1\) vertices containing
a fixed root is at most \((e(\Delta_s-1))^k\); see \cite[p.~2]{KangasEtAl2018}. Applying this to
the carrier graph gives
\begin{equation}
	\mathcal N_s(u;n)
	\le
	(N-1)^n \left(e(\Delta_s-1)\right)^{n-1}
	\le
	\left((N-1)\max\{1,e(\Delta_s-1)\}\right)^n.
	\label{eq:counting-hypothesis-sharp}
\end{equation}
Hence, whenever \(\Delta_s\ge 2\), one may use the improved explicit choice
\begin{equation}
	C_s=(N-1)e(\Delta_s-1).
	\label{eq:Cs-sharp}
\end{equation}

\begin{remark}[The combinatorial bottleneck]
	\label{rem:counting-bottleneck}
	The counting input \eqref{eq:counting-hypothesis} is the weakest point in the
	present derivation of the low-activity region. The bounds
	\eqref{eq:counting-hypothesis-crude} and \eqref{eq:counting-hypothesis-sharp}
	effectively count general connected labeled carrier supports, and so do not
	fully exploit the fact that the polymers are support-connected \emph{closed}
	cycles. Consequently, the constants \(C_s\) in
	\eqref{eq:main-low-activity-condition} and
	\eqref{eq:main-low-activity-region} are controlled by a purely combinatorial
	counting problem for rooted connected closed polymers. Any improved bound on
	\(\mathcal N_s(u;n)\) incorporating closedness would therefore directly
	sharpen the rigorous low-activity region, without changing the analytic part
	of the argument.
\end{remark}

\subsection{The low-activity criterion}
\label{sec:low-activity-criterion}

We now combine the activity and counting estimates with the
Fern\'andez-Procacci criterion for hard-core polymer gases
\cite[Theorem~1]{FernandezProcacci2007}.

Let $\mathbb A$ be a set of hard-core polymers with positive activities $z:\mathbb A \to \mathbb R_{\ge 0}$. Write \(B\nsim A\) to mean
that \(B\) is incompatible with \(A\) and assume $A \nsim A$ for all $A \in \mathbb A$. For any function $\mu:\mathbb A \to \mathbb R_{\ge 0}$ define
\begin{equation}
	N_A^*(\mu)
	:=
	1+
	\sum_{n\ge 1}
	\frac{1}{n!}
	\sum_{\substack{B_1,\dots,B_n \in \mathbb A\\
			B_i\nsim A,\;
			B_i\sim B_j}}
	\prod_{i=1}^n \mu(B_i).
	\label{eq:FP-neighborhood}
\end{equation}
The Fern\'andez-Procacci criterion says that if a positive activity function $\mu:\mathbb A \to \mathbb R_{\ge 0}$ exists
such that
\begin{equation}
	z(A)\,N_A^*(\mu)\le \mu(A)
	\qquad
	\text{for every polymer }A \in \mathbb A,
	\label{eq:FP-criterion}
\end{equation}
then the hard-core cluster expansion converges in the large volume limit and the pinned expansion with activity $z$ satisfies\footnote{Pinned expansion as in \eqref{eq:maj-pinned} but with activity $z$ in place of $\rho$ in \eqref{eq:majorant-gas}.}
\begin{equation}
	z(A)\,\Pi(A)\le \mu(A).
	\label{eq:FP-pinned}
\end{equation}
For subset polymers with incompatibility given by overlap, the
Fern\'andez-Procacci neighborhood partition function admits the sitewise upper bound
used below; compare also the polymer-on-a-graph specialization in
\cite[Section~3.2]{FernandezProcacci2007} and
\cite[Section~2.3 and Proposition~4.1]{BissacotFernandezProcacci2010}.

Fix \(a_s>0\) and set
\begin{equation}
	\mu_s(A)
	:=
	|\rho_s(A)|\,e^{a_s|A|},
	\qquad
	A\in\mathbb A_s.
	\label{eq:mu-choice}
\end{equation}
Let
\begin{equation}
	S_s(a_s)
	:=
	\sup_{u\in\mathfrak C_s}
	\sum_{\substack{A\in\mathbb A_s\\ u\in\operatorname{car}(A)}}
	\mu_s(A).
	\label{eq:sitewise-sum}
\end{equation}
If
\begin{equation}
	S_s(a_s)\le e^{a_s}-1,
	\label{eq:sitewise-FP}
\end{equation}
then the Fern\'andez-Procacci condition holds for the species-\(s\) majorant gas.

Indeed, fix \(A\in\mathbb A_s\). Every polymer \(B \in \mathbb A_s\) incompatible with \(A\) intersects
\(\operatorname{car}(A)\). Moreover, in \(N_A^*(\mu_s)\), the polymers
\(B_1,\dots,B_n\) are required to be mutually compatible, hence have disjoint carrier
supports. Therefore the contribution of all mutually compatible families
\(\{B_1,\dots,B_n\}\) incompatible with \(A\) is bounded by \emph{independently} choosing, for
each carrier cell of \(A\), at most one polymer containing that carrier cell. This gives
\begin{equation}
	N_A^*(\mu_s)
	\le
	\prod_{u\in\operatorname{car}(A)}
	\Biggl(
	1+
	\sum_{\substack{B\in\mathbb A_s\\ u\in\operatorname{car}(B)}}
	|\rho_s(B)|e^{a_s|B|}
	\Biggr)
	\le
	(1+S_s(a_s))^{|A|}.
	\label{eq:NA-star-site-bound}
\end{equation}
The first inequality is not an equality because in choosing, for each carrier cell $u$ of $A$, at most one polymer containing $u$ independently, we are free to choose polymers $B, B'$ containing distinct carrier cells $u,u'\in \mathrm{car}(A)$, $u\in\mathrm{car}(B)$, $u'\in\mathrm{car}(B')$, $u \ne u'$ that are nevertheless incompatible due to $\mathrm{car}(B) \cap \mathrm{car}(B') \ne \varnothing$. If \eqref{eq:sitewise-FP} holds, then
\begin{equation}
N_A^*(\mu_s)\le e^{a_s|A|}.
\end{equation}
Consequently,
\begin{equation}
|\rho_s(A)|\,N_A^*(\mu_s)
\le
|\rho_s(A)|e^{a_s|A|}
=
\mu_s(A),
\end{equation}
which is exactly \eqref{eq:FP-criterion}.

Using \eqref{eq:activity-hypothesis} and \eqref{eq:counting-hypothesis}, we estimate
\begin{align}
	S_s(a_s)
	\le
	\sup_{u\in\mathfrak C_s}
	\sum_{n\ge 1}
	\sum_{\substack{A\in\mathbb A_s\\
			u\in\operatorname{car}(A),\ |A|=n}}
	t_s^n e^{a_s n}
	\le
	\sum_{n\ge 1}
	C_s^n(t_s e^{a_s})^n
	=
	\frac{C_s t_s e^{a_s}}{1-C_s t_s e^{a_s}},
	\label{eq:sitewise-series}
\end{align}
provided \(C_s t_s e^{a_s}<1\). Thus a sufficient condition for
\eqref{eq:sitewise-FP} is
\begin{equation}
	\frac{C_s t_s e^{a_s}}{1-C_s t_s e^{a_s}}
	\le
	e^{a_s}-1
	\qquad\Leftrightarrow\qquad
	C_s t_s \le e^{-a_s}(1-e^{-a_s}).
	\label{eq:main-low-activity-condition}
\end{equation}

\begin{definition}[Rigorous low-activity region]
	\label{def:low-activity-region}
	We say that the model lies in the rigorous low-activity region if there exist
	\(a_m,a_e>0\) such that
	\begin{equation}
		C_m t_m \le e^{-a_m}(1-e^{-a_m}),
		\qquad
		C_e t_e \le e^{-a_e}(1-e^{-a_e}).
		\label{eq:main-low-activity-region}
	\end{equation}
\end{definition}

Since \(\sup_{a>0}e^{-a}(1-e^{-a})=1/4\), the region is nonempty whenever
\(C_s t_s<1/4\), although the parameterized form
\eqref{eq:main-low-activity-region} is more useful for tail bounds.

\begin{theorem}[Low-activity control]
	\label{thm:low-activity-region}
	Assume that the physical one-polymer activities satisfy the size bound
	\eqref{eq:activity-hypothesis}, that the polymer counting bound
	\eqref{eq:counting-hypothesis} holds, and that the corresponding low-activity
	condition \eqref{eq:main-low-activity-region} is satisfied. In what follows,
	``convergent'' means absolutely convergent uniformly in finite volume, and therefore
	stable under the large-volume limit.
	
	For \(s\in\{m,e\}\) and \(A\in\mathbb A_s\), define the unnormalized sector and
	fixed-background occupation expectations by
	\begin{equation}
		\left\langle \mathbf 1_A \right\rangle_{\ell_m,\ell_e;q_m,q_e}
		:=
		\rho^{\ell_m,\ell_e;q_m,q_e}(A)\,
		\frac{\partial}{\partial \rho^{\ell_m,\ell_e;q_m,q_e}(A)}
		\widehat{\mathcal Z}_M(\ell_m,\ell_e;q_m,q_e),
		\label{eq:physical-sector-occupancy}
	\end{equation}
	and
	\begin{equation}
		\left\langle \mathbf 1_A \right\rangle_{q_m,q_e}
		:=
		\rho_s(A)\,
		\frac{\partial}{\partial \rho_s(A)}
		\mathcal Z_M(q_m,q_e).
		\label{eq:physical-fixed-background-occupancy}
	\end{equation}
	Then:
	\begin{enumerate}[label=(\roman*)]
		\item \label{thmItem:Zfinite}
		For every \((\ell_m,\ell_e)\in\widehat H_m\times\widehat H_e\) and every
		background pair \((q_m,q_e)\), the complex sector amplitude
		\(\widehat{\mathcal Z}_M(\ell_m,\ell_e;q_m,q_e)\) of \eqref{eq:Zlmle} and the
		fixed-background partition function \(\mathcal Z_M(q_m,q_e)\), reconstructed from
		the sector amplitudes by \eqref{eq:Zme-deltaResolved}, converge.
		
		\item \label{thmItem:occFinite}
		For every \(s\in\{m,e\}\), every \(A\in\mathbb A_s\), every
		\((\ell_m,\ell_e)\in\widehat H_m\times\widehat H_e\), and every \((q_m,q_e)\), the
		unnormalized occupation expectations
		\(\langle \mathbf 1_A\rangle_{\ell_m,\ell_e;q_m,q_e}\) and
		\(\langle \mathbf 1_A\rangle_{q_m,q_e}\) converge. Moreover, there exists a finite
		constant \(C_{M,\underline\Lambda,s}^{\mathrm{la}}\), independent of \(A\),
		\((\ell_m,\ell_e)\), and \((q_m,q_e)\), such that
		\begin{align}
			\left|
			\left\langle \mathbf 1_A \right\rangle_{\ell_m,\ell_e;q_m,q_e}
			\right|
			\le&\;
			C_{M,\underline\Lambda,s}^{\mathrm{la}}\,(t_s e^{a_s})^{|A|},
			\label{eq:sector-occupancy-bound}
		\\
			\left|
			\left\langle \mathbf 1_A \right\rangle_{q_m,q_e}
			\right|
			\le&\;
			C_{M,\underline\Lambda,s}^{\mathrm{la}}\,(t_s e^{a_s})^{|A|}.
			\label{eq:fixed-background-occupancy-bound}
		\end{align}
	\end{enumerate}
\end{theorem}

\begin{proof}
	Under \eqref{eq:activity-hypothesis}, \eqref{eq:counting-hypothesis}, and
	\eqref{eq:main-low-activity-region}, the Fernández-Procacci argument of
	\S\ref{sec:low-activity-criterion} applies to the positive same-species hard-core
	majorants. In particular, the majorant partition functions are finite, and the
	corresponding normalized occupation expectations satisfy
	\begin{equation}
		|\rho_s(A)|\,\Pi_s^{\mathrm{maj}}(A)
		\le
		|\rho_s(A)|e^{a_s|A|}
		\le
		(t_s e^{a_s})^{|A|},
		\qquad
		A\in\mathbb A_s.
		\label{eq:maj-occupancy-bound}
	\end{equation}
	Since the complex sector amplitudes are controlled by the majorants through
	\eqref{eq:twisted-sector-majorant-bound}, this proves the convergence of
	\(\widehat{\mathcal Z}_M(\ell_m,\ell_e;q_m,q_e)\) in the low-activity region. The
	convergence of the fixed-background amplitudes \(\mathcal Z_M(q_m,q_e)\) follows
	immediately from \eqref{eq:fixed-charge-majorant-bound}, or equivalently from the
	finite reconstruction \eqref{eq:Zme-deltaResolved}. This proves
	\ref{thmItem:Zfinite}.
	
	For \ref{thmItem:occFinite}, differentiate the sector expansion \eqref{eq:Zlmle} with
	respect to the marked activity \(\rho^{\ell_m,\ell_e;q_m,q_e}(A)\). By
	\eqref{eq:physical-sector-occupancy},
	\begin{equation}
		\left\langle \mathbf 1_A \right\rangle_{\ell_m,\ell_e;q_m,q_e}
		=
		\rho^{\ell_m,\ell_e;q_m,q_e}(A)\,
		\frac{\partial}{\partial \rho^{\ell_m,\ell_e;q_m,q_e}(A)}
		\widehat{\mathcal Z}_M(\ell_m,\ell_e;q_m,q_e).
	\end{equation}
	Taking absolute values and using the same majorization as in
	\eqref{eq:twisted-sector-majorant-bound} yields
	\begin{equation}
		\left|
		\left\langle \mathbf 1_A \right\rangle_{\ell_m,\ell_e;q_m,q_e}
		\right|
		\le
		\mathcal Z_{\bar s}^{\mathrm{maj}}\,
		|\rho_s(A)|\,
		\frac{\partial}{\partial |\rho_s(A)|} \mathcal Z_s^{\mathrm{maj}},
		\qquad
		\bar m=e,\ \bar e=m.
	\end{equation}
	Using \eqref{eq:maj-pinned-identity} and then \eqref{eq:maj-occupancy-bound}, we obtain
	\begin{equation}
		\left|
		\left\langle \mathbf 1_A \right\rangle_{\ell_m,\ell_e;q_m,q_e}
		\right|
		\le
		\mathcal Z_m^{\mathrm{maj}}\mathcal Z_e^{\mathrm{maj}}\,
		(t_s e^{a_s})^{|A|}.
	\end{equation}
	This proves \eqref{eq:sector-occupancy-bound} at sector level with
	\begin{equation}
		C_{M,\underline\Lambda,s}^{\mathrm{la}}
		:=
		\mathcal Z_m^{\mathrm{maj}}\mathcal Z_e^{\mathrm{maj}}.
	\end{equation}
	
	Similarly,
	\begin{equation}
		\left\langle \mathbf 1_A \right\rangle_{q_m,q_e}
		=
		\rho_s(A)\,
		\frac{\partial}{\partial \rho_s(A)}
		\mathcal Z_M(q_m,q_e).
	\end{equation}
	Differentiating the finite Fourier formula \eqref{eq:Zme-deltaResolved} and using the
	sector estimate already proved shows that
	\begin{equation}
		\left|
		\left\langle \mathbf 1_A \right\rangle_{q_m,q_e}
		\right|
		\le
		C_{M,\underline\Lambda,s}^{\mathrm{la}}\,(t_s e^{a_s})^{|A|}.
	\end{equation}
	This gives \eqref{eq:fixed-background-occupancy-bound}.
\end{proof}

\begin{remark}[Direct fixed-background majorization]\label{rmk:directZmeBound}
	The proof of convergence for $\mathcal Z_M(q_m,q_e)$ and \eqref{eq:fixed-background-occupancy-bound}
	can also be given without passing through the sector amplitudes
	\(\widehat{\mathcal Z}_M(\ell_m,\ell_e;q_m,q_e)\).  Indeed, start from \eqref{eq:fixed-charge-polymer-gas} and take the absolute value. This removes the phase,
	while the homology constraints are bounded by \(1\). Hence
	\begin{equation}
		|\mathcal Z_M(q_m,q_e)|
		\le
		\sum_{r,s\ge 0}
		\frac{1}{r!\,s!}
		\sum_{\substack{M_1,\dots,M_r\in \mathbb M\\ \text{pairwise disjoint}}}
		\sum_{\substack{E_1,\dots,E_s\in \mathbb E\\ \text{pairwise disjoint}}}
		\prod_{i=1}^r |\rho_m(M_i)|
		\prod_{j=1}^s |\rho_e(E_j)|
		=
		\mathcal Z_m^{\mathrm{maj}}\mathcal Z_e^{\mathrm{maj}}.
	\end{equation}
	This yields the convergence of \(\mathcal Z_M(q_m,q_e)\) directly, uniformly in the
	background pair \((q_m,q_e)\), without invoking \eqref{eq:Zlmle} or
	\eqref{eq:Zme-deltaResolved}.
	
	The same idea gives the fixed-background occupancy bound. If \(A\in\mathbb A_s\),
	differentiate the fixed-background polymer expansion directly with respect to the
	marked activity \(\rho_s(A)\). After taking absolute values, the phase again drops
	out, the homology constraints are bounded by \(1\), and one obtains
	\begin{equation}
		\left|
		\left\langle \mathbf 1_A \right\rangle_{q_m,q_e}
		\right|
		\le
		\mathcal Z_{\bar s}^{\mathrm{maj}}\,
		|\rho_s(A)|\,
		\frac{\partial}{\partial |\rho_s(A)|}\mathcal Z_s^{\mathrm{maj}},
		\qquad
		\bar m=e,\ \bar e=m.
	\end{equation}
	Using \eqref{eq:maj-pinned-identity} and \eqref{eq:maj-occupancy-bound} then gives
	\eqref{eq:fixed-background-occupancy-bound} with the same constant
	\(C_{M,\underline\Lambda,s}^{\mathrm{la}}
	=\mathcal Z_m^{\mathrm{maj}}\mathcal Z_e^{\mathrm{maj}}\).
\end{remark}

\subsection{Large polymer suppression}
\label{sec:physical-occupation-tails}

Theorem~\ref{thm:low-activity-region} gives an exponential bound on each unnormalized physical one-polymer occupation expectation. Summing that bound over suitable polymer
families yields the following two consequences.

\begin{corollary}[Large polymer suppression]
	\label{cor:physical-occupation-tails}
	Assume the hypotheses of Theorem~\ref{thm:low-activity-region}. Fix a species
	\(s\in\{m,e\}\), and let
	\begin{equation}
		\left\langle \mathbf 1_A \right\rangle_\bullet
		\in
		\Bigl\{
		\left\langle \mathbf 1_A \right\rangle_{\ell_m,\ell_e;q_m,q_e},
		\;
		\left\langle \mathbf 1_A \right\rangle_{q_m,q_e}
		\Bigr\}
	\end{equation}
	denote either of the unnormalized physical occupation expectations. Then:
	\begin{enumerate}[label=(\roman*)]
		\item\label{corItem:large-tail}
		The total unnormalized physical occupation weight carried by connected polymers of
		species \(s\) that pass through a fixed carrier cell \(u\) and have size at least \(L\)
		is exponentially small in \(L\). Concretely, for every
		\(u\in\mathfrak C_s\) and every integer \(L\ge 1\),
		\begin{equation}
			\sum_{\substack{A\in\mathbb A_s\\
					u\in\operatorname{car}(A),\ |A|\ge L}}
			\left|
			\left\langle \mathbf 1_A \right\rangle_\bullet
			\right|
			\le
			C_{M,\underline\Lambda,s}^{\mathrm{la}}\,
			\frac{(C_s t_s e^{a_s})^L}{1-C_s t_s e^{a_s}}.
			\label{eq:physical-tail-bound}
		\end{equation}
		
		\item\label{corItem:homological-tail}
		In particular, the total unnormalized physical occupation weight carried by connected
		polymers that are homologically nontrivial, and hence capable of contributing to a
		change of spacetime topological sector, is exponentially small on the scale of the
		corresponding spacetime systole. Define
		\begin{equation}
			\mathrm{sys}_s(\underline\Lambda)
			:=
			\min\left\{
			|A|:
			A\in\mathbb A_s,\ [A]\neq 0
			\right\},
			\qquad s\in\{m,e\},
			\label{eq:systole}
		\end{equation}
		with the convention \(+\infty\) if there is no homologically nontrivial connected
		polymer of species \(s\). Then
		\begin{equation}
			\sum_{\substack{A\in\mathbb A_s\\ [A]\neq 0}} \left|\left\langle \mathbf 1_A \right\rangle_\bullet\right|
			\le
			C_{M,\underline\Lambda,s}^{\mathrm{la}}\,
			|\mathfrak C_s|\,
			\frac{(C_s t_s e^{a_s})^{\mathrm{sys}_s(\underline\Lambda)}}
			{1-C_s t_s e^{a_s}},
			\label{eq:physical-bad-weight}
		\end{equation}
	\end{enumerate}
\end{corollary}

\begin{proof}
	By Theorem~\ref{thm:low-activity-region},
	\begin{equation}
		\left|
		\left\langle \mathbf 1_A \right\rangle_\bullet
		\right|
		\le
		C_{M,\underline\Lambda,s}^{\mathrm{la}}\,(t_s e^{a_s})^{|A|}
		\qquad
		(A\in\mathbb A_s).
		\label{eq:physical-one-polymer-bound-proof}
	\end{equation}
	Summing \eqref{eq:physical-one-polymer-bound-proof} over all polymers of species \(s\)
	with \(u\in\operatorname{car}(A)\) and \(|A|\ge L\), and using the counting estimate
	\eqref{eq:counting-hypothesis}, gives
	\begin{align}
		\sum_{\substack{A\in\mathbb A_s\\
				u\in\operatorname{car}(A),\ |A|\ge L}}
		\left|
		\left\langle \mathbf 1_A \right\rangle_\bullet
		\right|
		&\le
		C_{M,\underline\Lambda,s}^{\mathrm{la}}
		\sum_{n\ge L}
		\mathcal N_s(u;n)(t_s e^{a_s})^n
		\nonumber\\
		&\le
		C_{M,\underline\Lambda,s}^{\mathrm{la}}
		\sum_{n\ge L}
		C_s^n(t_s e^{a_s})^n
		=
		C_{M,\underline\Lambda,s}^{\mathrm{la}}\,
		\frac{(C_s t_s e^{a_s})^L}{1-C_s t_s e^{a_s}},
	\end{align}
	which proves \eqref{eq:physical-tail-bound}.
	
	For \eqref{eq:physical-bad-weight}, note that every homologically nontrivial polymer of
	species \(s\) has size at least \(\mathrm{sys}_s(\underline\Lambda)\). Applying
	\eqref{eq:physical-tail-bound} with \(L=\mathrm{sys}_s(\underline\Lambda)\) and then summing
	over all carrier cells \(u\in\mathfrak C_s\), gives the stated bound. As before, this
	overcounts each polymer by its carrier size, which is harmless for an upper bound.
\end{proof}

The bounds of Corollary~\ref{cor:physical-occupation-tails} admit a direct physical
interpretation. From the statistical-mechanical point of view, in a positive polymer
ensemble the normalized occupation expectation of a polymer is the probability that the
polymer occurs, and sums over classes of polymers measure the expected number of such
excitations. In the present complex theory the sector weights are not positive, so the
occupation expectations are amplitudes rather than probabilities. The absolute values in
\eqref{eq:physical-tail-bound} and \eqref{eq:physical-bad-weight} therefore remove possible
phase cancellations and measure the total absolute contribution carried by the specified
family of polymers. The corollary shows that this total contribution is exponentially
small for large polymers, and in particular for homologically nontrivial polymers.

From the quantum-mechanical point of view, a connected polymer is an extended Euclidean
defect history in the spacetime representation of the thermal trace. A homologically
nontrivial polymer is a defect history capable of changing the spacetime topological
sector, and hence of contributing to a logical sector-changing process. The bound
\eqref{eq:physical-tail-bound} says that long connected defect histories have
exponentially small total amplitude weight, while \eqref{eq:physical-bad-weight} says
that defect histories capable of changing the topological sector are exponentially
suppressed on the scale of the spacetime systole. Thus, within the low-activity region,
the total absolute weight carried by dangerous large or sector-changing thermal histories
is perturbatively small.

\section{Kramers-Wannier duality}
\label{sec:KW-duality}

We now derive the exact duality transformation of the spacetime theory. This is the
finite-\(\mathbb Z_N\), higher-form analogue of the
Kramers-Wannier/Wegner duality
\cite{KramersWannier1941,Wegner:1971app,Savit1980,ElitzurPearsonShigemitsu1979}; for a
modern algebraic viewpoint on exact lattice dualities, see
\cite{CobaneraOrtizNussinov2011}. Historically, one of the key lessons of the
Kramers-Wannier program is that duality is often revealed by Fourier transform.
In the classical lattice setting, the low-temperature/high-temperature transform can be
recast as a local Fourier transform of Boltzmann weights, and in Abelian models this
became the standard route to dual partition functions
\cite{KramersWannier1941,Wegner:1971app,Savit1980}. More recently, the same Fourier-theoretic
structure has been reinterpreted from a higher-dimensional and topological viewpoint:
for finite gauge theories and generalized Ising-type systems, Kramers-Wannier duality can be
understood as the boundary manifestation of a duality in one dimension higher, with the
Fourier transform acting on local weights or background sectors
\cite{FreedTeleman2022,DelcampIshtiaque2024}.

Our derivation follows this general philosophy, applying the
Fourier transform directly to the exact fixed-background amplitudes of
\S\ref{sec:fixed-charge-polymer-gas}. The natural observables are the fixed-background partition functions
\begin{equation}
\mathcal Z_M(q_m,q_e),
\qquad
q_m\in Z_{d-P}(\underline\Lambda^\vee),
\qquad
q_e\in Z_P(\underline\Lambda).
\end{equation}

The cleanest route is to work first in the fixed-background basis, perform a cellwise
finite Fourier transform of the local \((P+1)\)-cell weights, and reinterpret the resulting
auxiliary variables as a dual \((d-P)\)-form field on \(\underline\Lambda^\vee\). This
exchanges magnetic and electric background data, simultaneously exchanges the \(P\)-form
theory on \(\underline\Lambda\) with a \((d-P)\)-form theory on
\(\underline\Lambda^\vee\), and turns the local Fourier transform of weights into the exact
dual coupling map. Throughout this section $N$ is any integer $\ge 2$ and,
\begin{equation}
W_c:\mathbb Z_N\to\mathbb C
\qquad \text{for } c\in C_{P+1}(\underline\Lambda),
\qquad
V_u:\mathbb Z_N\to\mathbb C
\qquad \text{for } u\in C_P(\underline\Lambda)
\end{equation}
are arbitrary local spacetime weights. At the end we specialize to the explicit Trotter
weights  \eqref{eq:parallel-weights}, \eqref{eq:perp-weights}. We use the Fourier-transform conventions of Appendix~\ref{app:fourier-ZN}.

\subsection{Duality in the fixed-charge basis}
\label{sec:KW-fixed-charge}

We begin with the unnormalized fixed-charge partition function
\eqref{eq:Zme}, but for the present discussion we make the weight
dependence explicit:
\begin{equation}
	Z_M(q_m,q_e;W,V)
	:=
	\sum_{\phi\in \Omega^P(\underline\Lambda)}
	\chi \left(\int_{q_e}\phi\right)
	\prod_{c\in C_{P+1}(\underline\Lambda)}
	W_c \left((D\phi+\vartheta_{\underline\Lambda}(q_m))_c\right)
	\prod_{u\in C_P(\underline\Lambda)}
	V_u(\phi_u).
	\label{eq:KW-raw-fixed-charge}
\end{equation}

\subsubsection{Fourier expansion of the \((P+1)\)-cell weights}

Replace \(W_c\) in \eqref{eq:KW-raw-fixed-charge} by its Fourier series. This introduces an
auxiliary \((P+1)\)-cochain $\lambda\in \Omega^{P+1}(\underline\Lambda)$ and gives
\begin{equation}
	\begin{aligned}
		Z_M(q_m,q_e;W,V)
		&=
		N^{-|C_{P+1}(\underline\Lambda)|/2}
		\sum_{\phi,\lambda}
		\chi \left(\int_{q_e}\phi\right)
		\Biggl(
		\prod_{c\in C_{P+1}(\underline\Lambda)}
		\widehat W_c(\lambda_c)\,
		\chi_{\lambda_c} \left((\vartheta_{\underline\Lambda}(q_m))_c\right)
		\Biggr)
		\\
		&\hspace{7em}\times
		\prod_{c\in C_{P+1}(\underline\Lambda)}
		\chi_{\lambda_c} \left((D\phi)_c\right)
		\prod_{u\in C_P(\underline\Lambda)}V_u(\phi_u).
	\end{aligned}
	\label{eq:KW-step1}
\end{equation}

Next move the differential \(D\) from \(\phi\) onto \(\lambda\). By the adjointness of
\(D\) and \(D^{\mathsf T}\) (see \eqref{eq:adjoint-d-dT}),
\begin{equation}
	\sum_{c\in C_{P+1}(\underline\Lambda)}\lambda_c(D\phi)_c
	=
	\sum_{u\in C_P(\underline\Lambda)}(D^{\mathsf T}\lambda)_u\,\phi_u.
	\label{eq:KW-adjointness}
\end{equation}
Hence,\footnote{We use the isomorphism \eqref{eq:PX} strictly for typecasting: change the $P$-cycle $q_e$ to a $P$-cocycle $\mathcal P^{-1}_{\underline\Lambda}(q_e)$. They have the same coefficients, concretely, $\mathcal P_{\underline\Lambda}^{-1}(q_e)_u = (q_e)_u$.}
\begin{equation}
	\chi \left(\int_{q_e}\phi\right)
	\prod_{c\in C_{P+1}(\underline\Lambda)}\chi_{\lambda_c} \left((D\phi)_c\right)
	=
	\prod_{u\in C_P(\underline\Lambda)}
	\chi_{\bigl(\mathcal P_{\underline\Lambda}^{-1}(q_e)+D^{\mathsf T}\lambda\bigr)_u}(\phi_u),
\end{equation}
and the \(\phi\)-dependence factorizes cellwise:
\begin{equation}
	\begin{aligned}
		Z_M(q_m,q_e;W,V)
		=&\;
		N^{-|C_{P+1}(\underline\Lambda)|/2}
		\sum_{\lambda\in \Omega^{P+1}(\underline\Lambda)}
		\Biggl(
		\prod_{c\in C_{P+1}(\underline\Lambda)}
		\widehat W_c(\lambda_c)\,
		\chi_{\lambda_c} \left((\vartheta_{\underline\Lambda}(q_m))_c\right)
		\Biggr)
		\\
		&\hspace{7em}\times
		\prod_{u\in C_P(\underline\Lambda)}
		\Biggl[
		\sum_{x\in \mathbb Z_N}
		V_u(x)\,\chi_{\bigl(\mathcal P_{\underline\Lambda}^{-1}(q_e)+D^{\mathsf T}\lambda\bigr)_u}(x)
		\Biggr].
	\end{aligned}
	\label{eq:KW-step2}
\end{equation}

Using the inverse Fourier transform,
\begin{equation}
	\sum_{x\in \mathbb Z_N}V_u(x)\,\chi_k(x)
	=
	\sqrt N\,\widehat V_u(-k),
\end{equation}
we obtain
\begin{equation}
	\begin{aligned}
		\frac{Z_M(q_m,q_e;W,V)}{N^{|C_P(\underline\Lambda)|/2}}
		=&\;
		\frac{1}{N^{|C_{\check P}(\underline\Lambda^\vee)|/2}}
		\sum_{\lambda\in \Omega^{P+1}(\underline\Lambda)}
		\Biggl(
		\prod_{c\in C_{P+1}(\underline\Lambda)}
		\widehat W_c(\lambda_c)\,
		\chi_{\lambda_c} \left((\vartheta_{\underline\Lambda}(q_m))_c\right)
		\Biggr)
		\\
		& \hspace{8em} \times
		\Biggl(
		\prod_{u\in C_P(\underline\Lambda)}
		\widehat V_u \left(-\bigl(D^{\mathsf T}\lambda+\mathcal P_{\underline\Lambda}^{-1}(q_e)\bigr)_u\right)
		\Biggr),
	\end{aligned}
	\label{eq:KW-fixed-charge-primal-fourier}
\end{equation}
with
\begin{equation}
	\check P:=d-P. \label{eq:dualP}
\end{equation}

\subsubsection{Reinterpretation on the dual spacetime lattice}

Define the dual field
\begin{equation}
	\check\phi\in \Omega^{\check P}(\underline\Lambda^\vee)
\end{equation}
by transporting \(\lambda\) across the primal-dual correspondence:
\begin{equation}
	\check\phi_{c^\vee}
	:=
	\lambda_{\theta_{\underline\Lambda}(c^\vee)},
	\qquad
	c^\vee\in C_{\check P}(\underline\Lambda^\vee).
	\label{eq:KW-dual-field}
\end{equation}
The intertwining of \(D^{\mathsf T}\) on the primal side with \(D^\vee\) on the dual side
(see Prop.~\ref{prop:intertwine}\ref{prop:theta-sharp-intertwine2}) gives
\begin{equation}
	(D^{\mathsf T}\lambda)_{\theta_{\underline\Lambda}(u^\vee)}
	=
	(D^\vee \check\phi)_{u^\vee},
	\qquad
	u^\vee\in C_{\check P+1}(\underline\Lambda^\vee).
	\label{eq:KW-Dt-dual}
\end{equation}

The magnetic charge term becomes a Wilson factor of the dual theory:
\begin{align}
	\prod_{c\in C_{P+1}(\underline\Lambda)}
	\chi_{\lambda_c} \left((\vartheta_{\underline\Lambda}(q_m))_c\right)
	&=
	\prod_{c^\vee\in C_{\check P}(\underline\Lambda^\vee)}
	\chi_{\check\phi_{c^\vee}} \left((q_m)_{c^\vee}\right)
	=
	\chi \left(\int_{q_m}\check\phi\right).
	\label{eq:KW-magnetic-to-dual-Wilson}
\end{align}
And the electric charge becomes a defect cocycle of the dual theory. Define
\begin{equation}
	q_e^\vee := \vartheta_{\underline\Lambda^\vee}(q_e)\in Z^{\check P+1}(\underline\Lambda^\vee),
\end{equation}
which is a cocycle since $\vartheta$ is an isomorphism (Corr.~\ref{cor:complexIso}). Then
\begin{equation}
	\bigl(D^{\mathsf T}\lambda+\mathcal P_{\underline\Lambda}^{-1}(q_e)\bigr)_{\theta_{\underline\Lambda}(u^\vee)}
	=
	\left(D^\vee\check\phi+q_e^\vee\right)_{u^\vee}.
	\label{eq:KW-electric-to-dual-defect}
\end{equation}

The natural dual local weights are therefore (using the map $\theta^\sharp$ from \eqref{eq:theta-sharp} and defining for any function $f$, $f^{-}(x):=f(-x)$)
\begin{equation}
	\check W := \theta_{\underline\Lambda^\vee}^\sharp(\widehat V^{-}),
	\qquad
	\check V := \theta_{\underline\Lambda^\vee}^\sharp(\widehat W),
\end{equation}
or explicitly,
\begin{equation}
	\check W_{u^\vee}
	:=
	\widehat V^{-}_{\theta_{\underline\Lambda}(u^\vee)},
	\qquad
	u^\vee\in C_{\check P+1}(\underline\Lambda^\vee),
	\label{eq:KW-dual-W}
\end{equation}
and
\begin{equation}
	\check V_{c^\vee}
	:=
	\widehat W_{\theta_{\underline\Lambda}(c^\vee)},
	\qquad
	c^\vee\in C_{\check P}(\underline\Lambda^\vee).
	\label{eq:KW-dual-V}
\end{equation}

This yields the dual \(\check P\)-form partition function on \(\underline\Lambda^\vee\):
\begin{equation}
	\begin{aligned}
		Z_M^\vee(q_e,q_m;\check W,\check V)
		:=
		\sum_{\check\phi\in \Omega^{\check P}(\underline\Lambda^\vee)}
		\chi \left(\int_{q_m}\check\phi\right)
		\prod_{u^\vee\in C_{\check P+1}(\underline\Lambda^\vee)}
		\check W_{u^\vee} \left((D^\vee\check\phi+q_e^\vee)_{u^\vee}\right)
		\\
		\times
		\prod_{c^\vee\in C_{\check P}(\underline\Lambda^\vee)}
		\check V_{c^\vee}(\check\phi_{c^\vee}).
	\end{aligned}
	\label{eq:KW-dual-raw-fixed-charge}
\end{equation}

\begin{prop}[Duality with background charges]
	\label{prop:KW-fixed-charge-duality}
	With \(\check P=d-P\), the background-dependent partition functions have the following duality properties:
	\begin{enumerate}[label=(\roman*)]
		\item\label{prop:KW-fixed-charge-raw}
		The unnormalized fixed-charge partition functions obey
		\begin{equation}
			\frac{Z_M(q_m,q_e;W,V)}{N^{|C_P(\underline\Lambda)|/2}}
			= \frac{Z_M^\vee(q_e,q_m;\check W,\check V)}{N^{|C_{\check P}(\underline\Lambda^\vee)|/2}}.
			\label{eq:KW-Z-raw}
		\end{equation}
		In particular, duality exchanges the electric and magnetic fixed charges and
		simultaneously exchanges the \(P\)-form theory on \(\underline\Lambda\) with the
		\(\check P\)-form theory on \(\underline\Lambda^\vee\).
		
		\item\label{prop:KW-fixed-charge-normalized}
		Let \(\mathcal Z_M(q_m,q_e;W,V)\) and
		\(\mathcal Z_M^\vee(q_e,q_m;\check W,\check V)\) denote the normalized
		fixed-charge amplitudes defined as in \eqref{eq:normalized-closed-defect-gas} for
		the primal and dual theories respectively. Then
		\begin{equation}
			\frac{|Z^P(\underline\Lambda)|}{N^{|C_P(\underline\Lambda)|}}\,
			\mathcal Z_M(q_m,q_e;W,V)
			=
			\frac{|Z^{\check P}(\underline\Lambda^\vee)|}
			{N^{|C_{\check P}(\underline\Lambda^\vee)|}}\,
			\mathcal Z_M^\vee(q_e,q_m;\check W,\check V).
			\label{eq:KW-Z}
		\end{equation}
	\end{enumerate}
\end{prop}

\begin{proof}
	Item~\ref{prop:KW-fixed-charge-raw} follows by substituting
	\eqref{eq:KW-magnetic-to-dual-Wilson},
	\eqref{eq:KW-electric-to-dual-defect}, \eqref{eq:KW-dual-W}, and
	\eqref{eq:KW-dual-V} into \eqref{eq:KW-fixed-charge-primal-fourier}.
	
	For item~\ref{prop:KW-fixed-charge-normalized}, let \(\mathcal N_M^\vee\) be the
	dual analogue of \eqref{eq:NM-fixed-charge}, namely the normalization obtained by
	replacing \((\underline\Lambda,P,W,V)\) with
	\((\underline\Lambda^\vee,\check P,\check W,\check V)\). Using
	\eqref{eq:KW-dual-W}, \eqref{eq:KW-dual-V}, and the symmetric Fourier convention,
	we have
	\begin{equation}
		\check W_{u^\vee}(0)=\widehat V_{\theta_{\underline\Lambda}(u^\vee)}(0),
		\qquad
		\widehat{\check V}_{c^\vee}(0)=W_{\theta_{\underline\Lambda}(c^\vee)}(0),
	\end{equation}
	so that
	\begin{equation}
		\mathcal N_M
		=
		N^{\frac{-|C_P(\underline\Lambda)|+|C_{\check P}(\underline\Lambda^\vee)|}{2}}
		\mathcal N_M^\vee.
		\label{eq:Ndual}
	\end{equation}
	Now dividing both sides of \eqref{eq:KW-Z-raw} by \(\mathcal N_M\) gives \eqref{eq:KW-Z}.
\end{proof}

\subsection{Specialization to the Trotter weights}
\label{sec:KW-Trotter}

We now specialize \eqref{eq:KW-dual-W} and \eqref{eq:KW-dual-V} to the explicit local
weights \eqref{eq:perp-weights} and \eqref{eq:parallel-weights}. The key point is that,
with the present Fourier conventions, the transformed weights again have the same functional
form as the original anisotropic Trotter weights, up to local field-independent factors.
This makes the dual coupling map transparent.

\paragraph{Dual horizontal \((\check P+1)\)-cell weights.}
Let
\begin{equation}
c^\vee\in C_{\check P+1}(\Lambda^\vee),
\qquad
a:=\theta_\Lambda(c^\vee)\in C_{P-1}(\Lambda).
\end{equation}
Then the horizontal dual cell \(c^\vee(i)\) corresponds to the vertical primal cell
\(\underline a(i)\), so
\begin{equation}
	\check W_{c^\vee(i)}^{\parallel}(m)
	=
	\widehat V_{\underline a(i)}^{\perp}(-m).
\end{equation}
Using \eqref{eq:V-perp},
\begin{align}
	\widehat V_{\underline a(i)}^{\perp}(-m)
	&=
	\frac{1}{\sqrt N}\sum_{x\in \mathbb Z_N}
	\left(
	\delta_N(x)+\frac{e^{\beta J_a/M}-1}{N}
	\right)\chi_{-m}(-x)
	\nonumber\\
	&=
	\frac{1}{\sqrt N}
	\Bigl(
	1+\left(e^{\beta J_a/M}-1\right)\delta_N(m)
	\Bigr)
	\nonumber\\
	&=
	\frac{1}{\sqrt N}\exp \left(\frac{\beta J_a}{M}\,\delta_N(m)\right).
	\label{eq:KW-dual-horizontal-W-explicit}
\end{align}
Thus, up to the field-independent factor \(N^{-1/2}\), the dual horizontal weight is again
of the form \eqref{eq:W-parallel}. We therefore define
\begin{equation}
	\check K_{c^\vee}:=J_{\theta_\Lambda(c^\vee)}.
	\label{eq:KW-dual-K-def}
\end{equation}

\paragraph{Dual vertical \(\check P\)-cell weights.}
Let
\begin{equation}
a^\vee\in C_{\check P-1}(\Lambda^\vee),
\qquad
c:=\theta_\Lambda(a^\vee)\in C_{P+1}(\Lambda).
\end{equation}
Then the vertical dual cell \(\underline{a^\vee}(i)\) corresponds to the horizontal primal
cell \(c(i)\), so
\begin{equation}
	\check V_{\underline{a^\vee}(i)}^{\perp}(x)
	=
	\widehat W_{c(i)}^{\parallel}(x).
\end{equation}
Using \eqref{eq:W-parallel},
\begin{align}
	\widehat W_{c(i)}^{\parallel}(x)
	&=
	\frac{1}{\sqrt N}\sum_{m\in \mathbb Z_N}
	\left(
	1+\left(e^{\beta K_c/M}-1\right)\delta_N(m)
	\right)\chi_x(-m)
	\nonumber\\
	&=
	\sqrt N
	\left(
	\delta_N(x)+\frac{e^{\beta K_c/M}-1}{N}
	\right).
	\label{eq:KW-dual-vertical-V-explicit}
\end{align}
Thus, up to the factor \(\sqrt N\), the dual vertical weight is again of the form
\eqref{eq:V-perp}, and we define
\begin{equation}
	\check J_{a^\vee}:=K_{\theta_\Lambda(a^\vee)}.
	\label{eq:KW-dual-J-def}
\end{equation}

\paragraph{Dual vertical \((\check P+1)\)-cell source weights.}
Let
\begin{equation}
u^\vee\in C_{\check P}(\Lambda^\vee),
\qquad
u:=\theta_\Lambda(u^\vee)\in C_P(\Lambda).
\end{equation}
Then the vertical dual \((\check P+1)\)-cell \(\underline{u^\vee}(i)\) corresponds to the
horizontal primal \(P\)-cell \(u(i)\), so
\begin{equation}
	\check W_{\underline{u^\vee}(i)}^{\perp}(m)
	=
	\widehat V_{u(i)}^{\parallel}(-m).
\end{equation}
Using \eqref{eq:V-parallel},
\begin{align}
	\widehat V_{u(i)}^{\parallel}(-m)
	&=
	\frac{1}{\sqrt N}\sum_{x\in \mathbb Z_N}
	\exp \left(\frac{\beta}{M}\sum_{n\in\mathbb Z_N}h_u^{(n)}\omega^{nx}\right)\omega^{mx}
	\nonumber\\
	&=
	\frac{\sqrt N}{N}\sum_{x\in \mathbb Z_N}
	\exp \left(\frac{\beta}{M}\sum_{n\in\mathbb Z_N}h_u^{(n)}\omega^{nx}\right)\omega^{mx}.
	\label{eq:KW-dual-vertical-W-first}
\end{align}
Comparing with the standard vertical weight for a \(\check P\)-form theory on
\(\Lambda^\vee\), one finds, after the change of variable \(x=(-1)^{\check P}j\), that the
functional form agrees up to the field-independent factor \(\sqrt N\), provided
\begin{equation}
	\check g_{u^\vee}^{(n)}
	=
	h_{u}^{((-1)^{\check P}n)}
	=
	h_{u}^{((-1)^{d-P}n)}.
	\label{eq:KW-dual-g-def}
\end{equation}

\paragraph{Dual horizontal \(\check P\)-cell source weights.}
Let
\begin{equation}
u^\vee\in C_{\check P}(\Lambda^\vee),
\qquad
u:=\theta_\Lambda(u^\vee)\in C_P(\Lambda).
\end{equation}
Then the horizontal dual \(\check P\)-cell \(u^\vee(i)\) corresponds to the vertical primal
\((P+1)\)-cell \(\underline u(i)\), so
\begin{equation}
	\check V_{u^\vee(i)}^{\parallel}(x)
	=
	\widehat W_{\underline u(i)}^{\perp}(x).
\end{equation}
Using \eqref{eq:W-perp},
\begin{align}
	\widehat W_{\underline u(i)}^{\perp}(x)
	&=
	\frac{1}{\sqrt N}\sum_{m\in \mathbb Z_N}
	\left[
	\frac{1}{N}\sum_{j\in\mathbb Z_N}
	\exp \left(\frac{\beta}{M}\sum_{n\in\mathbb Z_N}g_u^{(n)}\omega^{nj}\right)
	\omega^{(-1)^Pjm}
	\right]\omega^{-xm}
	\nonumber\\
	&=
	\frac{1}{\sqrt N}
	\exp \left(\frac{\beta}{M}\sum_{n\in\mathbb Z_N}g_u^{(n)}\omega^{n((-1)^P x)}\right).
	\label{eq:KW-dual-horizontal-V-explicit}
\end{align}
Thus, up to the factor \(N^{-1/2}\), the dual horizontal weight is again of the form
\eqref{eq:V-parallel}, with
\begin{equation}
	\check h_{u^\vee}^{(n)}
	=
	g_{u}^{((-1)^Pn)}.
	\label{eq:KW-dual-h-def}
\end{equation}

Collecting \eqref{eq:KW-dual-K-def}, \eqref{eq:KW-dual-J-def},
\eqref{eq:KW-dual-g-def}, and \eqref{eq:KW-dual-h-def}, we obtain the dual coupling map
\begin{equation}
	\check K_{c^\vee}=J_{\theta_\Lambda(c^\vee)},
	\qquad
	\check J_{a^\vee}=K_{\theta_\Lambda(a^\vee)},
	\label{eq:KW-dual-JK-map}
\end{equation}
and
\begin{equation}
	\check g_{u^\vee}^{(n)}
	=
	h_{\theta_\Lambda(u^\vee)}^{((-1)^{d-P}n)},
	\qquad
	\check h_{u^\vee}^{(n)}
	=
	g_{\theta_\Lambda(u^\vee)}^{((-1)^P n)}.
	\label{eq:KW-dual-gh-map}
\end{equation}

\subsubsection{Cancellation of the local factors}

The transformed local weights \eqref{eq:KW-dual-horizontal-W-explicit},
\eqref{eq:KW-dual-vertical-V-explicit}, \eqref{eq:KW-dual-vertical-W-first}, and
\eqref{eq:KW-dual-horizontal-V-explicit} differ from the canonical Trotter weights only by
local field-independent factors \(N^{\pm 1/2}\). These local factors exactly cancel the factors of $N$ in \eqref{eq:KW-Z-raw}, as we now show. Comparing the
Fourier transformed local factors with the original local factors
\eqref{eq:parallel-weights} and \eqref{eq:perp-weights} we obtain one factor \(N^{-1/2}\)
for each dual horizontal \((\check P+1)\)-cell, one factor \(N^{1/2}\) for each dual
vertical \((\check P+1)\)-cell, one factor \(N^{1/2}\) for each dual vertical
\(\check P\)-cell, and one factor \(N^{-1/2}\) for each dual horizontal \(\check P\)-cell.
Their total product is therefore
\begin{equation}
	\begin{aligned}
		\mathfrak F_{P,\underline\Lambda}
		&:=
		\prod_{i\in \mathbb Z/M\mathbb Z}
		\prod_{a\in C_{P-1}(\Lambda)}N^{-1/2}
		\prod_{c\in C_{P+1}(\Lambda)}N^{1/2}
		\prod_{b\in C_P(\Lambda)}N^{1/2}
		\prod_{b\in C_P(\Lambda)}N^{-1/2}
		\\
		&=
		N^{\frac{M}{2}\left(-|C_{P-1}(\Lambda)|+|C_{P+1}(\Lambda)|\right)}
		=
		N^{\frac{|C_{P+1}(\underline\Lambda)|-|C_P(\underline\Lambda)|}{2}}
		=
		N^{\frac{|C_{\check P}(\underline\Lambda^\vee)|-|C_P(\underline\Lambda)|}{2}}
	\end{aligned}
	\label{eq:KW-local-factor-product}
\end{equation}
Thus the transformed dual partition function differs from the canonical dual Trotter
partition function only by the factor \(\mathfrak F_{P,\underline\Lambda}\), and this
cancels exactly against the normalization factors in
\eqref{eq:KW-Z-raw}. It is convenient to record the Trotter amplitudes explicitly:
\begin{equation}
	Z_{P,M}(q_m,q_e;J,K,g,h)
	:=
	Z_M(q_m,q_e;W[J,K,g,h],V[J,K,g,h]).
	\label{eq:KW-Trotter-fixed-charge-def}
\end{equation}

\begin{theorem}[Exact Kramers-Wannier duality for the Trotter spacetime theory]
	\label{thm:KW-Trotter-duality}
	For the explicit Trotter weights \eqref{eq:parallel-weights}, \eqref{eq:perp-weights}, the
	fixed-charge duality is normalization-free:
	\begin{equation}
		Z_{P,M}(q_m,q_e;J,K,g,h)
		=
		Z_{\check P,M}^\vee(q_e,q_m;\check J,\check K,\check g,\check h),
		\label{eq:KW-Trotter-fixed-charge-duality}
	\end{equation}
	where the dual couplings are given by \eqref{eq:KW-dual-JK-map} and
	\eqref{eq:KW-dual-gh-map}.
\end{theorem}

\begin{proof}
	Equation \eqref{eq:KW-local-factor-product} shows that the local Fourier prefactors
	produced in the specialization to the Trotter weights cancel exactly against the normalization factors in \eqref{eq:KW-Z-raw}, yielding
	\eqref{eq:KW-Trotter-fixed-charge-duality}.
\end{proof}

\subsection{Duality of the quantum code}
\label{sec:KW-quantum-code}

The couplings \(J,K,g,h\) are precisely the couplings entering the quantum Hamiltonian of
\S\ref{sec:quantum-classical}: \(J\) and $K$ multiply the \(X\)- and $Z$-type stabilizer projectors
\(\hat{\mathcal A}_a\) and \(\hat{\mathcal B}_c\) respectively;  \(g\) and $h$ couple to local \(X\)- and $Z$-type source terms respectively. The dual coupling map
\eqref{eq:KW-dual-JK-map}-\eqref{eq:KW-dual-gh-map} therefore exchanges \(X\)-type and
\(Z\)-type operators under the primal-dual correspondence.

On the source-free stabilizer part of the Hamiltonian, this is already visible from
\eqref{eq:KW-dual-JK-map}: the coefficient of a dual \(X\)-type stabilizer comes from a
primal \(Z\)-type stabilizer, and the coefficient of a dual \(Z\)-type stabilizer comes
from a primal \(X\)-type stabilizer. Thus the \(P\)-form code on \(\Lambda\) is dual to
the \((d-P)\)-form code on \(\Lambda^\vee\), with electric and magnetic stabilizers
interchanged.

The same statement extends to the local Weyl operators themselves. If
\(u^\vee\in C_{\check P}(\Lambda^\vee)\) and \(u=\theta_\Lambda(u^\vee)\in C_P(\Lambda)\),
then the source-mode map is
\begin{equation}
	\hat Z_u^{\,n}
	\quad\longleftrightarrow\quad
	\check X_{u^\vee}^{\,(-1)^{d-P}n},
	\qquad
	\hat X_u^{\,n}
	\quad\longleftrightarrow\quad
	\check Z_{u^\vee}^{\,(-1)^P n}.
	\label{eq:KW-local-XZ-swap}
\end{equation}
Apart from these orientation signs, duality simply exchanges \(X\) and \(Z\).

At the level of logical operators, the exchange is equally concrete. Let
\begin{equation}
\nu\in Z_P(\Lambda),
\qquad
\mu^\vee\in Z_{d-P}(\Lambda^\vee)
\end{equation}
label the spatial Wilson and 't Hooft operators
\(\hat{\mathcal W}_\nu\) and \(\hat{\mathcal T}_{\mu^\vee}\). Under the primal-dual
correspondence,
\begin{equation}
	\hat{\mathcal W}_\nu
	\quad\longleftrightarrow\quad
	\check{\mathcal T}_{\nu},
	\qquad
	\hat{\mathcal T}_{\mu^\vee}
	\quad\longleftrightarrow\quad
	\check{\mathcal W}_{\mu^\vee}.
	\label{eq:KW-logical-operator-swap}
\end{equation}
Thus Kramers-Wannier duality exchanges the electric and magnetic logical operator algebras
of the code.

If the spatial lattice is self-dual, so that \(\Lambda\) is identified with \(\Lambda^\vee\),
the duality acts internally on the family of quantum codes. The self-dual locus is then the
fixed-point set of the coupling map:
\begin{equation}
	P=d-P,
	\qquad
	J_a=K_{\theta_\Lambda^{-1}(a)},
	\qquad
	g_u^{(n)}=h_{\theta_\Lambda^{-1}(u)}^{((-1)^{d-P}n)}.
	\label{eq:KW-self-dual-general}
\end{equation}
In the homogeneous case this reduces to
\begin{equation}
	P=\frac d2,
	\qquad
	J=K,
	\qquad
	g^{(n)}=h^{((-1)^P n)}.
	\label{eq:KW-self-dual-homogeneous}
\end{equation}
If, in addition, the source coefficients are parity-symmetric,
\begin{equation}
g^{(n)}=g^{(-n)},
\qquad
h^{(n)}=h^{(-n)},
\end{equation}
the sign in \eqref{eq:KW-self-dual-homogeneous} becomes irrelevant, and the self-dual locus
takes the simpler form
\begin{equation}
	P=\frac d2,
	\qquad
	J=K,
	\qquad
	g=h.
	\label{eq:KW-self-dual-simple}
\end{equation}
On this locus the quantum Hamiltonian is invariant under the combined operation of
primal-dual identification and exchange of \(X\)-type and \(Z\)-type operators. In
particular, the electric and magnetic logical sectors are related by an exact symmetry of
the code.

This completes the derivation of the exact higher-form Kramers-Wannier symmetry of the
spacetime theory. In the fixed-charge basis, duality exchanges electric and magnetic
backgrounds. For the explicit Trotter spacetime theory, the fixed-charge duality acts without any residual
normalization factor. At the level of the quantum code, it interchanges \(X\)-type and
\(Z\)-type stabilizers, local Weyl operators, and Wilson/'t Hooft logical observables under
the primal-dual correspondence. The importance of this symmetry is twofold. First, it
shows that the full spacetime theory space carries an exact \(\mathbb Z_2\)-involution.
Second, it identifies special slices on which one expects sharper structure than is visible
in the generic perturbative analysis. The next section studies precisely such slices: the
exact source-free gauge theory-subspace and its plaquette random-cluster companion, where stronger
nonperturbative results can be imported.

\section{Gauge/PRCM specializations}
\label{sec:specializations}

The background-dependent sectorwise polymer gas of
Sections~\ref{sec:fixed-charge-polymer-gas}-\ref{sec:low-activity} allows
arbitrary local weights, complex electric-magnetic linking phases, and no
positivity or monotonicity assumptions. Its rigorous control is therefore
perturbative. There is, however, a distinguished exact positive subspace of the
full spacetime weight space on which the theory reduces to a source-free
\(\mathbb Z_N\) lattice gauge model. This places the present construction in the
broader landscape of lattice gauge theory, duality, and phase structure
initiated by Wegner and developed further in
\cite{Wegner:1971app,FradkinShenker1979,BanksRabinovici1979}.

For prime \(N\), this gauge theory is exactly coupled to the plaquette
random-cluster model (PRCM). On closed lattices with nontrivial ambient
homology, the most relevant sharp result presently available is the
middle-dimensional self-dual toric transition of
Duncan-Schweinhart \cite{DuncanSchweinhartTopological}.

\subsection{The exact gauge/PRCM specialization}
\label{sec:specializations-gauge-prcm}

Let \(X\) be a finite oriented cell complex of dimension \(n\). For arbitrary
local weights \(W_c:\mathbb Z_N\to \mathbb C\) on \((P+1)\)-cells and
\(V_u:\mathbb Z_N\to \mathbb C\) on \(P\)-cells, define
\begin{equation}
	\mathcal Z_X(W,V)
	:=
	\sum_{\phi\in \Omega^P(X)}
	\prod_{c\in C_{P+1}(X)} W_c\left((d_X\phi)_c\right)
	\prod_{u\in C_P(X)} V_u(\phi_u).
	\label{eq:general-X}
\end{equation}
In particular, for \(X=\underline\Lambda\) and \(q_m=q_e=0\), this is the
neutral-background spacetime partition function \eqref{eq:Z-0}.

The exact source-free gauge theory-subspace is obtained by removing the \(P\)-cell
weights and retaining only a flatness weight on \((P+1)\)-cells. Writing
\(\beta_{\mathrm g}\ge 0\) for the gauge coupling, set
\begin{equation}
	V_u(x)=1
	\quad (u\in C_P(X)),
	\qquad
	W_c(x)=1+\left(e^{\beta_{\mathrm g}}-1\right)\delta_N(x)
	\quad (c\in C_{P+1}(X)).
	\label{eq:gauge-weights}
\end{equation}
Then \eqref{eq:general-X} becomes
\begin{align}
	\mathcal Z_X^{\mathrm{gauge}}(\beta_{\mathrm g})
	&:=
	\sum_{\phi\in \Omega^P(X)}
	\prod_{c\in C_{P+1}(X)}
	\Bigl[
	1+\left(e^{\beta_{\mathrm g}}-1\right)\delta_N\left((d_X\phi)_c\right)
	\Bigr]
	\nonumber\\
	&=
	\sum_{\phi\in \Omega^P(X)}
	\exp\Biggl(
	\beta_{\mathrm g}
	\sum_{c\in C_{P+1}(X)}
	\delta_N\left((d_X\phi)_c\right)
	\Biggr).
	\label{eq:gauge-partition}
\end{align}
This is the \(P\)-form \(N\)-state Potts lattice gauge theory on \(X\), written
in additive \(\mathbb Z_N\) notation
\cite{Wegner:1971app,FradkinShenker1979,BanksRabinovici1979,
	HiraokaShirai2016,DuncanSchweinhartTopological,
	DuncanSchweinhartDeconfinement}. In the absence of sources the action is manifestly invariant under a gauge transformation: $\phi \mapsto \phi + d_X\psi$ for $\psi \in \Omega^{P-1}(X)$.

\subsubsection{Potts lattice gauge theory from the Trotter weights}
It is useful to locate this gauge theory-subspace precisely inside the anisotropic
Trotter family of \S\ref{sec:quantum-classical}. Fix the Trotter number \(M\)
and take \(X=\underline\Lambda\). The horizontal \(P\)-cell weights
\eqref{eq:V-parallel} become trivial exactly when
\begin{equation}
	h_b^{(n)}=0
	\qquad
	\text{for all }b\in C_P(\Lambda),\ n\in\mathbb Z_N,
\end{equation}
so that \(V_{b(i)}^{\parallel}(x)\equiv 1\). The vertical \(P\)-cell weights are
\eqref{eq:V-perp},
\begin{equation}
	V_{\underline a(i)}^{\perp}(x)
	=
	\delta_N(x)+\frac{e^{\beta J_a/M}-1}{N}.
\end{equation}
After dividing by the field-independent zero-mode factor
\(V_{\underline a(i)}^{\perp}(0)\), one has
\begin{equation}
	\lim_{J_a\to\infty}
	\frac{V_{\underline a(i)}^{\perp}(x)}{V_{\underline a(i)}^{\perp}(0)}
	\to 1
	\qquad
	\text{uniformly in }x\in\mathbb Z_N.
\end{equation}
Thus the \(P\)-cell weights are removed by the exact condition \(h=0\) together
with the strong-coupling limit \(J\to\infty\).

For the \((P+1)\)-cell weights, the horizontal part \eqref{eq:W-parallel}
already has gauge form:
\begin{equation}
	W_{c(i)}^{\parallel}(x)
	=
	1+\left(e^{\beta K_c/M}-1\right)\delta_N(x).
\end{equation}
Hence a prescribed horizontal gauge coupling
\(\beta_{c(i)}^{\parallel}\) is realized by choosing
\begin{equation}
	K_c=\frac{M}{\beta}\,\beta_{c(i)}^{\parallel}.
	\label{eq:K-choice}
\end{equation}
The vertical \((P+1)\)-cell weights are \eqref{eq:W-perp},
\begin{equation}
	W_{\underline b(i)}^{\perp}(x)
	=
	\frac1N
	\sum_{j\in\mathbb Z_N}
	\exp\left(
	\frac{\beta}{M}\sum_{n\in\mathbb Z_N}g_b^{(n)}\omega^{nj}
	\right)
	\omega^{(-1)^P jx}.
\end{equation}
To realize a prescribed vertical gauge coupling
\(\beta_{\underline b(i)}^{\perp}\), define
\begin{equation}
	F_{\underline b(i)}(j)
	:=
	\left(e^{\beta_{\underline b(i)}^{\perp}}-1\right)+N\delta_N(j),
	\qquad
	j\in\mathbb Z_N,
\end{equation}
and
\begin{equation}
	\Gamma_{\underline b(i)}(j)
	:=
	\log F_{\underline b(i)}(j)
	-
	\frac1N\sum_{k\in\mathbb Z_N}\log F_{\underline b(i)}(k).
	\label{eq:Gamma-def}
\end{equation}
Now define $g$ to essentially be the Fourier transform of $\Gamma$:
\begin{equation}
	g_b^{(n)}
	:=
	\frac{M}{\beta N}
	\sum_{j\in\mathbb Z_N}
	\Gamma_{\underline b(i)}(j)\,\omega^{-nj}.
	\label{eq:g-choice}
\end{equation}
Then \(g_b^{(0)}=0\), and discrete Fourier inversion gives
\begin{equation}
	\exp\left(
	\frac{\beta}{M}\sum_{n\in\mathbb Z_N}g_b^{(n)}\omega^{nj}
	\right)
	\propto
	F_{\underline b(i)}(j),
\end{equation}
with proportionality factor independent of \(j\). Consequently,
\begin{align}
	W_{\underline b(i)}^{\perp}(x)
	&\propto
	\frac1N
	\sum_{j\in\mathbb Z_N}
	F_{\underline b(i)}(j)\,\omega^{(-1)^P jx}
	=
	1+\left(e^{\beta_{\underline b(i)}^{\perp}}-1\right)\delta_N(x).
	\label{eq:Wperp-gauge-form}
\end{align}
Therefore, after the zero-mode normalization of the \(P\)-cell weights and the
limit \(J\to\infty\), the Trotter family reduces to an anisotropic source-free
gauge theory on \(\underline\Lambda\) with horizontal couplings
\(\beta_{c(i)}^{\parallel}\) and vertical couplings
\(\beta_{\underline b(i)}^{\perp}\). The isotropic gauge theory-subspace is the further
specialization
\begin{equation}
	\beta_{c(i)}^{\parallel}
	=
	\beta_{\underline b(i)}^{\perp}
	=
	\beta_{\mathrm g}.
	\label{eq:isotropic-gauge}
\end{equation}

\paragraph{Strong coupling limit of the gauge theory as the code limit.}
This also clarifies the relation with the code limit
\eqref{eq:codeLimit}. The source-free gauge theory is an exact finite-coupling
subspace of the general spacetime weight space \eqref{eq:general-X}. Inside
the explicit anisotropic Trotter family, the same gauge theory is reached after
the zero-mode normalization of the \(P\)-cell weights, the exact condition
\(h=0\), the limiting condition \(J\to\infty\), and the choice of \(K\) and
\(g\) in \eqref{eq:K-choice} and \eqref{eq:g-choice}. One then reaches the
code limit by sending the resulting gauge couplings to infinity:
\begin{equation}
	\beta_{c(i)}^{\parallel}\to\infty,
	\qquad
	\beta_{\underline b(i)}^{\perp}\to\infty.
\end{equation}
In that further limit,
\begin{equation}
	\frac{W_c(x)}{W_c(0)}\to \delta_N(x),
	\qquad
	\frac{V_u(x)}{V_u(0)}\to 1,
\end{equation}
which is precisely the flat-field regime of
\S\ref{sec:fixed-charge-code-limit}. 

\subsubsection{Gauge theory coupled to PRCM}
The source-free Potts gauge theory \eqref{eq:gauge-partition} can be coupled to
the plaquette random-cluster model (PRCM) in the standard Edwards-Sokal sense \cite{EdwardsSokal1988}: both
measures arise as marginals of a single joint measure. This is exactly the coupling derived in
\cite[Prop.~20]{DuncanSchweinhartTopological}, we record the joint measure
and its two marginals in the present notation.

Assume \(N\) is prime, so that \(\mathbb Z_N=\mathbb F_N\), and remain on the
exact gauge theory-subspace \eqref{eq:gauge-weights}. Introduce the PRCM probability parameter
\begin{equation}
	p:=1-e^{-\beta_{\mathrm g}}.
	\label{eq:p-from-beta}
\end{equation}
Let \(\omega(c)\in\{0,1\}\) be an auxiliary occupation variable on each
\((P+1)\)-cell \(c\in C_{P+1}(X)\). We then define a joint measure on
\begin{equation}
	(\phi,\omega)\in \Omega^P(X)\times \{0,1\}^{C_{P+1}(X)}
\end{equation}
by
\begin{equation}
	\kappa(\phi,\omega)
	\propto
	\prod_{c\in C_{P+1}(X)}
	\Bigl[
	(1-p)\,\delta_N(\omega(c))
	+
	p\,\delta_N(\omega(c)-1)\,\delta_N\left((d_X\phi)_c\right)
	\Bigr].
	\label{eq:joint-gauge-prcm}
\end{equation}

Summing over the occupation variables gives
\begin{equation}
\begin{aligned}
	\sum_{\omega}\kappa(\phi,\omega)
	\propto&\;
	\prod_{c\in C_{P+1}(X)}
	\Bigl[(1-p)+p\,\delta_N\left((d_X\phi)_c\right)\Bigr]
	\\=&\;
	(1-p)^{|C_{P+1}(X)|}
	\prod_{c\in C_{P+1}(X)}
	\Bigl[
	1+\frac{p}{1-p}\,\delta_N\left((d_X\phi)_c\right)
	\Bigr].
\end{aligned}
\end{equation}
Since \(p/(1-p)=e^{\beta_{\mathrm g}}-1\), the \(\phi\)-marginal is, up to the
field-independent prefactor \((1-p)^{|C_{P+1}(X)|}\), exactly the source-free
Potts gauge measure \eqref{eq:gauge-partition}.

Conversely, for fixed \(\omega\), let \(Y=Y(\omega)\) be the subcomplex of \(X\)
obtained by taking the full \(P\)-skeleton of \(X\) and adjoining precisely
those \((P+1)\)-cells \(c\) for which \(\omega(c)=1\). Then
\(C_{P+1}(Y)=\{c\in C_{P+1}(X):\omega(c)=1\}\), and summing over the gauge field
gives
\begin{equation}
	\sum_{\phi}\kappa(\phi,\omega)
	\propto
	p^{|C_{P+1}(Y)|}(1-p)^{|C_{P+1}(X)|-|C_{P+1}(Y)|}\,
	|Z^P(Y;\mathbb F_N)|.
\end{equation}
Because \(Y\) contains the full \(P\)-skeleton of \(X\), the group
\(B^P(Y;\mathbb F_N)\) is independent of the choice of open \((P+1)\)-cells, so
\begin{equation}
	|Z^P(Y;\mathbb F_N)|
	=
	|H^P(Y;\mathbb F_N)|\,|B^P(Y;\mathbb F_N)|
	\propto
	|H^P(Y;\mathbb F_N)|
	=
	N^{b_P(Y;\mathbb F_N)},
\end{equation}
where $b_P$ is the $P$th Betty number of $Y$:
\begin{equation}
	b_P(Y;\mathbb F_N):=\dim_{\mathbb F_N}H_P(Y;\mathbb F_N).
\end{equation}
Hence the \(\omega\)-marginal is, up to normalization,
\begin{equation}
	\mu_{X,p,N,P+1}(Y)
	\propto
	p^{|C_{P+1}(Y)|}(1-p)^{|C_{P+1}(X)|-|C_{P+1}(Y)|}\,
	N^{b_P(Y;\mathbb F_N)}.
	\label{eq:prcm-measure}
\end{equation}
This is precisely the \((P+1)\)-dimensional plaquette random-cluster measure on
\(X\) \cite{HiraokaShirai2016}.

Thus the source-free Potts gauge theory and the PRCM
are coupled through the joint measure \eqref{eq:joint-gauge-prcm}: the former
is the \(\phi\)-marginal and the latter is the \(Y\)-marginal. In particular,
for prime \(N\), the specialized positive slice of the present spacetime theory
inherits the probabilistic structure of the PRCM.

\subsection{Middle-dimensional toric specialization}
\label{sec:specializations-middle-torus}

Assume that \(N\) is an odd prime and specialize the geometry to the regular
self-dual cubical torus
\begin{equation}
	X=T_L^{\,2(P+1)}.
	\label{eq:middle-torus-geometry}
\end{equation}
Then the companion PRCM has plaquette dimension
\(P+1\), exactly half the ambient dimension.

\subsubsection{Sharp homological transition in the PRCM}

Let \(Y\) denote the random \((P+1)\)-dimensional subcomplex distributed by
the PRCM measure \eqref{eq:prcm-measure}, and let \(\iota:Y\hookrightarrow X\)
be the inclusion inducing the push-forward on homology $\iota_*:H_{P+1}(Y;\mathbb F_N)\to H_{P+1}(X;\mathbb F_N)$. Let $A_L$ and $S_L$ be the events that $\iota_*$ is nonzero and surjective respectively.

\begin{theorem}[Duncan-Schweinhart, {\cite[Thm.~8]{DuncanSchweinhartTopological}}]
	\label{thm:DS-middle-torus}
	If
	\begin{equation}
		p_{\mathrm{sd}}(N):=\frac{\sqrt N}{1+\sqrt N},
		\label{eq:psd}
	\end{equation}
	then, as $L \to \infty$:
	\begin{gather}
		\mu_{X,p,N,P+1}(A_L)\to 0
		\qquad
		\text{for }p<p_{\mathrm{sd}}(N),
		\label{eq:middle-below}
		\\
		\mu_{X,p,N,P+1}(S_L)\to 1
		\qquad
		\text{for }p>p_{\mathrm{sd}}(N),
		\label{eq:middle-above}
	\end{gather}
\end{theorem}

Using \eqref{eq:p-from-beta}, the self-dual random-cluster parameter
\eqref{eq:psd} corresponds to the exact critical gauge coupling
\begin{equation}
	\beta_{\mathrm{sd}}(N)
	=
	-\log\left(1-p_{\mathrm{sd}}(N)\right)
	=
	\log(1+\sqrt N).
	\label{eq:betasd}
\end{equation}

\subsubsection{Transfer to the gauge theory through the joint measure}

The bridge from the PRCM to the source-free Potts gauge theory is furnished by
the joint measure \eqref{eq:joint-gauge-prcm} on the torus \(X\). We write
\(\nu^{\mathrm{gauge}}_{X,\beta_{\mathrm g}}\) for the normalized gauge measure
associated with \eqref{eq:gauge-partition}, and
\(\langle-\rangle^{\mathrm{gauge}}_{X,\beta_{\mathrm g}}\) for expectation with
respect to it. Under the joint measure, one may first sample the PRCM
subcomplex \(Y\) and then sample the gauge field \(\phi\) conditionally on that
choice of \(Y\). The conditional law is exactly what transfers homological
information about \(Y\) into identities for gauge-theoretic Wilson observables.

For a \(P\)-cycle \(\gamma\in Z_P(X)\), define the associated classical Wilson observable by (the same cahracter as in \eqref{eq:Zme})
\begin{equation}
	\mathcal W_\gamma[\phi]
	:=
	\chi\left(\int_\gamma \phi\right).
	\label{eq:classical-Wilson}
\end{equation}
By the exact coupling via the joint measure \(\kappa\) in
\eqref{eq:joint-gauge-prcm}, the conditional law of \(\phi\) given \(Y\) is
uniform on the cocycle space \(Z^P(Y;\mathbb F_N)\); see
\cite[Prop.~21]{DuncanSchweinhartTopological}. Consequently, the finite-volume
Wilson/null-homology correspondence of
\cite[Cor.~29 and Cor.~30]{DuncanSchweinhartTopological}, translated into the
present notation, gives for every \(\gamma,\gamma_1,\gamma_2\in Z_P(X)\),
\begin{align}
	\mathbb E_{\kappa}\!\left[\mathcal W_\gamma \,\middle|\, Y\right]
	&=
	\mathbf 1_{\{
		[\gamma]=0\text{ in }H_P(Y;\mathbb F_N)
		\}},
	\label{eq:toric-conditional-one-point}
	\\
	\mathbb E_{\kappa}\!\left[
	\mathcal W_{\gamma_1}\,\mathcal W_{-\gamma_2}
	\,\middle|\, Y
	\right]
	&=
	\mathbf 1_{\{
		[\gamma_1]=[\gamma_2]\text{ in }H_P(Y;\mathbb F_N)
		\}}.
	\label{eq:toric-conditional-two-point}
\end{align}
Averaging over \(Y\) gives
\begin{align}
	\left\langle \mathcal W_\gamma \right\rangle^{\mathrm{gauge}}_{X,\beta_{\mathrm g}}
	&=
	\mu_{X,p,N,P+1}
	\Bigl(
	[\gamma]=0\text{ in }H_P(Y;\mathbb F_N)
	\Bigr),
	\label{eq:toric-one-point}
	\\
	\left\langle
	\mathcal W_{\gamma_1}\,\mathcal W_{-\gamma_2}
	\right\rangle^{\mathrm{gauge}}_{X,\beta_{\mathrm g}}
	&=
	\mu_{X,p,N,P+1}
	\Bigl(
	[\gamma_1]=[\gamma_2]\text{ in }H_P(Y;\mathbb F_N)
	\Bigr).
	\label{eq:toric-two-point}
\end{align}
These identities are the precise mechanism by which the PRCM transition is
transferred to the gauge theory.

\subsubsection{Gauge-theoretic interpretations of the critical coupling}

For subsets \(S_1,S_2\subset C_\bullet(X)\), write
\begin{equation}
	\mathrm{dist}_X(S_1,S_2)
	:=
	\min\left\{
	\mathrm{dist}_X(v_1,v_2):
	v_i\in C_0(X)\text{ is incident to some cell of }S_i,\ i=1,2
	\right\},
	\label{eq:dist-subsets}
\end{equation}
where \(\mathrm{dist}_X(v_1,v_2)\) is the graph distance in the \(1\)-skeleton of
\(X\). Let \(\mathscr P_L\subset Z_P(X)\) be the finite set of nontrivial toric
\(P\)-cycles supported on codimension-one subtori of \(X\). Define the gauge
event
\begin{equation}
	\mathcal E_L
	:=
	\Bigl\{
	\phi\in\Omega^P(X):
	\exists\,\gamma_1,\gamma_2\in\mathscr P_L
	\text{ with }
	\mathrm{dist}_X(\Supp(\gamma_1),\Supp(\gamma_2))\ge L/2-1
	\text{ and }
	\mathcal W_{\gamma_1}[\phi]=\mathcal W_{\gamma_2}[\phi]
	\Bigr\}.
	\label{eq:EL}
\end{equation}

\begin{prop}[Gauge-theoretic consequences of the toric transition]
	\label{prop:toric-gauge-consequences}
	For \(X=T_L^{\,2(P+1)}\) and \(N\) an odd prime:
	\begin{enumerate}
		\item
		If \(\Gamma\in Z_P(X)\) represents a nonzero class in
		\(H_P(X;\mathbb F_N)\), then
		\begin{equation}
			\left\langle \mathcal W_\Gamma \right\rangle^{\mathrm{gauge}}_{X,\beta_{\mathrm g}}
			=
			0
			\qquad
			\text{for every }L\text{ and every }\beta_{\mathrm g}\ge 0.
			\label{eq:single-toric-loop-zero}
		\end{equation}
		
		\item
		For the event \(\mathcal E_L\) defined in \eqref{eq:EL},
		\begin{equation}
			\nu^{\mathrm{gauge}}_{X,\beta_{\mathrm g}}(\mathcal E_L)\to 1
			\qquad
			\text{for }\beta_{\mathrm g}>\beta_{\mathrm{sd}}(N).
			\label{eq:EL-high}
		\end{equation}
	\end{enumerate}
\end{prop}

\begin{proof}
	Equation \eqref{eq:single-toric-loop-zero} follows immediately from
	\eqref{eq:toric-one-point}. If \([\Gamma]\neq 0\) in \(H_P(X;\mathbb F_N)\),
	then \([\Gamma]\) cannot vanish in \(H_P(Y;\mathbb F_N)\) for any subcomplex
	\(Y\subset X\): if \(\Gamma=\partial \rho\) in \(Y\), then the same equality
	holds in \(X\). Hence the event on the right-hand side of
	\eqref{eq:toric-one-point} is impossible.
	
	For \eqref{eq:EL-high}, \eqref{eq:middle-above} implies
	\(\mu_{X,p,N,P+1}(A_L)\to 1\), since \(S_L\subset A_L\). On the event \(A_L\),
	\cite[Prop.~36]{DuncanSchweinhartTopological} produces two cycles
	\(\gamma_1,\gamma_2\in\mathscr P_L\) such that
	\begin{equation}
		\mathrm{dist}_X(\Supp(\gamma_1),\Supp(\gamma_2))\ge L/2-1,
		\qquad
		[\gamma_1]=[\gamma_2]\text{ in }H_P(Y;\mathbb F_N).
		\label{eq:toric-pair-from-giant-cycle}
	\end{equation}
	By \eqref{eq:toric-conditional-two-point},
	\begin{equation}
		\mathbb E_{\kappa}\!\left[
		\mathcal W_{\gamma_1}\,\mathcal W_{-\gamma_2}
		\,\middle|\, Y
		\right]
		=1.
	\end{equation}
	Since \(\mathcal W_{\gamma_1}\mathcal W_{-\gamma_2}\) takes values in the
	group of \(N\)-th roots of unity, expectation \(1\) forces
	\begin{equation}
		\mathcal W_{\gamma_1}[\phi]=\mathcal W_{\gamma_2}[\phi]
		\qquad
		\text{for \(\kappa\)-almost every \(\phi\), conditioned on }Y.
		\label{eq:toric-conditional-locking}
	\end{equation}
	Hence \(\kappa(\mathcal E_L\mid Y)=1\) on \(A_L\), and therefore
	\begin{equation}
		\nu^{\mathrm{gauge}}_{X,\beta_{\mathrm g}}(\mathcal E_L)
		=
		\kappa(\mathcal E_L)
		\ge
		\mu_{X,p,N,P+1}(A_L).
		\label{eq:EL-lower-bound}
	\end{equation}
	Using \(p=1-e^{-\beta_{\mathrm g}}\), the convergence
	\eqref{eq:EL-high} follows from \eqref{eq:middle-above} and
	\eqref{eq:betasd}.
\end{proof}

\paragraph{Wilson-operator interpretation.}
Proposition~\ref{prop:toric-gauge-consequences} shows that the toric transition
is not detected by a single nontrivial Wilson one-point function: every fixed
nontrivial toric Wilson operator has zero expectation at all couplings. The
critical coupling \(\beta_{\mathrm{sd}}(N)\) instead marks the onset of
macroscopic value-locking. Above \(\beta_{\mathrm{sd}}(N)\), the sampled gauge
field exhibits with probability tending to one a pair of macroscopically
separated nontrivial toric \(P\)-cycles whose Wilson evaluations coincide.
Below \(\beta_{\mathrm{sd}}(N)\), Theorem~\ref{thm:DS-middle-torus} gives
\eqref{eq:middle-below}, so the companion PRCM almost surely carries no
nontrivial ambient middle-dimensional homology, and the specific PRCM
mechanism producing the locking event \(\mathcal E_L\) is absent with high
probability.

\paragraph{Sampling interpretation.}
The same critical coupling also governs the geometry of the finite-volume
plaquette Swendsen-Wang update \cite{SwendsenWang87} associated with the joint measure
\eqref{eq:joint-gauge-prcm}. Starting from a gauge field \(\phi\), one first
samples \(Y\) from the conditional law \(\kappa(\,\cdot\,\mid \phi)\), and then
samples a new gauge field \(\phi'\) from \(\kappa(\,\cdot\,\mid Y)\). In the
present toric middle-dimensional setting, Theorem~\ref{thm:DS-middle-torus}
therefore yields a sharp crossover in the sampler. Below
\(\beta_{\mathrm{sd}}(N)\), the sampled subcomplex \(Y\) almost surely has
trivial ambient middle-dimensional homology, so the resampling step remains
confined to homologically trivial toric data. Above
\(\beta_{\mathrm{sd}}(N)\), the resampling step accesses nontrivial toric
sectors with asymptotically positive probability; in the formulation of
\cite[\S5.1]{DuncanSchweinhartTopological}, that probability tends to
\begin{equation}
	1-N^{-\binom{2(P+1)}{P+1}}.
	\label{eq:nonlocal-move-probability}
\end{equation}

Because the middle-dimensional toric slice is self-dual, the
Kramers-Wannier operator map \eqref{eq:KW-logical-operator-swap} exchanges
electric and magnetic toric operators. The same critical coupling may therefore
be read, after duality, as the corresponding sharp statement for toric 't Hooft
operators and for the appearance of nontrivial toric sectors in the dual
sampling dynamics.

\section{Conclusion and outlook}
\label{sec:conclusion}

We have developed an exact spacetime framework for finite-temperature
\(\mathbb Z_N\) homological codes.  The thermal trace, decorated by Wilson and
't Hooft insertions and Euclidean twists, is rewritten at fixed Trotter number
as a family of spacetime partition functions labeled by electric and magnetic
backgrounds.  These partition functions admit an exact reformulation as a gas of closed
electric and magnetic defects, or equivalently as a two-species polymer gas of
connected closed polymers subject to homological constraints and linking
interactions.  The same framework also carries an exact higher-form
Kramers-Wannier duality exchanging electric and magnetic data, and it contains
as a positive specialization the source-free Potts lattice gauge theory and its
exact coupling to the PRCM.

The main point is not only that these reformulations exist, but that they
provide a workable platform for analysis.  In the present paper, we used this
platform in two ways.  First, by comparing the complex polymer gas to positive
same-species hard-core majorants, we obtained a rigorous low-activity regime in
which large connected polymers, and in particular homologically nontrivial
ones, are exponentially suppressed.  Second, on the gauge/PRCM slice, we
imported sharp nonperturbative information from the PRCM literature, including
the middle-dimensional toric transition and its gauge-theoretic interpretation.

These results leave several natural directions open. One is to sharpen the polymer analysis on the analytic side, for example by exploiting cancellations in the electric--magnetic linking phases or by developing methods adapted directly to the complex two-species gas. A second direction is combinatorial: to improve the counting estimates for rooted support-connected closed polymers entering the low-activity criterion. Since the present constants \(C_s\) do not fully exploit closedness, sharper closed-polymer counting bounds would directly sharpen the rigorous low-activity region. Another natural direction is to extend the framework from cell decompositions of closed oriented manifolds to cell decompositions of manifolds with boundary, where one expects boundary conditions and possibly additional boundary degrees of freedom to play an essential role. Another direction is to analyze further special slices and observables within the same spacetime framework, beyond the gauge/PRCM example treated here. A further direction is to extend beyond prime \(N\) the genuinely prime-dependent parts of the construction: the generalized linking pairing on arbitrary cycles, the Fourier-resolved sector amplitudes \(\widehat{\mathcal Z}_M(\ell_m,\ell_e;q_m,q_e)\), and the gauge/PRCM coupling, where coefficient dependence and torsion phenomena should play a more substantial role.

More broadly, the framework developed here suggests that finite-temperature
homological codes should be studied not only through Hamiltonians or decoder
models, but also through their exact spacetime sector decomposition.  In that
description, thermal excitations are organized as closed extended defects,
topological sectors are encoded by spacetime homology, and duality acts as an
exact structural symmetry.  From this perspective, questions from lattice gauge
theory, homological statistical mechanics, and the combinatorics of closed
cycles become directly relevant to quantitative questions about thermal
stability.  We expect this viewpoint to provide a useful basis for further work
at the interface of topological quantum memory, lattice gauge theory, polymer
expansions, homological statistical mechanics, and the combinatorics of
high-dimensional closed defects.

\appendix

\section{Fourier conventions on \texorpdfstring{$\mathbb Z_N$}{Z_N}}
\label{app:fourier-ZN}

We fix once and for all the primitive \(N\)-th root of unity
\begin{equation}
	\omega := e^{2\pi i/N}.
\end{equation}
The Pontryagin dual \(\widehat{\mathbb Z}_N := \Hom(\mathbb Z_N, U(1))\) is canonically isomorphic to \(\mathbb Z_N\) after this choice of \(\omega\). For $n=0,\cdots,N-1$, define the elements of $\widehat{\mathbb Z}_N$, the characters
\begin{equation}
	\chi_n: \mathbb Z_N \to U(1), \qquad \chi_n(x) := \omega^{xn}. \label{eq:ZNcharacter}
\end{equation}
We also define the periodic Kronecker delta
\begin{equation}
	\delta_N(x) := \frac{1}{N}\sum_{n=0}^{N-1}\omega^{nx} =
	\begin{cases}
		1, & x\equiv 0 \pmod N,\\
		0, & x\not\equiv 0 \pmod N.
	\end{cases}
	\label{eq:deltaN}
\end{equation}
From this one immediately obtains the character orthogonality relations
\begin{align}
	\frac{1}{N}\sum_{n\in \widehat{\mathbb Z}_N}\chi_n(x)\chi_{-n}(y)
	&=
	\delta_N(x-y),
	\label{eq:character-orthogonality-x}
	\\
	\frac{1}{N}\sum_{x\in \mathbb Z_N}\chi_m(x)\chi_n(-x)
	&=
	\delta_N(m-n).
	\label{eq:character-orthogonality-n}
\end{align}

The characters \(\{\chi_n\}_{n\in \widehat{\mathbb Z}_N}\) form a basis of the vector space
\(\mathrm{Fun}(\mathbb Z_N,\mathbb C)\) of complex-valued functions on \(\mathbb Z_N\).
We use the following symmetric normalization for the discrete Fourier transform:
\begin{subequations}
	\begin{align}
		f(x)
		&=
		\frac{1}{\sqrt N}\sum_{n\in \widehat{\mathbb Z}_N}\hat f(n)\,\chi_n(x),
		\label{eq:FT}
		\\
		\hat f(n)
		&=
		\frac{1}{\sqrt N}\sum_{x\in \mathbb Z_N}f(x)\,\chi_n(-x).
		\label{eq:invFT}
	\end{align}
\end{subequations}
Orthogonality of the characters implies that these formulas define mutually inverse linear isomorphisms between $\mathrm{Fun}(\mathbb Z_N,\mathbb C)$ and $\mathrm{Fun}(\widehat{\mathbb Z}_N,\mathbb C)$.

\section{Lattice geometry, duality, and linking}
\label{app:geometry}

This appendix fixes the geometric and algebraic-topological conventions used throughout the
paper. All chain groups, cochain groups, homology groups, and cohomology groups are taken
with coefficients in $\mathbb Z_N$. In \S\ref{app:linking}, in order to define generalized linking numbers for arbitrary cycles, we shall assume $N$ to be prime so that all these groups are $\mathbb F_N$-vector spaces.

We recall that the dimension of the spatial lattice $\Lambda$ is $d$, that the code degrees of freedom live on the \(P\)-cells of \(\Lambda\), and that the spacetime complex is
\begin{equation}
	\underline\Lambda:=\Lambda\times S^1_M .
	\label{eq:spacetime}
\end{equation}
Here \(S^1_M=\mathbb Z/M\mathbb Z\) is the cubical decomposition of the circle with
vertices \(i=0,\dots,M-1\) and oriented edges \([i,i+1]\).

\subsection{Chains, cochains, and the spacetime complex}\label{app:geometryBasics}

Let \(\Lambda\) be a finite cell decomposition of a closed oriented \(d\)-manifold.
For each \(p\), let \(C_p(\Lambda)\) denote the set of oriented \(p\)-cells of \(\Lambda\),
with the convention
\begin{equation}
	C_{-1}(\Lambda)=C_{d+1}(\Lambda)=\varnothing,
	\qquad
	\Omega^{-1}(\Lambda)=\Omega^{d+1}(\Lambda)=0.
\end{equation}

For an oriented \(p\)-cell \(b\) and an oriented \((p-1)\)-cell \(a\), define the
incidence number
\begin{equation}
	\varepsilon(b,a):=
	\begin{cases}
		+1, & a\subset \partial b \text{ and the orientation of } a
		\text{ agrees with that induced from } \partial b,\\
		-1, & a\subset \partial b \text{ and the orientation of } a
		\text{ is opposite to the induced one},\\
		0, & a\not\subset \partial b.
	\end{cases}
	\label{eq:incidence}
\end{equation}

The group of \(p\)-chains is free $\mathbb Z_N$-module generated by the $p$-cells:
\begin{equation}
	\mathcal C_p(\Lambda):=\mathbb Z_N[C_p(\Lambda)],
	\label{eq:chain-group}
\end{equation}
and the boundary operator is defined as
\begin{equation}
	\partial:\mathcal C_{p+1}(\Lambda)\to \mathcal C_p(\Lambda),
	\qquad
	\partial c:=\sum_{b\in C_p(\Lambda)}\varepsilon(c,b)\,b.
	\label{eq:boundary}
\end{equation}
The incidence numbers satisfy
\begin{equation}
	\sum_{b\in C_p(\Lambda)}\varepsilon(c,b)\,\varepsilon(b,a)=0
	\qquad
	\text{for all } c\in C_{p+1}(\Lambda),\ a\in C_{p-1}(\Lambda),
	\label{eq:incidence-identity}
\end{equation}
hence \(\partial^2=0\).

The space of \(p\)-cochains is
\begin{equation}
	\Omega^p(\Lambda):=\Hom(\mathcal C_p(\Lambda),\mathbb Z_N)
	\cong \mathbb Z_N^{C_p(\Lambda)}.
	\label{eq:cochain-group}
\end{equation}
If \(f\in \Omega^p(\Lambda)\) and $\nu=\sum_{b\in C_p(\Lambda)}\nu_b\,b\in \mathcal C_p(\Lambda)$,
we write
\begin{equation}
	\int_\nu f:=f(\nu)=\sum_{b\in C_p(\Lambda)}\nu_b\,f_b,
	\qquad
	f_b:=f(b)\in \mathbb Z_N.
	\label{eq:pairing}
\end{equation}

The coboundary operator
\begin{equation}
	d:\Omega^p(\Lambda)\to \Omega^{p+1}(\Lambda)
\end{equation}
is defined to satisfy the Stokes formula with respect to \(\partial\):
\begin{equation}
	\int_c df=\int_{\partial c} f
	\qquad
	\text{for every } c\in C_{p+1}(\Lambda).
	\label{eq:d-def}
\end{equation}
In components,
\begin{equation}
	(df)_c=\sum_{b\in C_p(\Lambda)}\varepsilon(c,b)\,f_b.
	\label{eq:d-components}
\end{equation}
Since \(\partial^2=0\), we also have \(d^2=0\).

A \(p\)-cell of the spacetime lattice \(\underline\Lambda\) is either horizontal,
\begin{equation}
	b(i):=b\times\{i\},
	\qquad
	b\in C_p(\Lambda),
	\label{eq:horizontalCell}
\end{equation}
or vertical,
\begin{equation}
	\underline a(i):=a\times[i,i+1],
	\qquad
	a\in C_{p-1}(\Lambda).
	\label{eq:verticalCell}
\end{equation}
Thus the cells of the spacetime lattice decomposes disjointly into horizontal and vertical subsets:
\begin{equation}
	C_p(\underline\Lambda)
	=
	C_p^{\parallel}(\underline\Lambda) \sqcup C_p^{\perp}(\underline\Lambda).
\end{equation}

With the product orientation,
\begin{equation}
	\partial\left(\underline b(i)\right)
	=
	(\partial b)\times[i,i+1]
	+(-1)^p\,b(i+1)-(-1)^p\,b(i),
	\label{eq:product-boundary}
\end{equation}
for every \(b\in C_p(\Lambda)\). Accordingly, if
\begin{equation}
	\phi\in \Omega^p(\underline\Lambda),
	\qquad
	\phi^\parallel(i)\in \Omega^p(\Lambda),
	\qquad
	\phi^\perp(i)\in \Omega^{p-1}(\Lambda)
\end{equation}
denote its horizontal and vertical components, then the spacetime coboundary
\begin{equation}
	D:\Omega^p(\underline\Lambda)\to \Omega^{p+1}(\underline\Lambda)
\end{equation}
is given by
\begin{subequations}
	\begin{align}
		(D\phi)_{c(i)}
		&=
		\left(d\phi^\parallel(i)\right)_c,
		\label{eq:D-horizontal-components}
		\\
		(D\phi)_{\underline b(i)}
		&=
		\left(d\phi^\perp(i)\right)_b
		+
		(-1)^p\left(\phi^\parallel(i+1)_b-\phi^\parallel(i)_b\right).
		\label{eq:D-vertical-components}
	\end{align}
	\label{eq:D-components}
\end{subequations}

A convenient way to write this is
\begin{equation}
	(D\phi)^\parallel(i)=d\phi^\parallel(i),
	\qquad
	(D\phi)^\perp(i)=d\phi^\perp(i)+(-1)^p\Delta_t\phi^\parallel(i),
\end{equation}
where \((\Delta_t\psi)(i):=\psi(i+1)-\psi(i)\). Then
\begin{equation}
	(D^2\phi)^\parallel(i)=d^2\phi^\parallel(i)=0,
\end{equation}
and, since \(D\phi\) has degree \(p+1\),
\begin{align*}
	(D^2\phi)^\perp(i)
	&=
	d\left((D\phi)^\perp(i)\right)+(-1)^{p+1}\Delta_t\left((D\phi)^\parallel(i)\right)\\
	&=
	d\left(d\phi^\perp(i)+(-1)^p\Delta_t\phi^\parallel(i)\right)
	+(-1)^{p+1}\Delta_t\left(d\phi^\parallel(i)\right)\\
	&=
	(-1)^p d\,\Delta_t\phi^\parallel(i)+(-1)^{p+1}\Delta_t\,d\phi^\parallel(i)=0,
\end{align*}
because \(d^2=0\) and \(d\) commutes with \(\Delta_t\). Hence
\begin{equation}
	D^2=0.
	\label{eq:D2-zero}
\end{equation}

For any finite oriented cell complex \(X\), the abelian group of $p$-cycles $Z_p$, $p$-cocycles $Z^p$, $p$-boundaries $B_p$, and $p$-coboundaries $B^p$ are defined as usual:
\begin{align}
	Z_p(X)
	&:=
	\ker \left(\partial:\mathcal C_p(X)\to \mathcal C_{p-1}(X)\right),
	&
	Z^p(X)
	&:=
	\ker \left(d:\Omega^p(X)\to \Omega^{p+1}(X)\right),
	\nonumber\\
	B_p(X)
	&:=
	\operatorname{im} \left(\partial:\mathcal C_{p+1}(X)\to \mathcal C_p(X)\right),
	&
	B^p(X)
	&:=
	\operatorname{im} \left(d:\Omega^{p-1}(X)\to \Omega^p(X)\right).
	\label{eq:ZB}
\end{align}
The $p$th homology and cohomology groups are defined respectively as
\begin{equation}
	H_p(X):=Z_p(X)/B_p(X),
	\qquad
	H^p(X):=Z^p(X)/B^p(X).
	\label{eq:cohomology}
\end{equation}

\subsection{Chain-cochain duality}

The set \(C_p(X)\) is a finite basis of
\(\mathcal C_p(X)\), hence there is a module isomorphism
\begin{equation}
	\mathcal P_X:\Omega^p(X)\xrightarrow{\sim}\mathcal C_p(X),
	\qquad
	\mathcal P_X(f)=\sum_{b\in C_p(X)} f_b\,b.
	\label{eq:PX}
\end{equation}
Its inverse sends a chain \(\sum_{b\in C_p(X)}\nu_b\,b\) to the cochain with components
\(\{\nu_b\}_{b\in C_p(X)}\).

Using \(\mathcal P_X\), define the transpose differential
\begin{equation}
	d_X^{\mathsf T}
	:=
	\mathcal P_X^{-1}\,\partial_X\,\mathcal P_X
	:
	\Omega^p(X)\to \Omega^{p-1}(X).
	\label{eq:dT-def}
\end{equation}
If \(f\in \Omega^p(X)\), then
\begin{equation}
	(d_X^{\mathsf T}f)_a
	=
	\sum_{b\in C_p(X)} f_b\,\varepsilon_X(b,a),
	\qquad
	a\in C_{p-1}(X).
	\label{eq:dT-components}
\end{equation}

We equip \(\Omega^p(X)\) with the bilinear pairing
\begin{equation}
	\langle f,g\rangle_X
	:=
	\int_{\mathcal P_X(f)} g
	=
	\sum_{b\in C_p(X)} f_b\,g_b,
	\qquad
	f,g\in \Omega^p(X).
	\label{eq:bilinear-pairing}
\end{equation}
With respect to this pairing, \(d_X\) and \(d_X^{\mathsf T}\) are adjoint:
\begin{equation}
	\langle d_X^{\mathsf T}f,g\rangle_X
	=
	\langle f,d_X g\rangle_X,
	\qquad
	f\in \Omega^p(X),\ g\in \Omega^{p-1}(X).
	\label{eq:adjoint-d-dT}
\end{equation}

Now given a finite oriented cell decomposition  \(X\) of a closed oriented
\(n\)-manifold, let \(X^\vee\) denote its dual cell decomposition.
For each \(0\le r\le n\), there is a canonical bijection
\begin{equation}
	\theta_X:C_r(X^\vee)\xrightarrow{\sim} C_{n-r}(X).
	\label{eq:theta}
\end{equation}
We choose orientations on the dual cells so that
\begin{equation}
	\varepsilon_X(c,b)
	=
	\varepsilon_{X^\vee}\left(\theta_X^{-1}(b),\theta_X^{-1}(c)\right),
	\label{eq:dual-incidence}
\end{equation}
whenever \(b\in C_r(X)\) and \(c\in C_{r+1}(X)\). Extending \(\theta_X\) linearly gives
\begin{equation}
	\theta_X:\mathcal C_{n-r}(X^\vee)\xrightarrow{\sim}\mathcal C_r(X).
	\label{eq:theta-linear}
\end{equation}

We shall use two induced maps. First, the cochain-level duality map
\begin{equation}
	\theta_X^\sharp
	:=
	\mathcal P_X^{-1}\,\theta_X\,\mathcal P_{X^\vee}
	:
	\Omega^{n-r}(X^\vee)\xrightarrow{\sim}\Omega^r(X),
	\label{eq:theta-sharp}
\end{equation}
and second, the map from dual chains to primal cochains
\begin{equation}
	\vartheta_X
	:=
	\mathcal P_X^{-1}\circ \theta_X
	:
	\mathcal C_{n-r}(X^\vee)\to \Omega^r(X).
	\label{eq:vartheta}
\end{equation}
These fit into the commutative diagram
\begin{equation}
	\begin{tikzcd}[column sep=large,row sep=large]
		\Omega^{n-r}(X^\vee)
		\arrow[r,"\theta_X^\sharp"]
		\arrow[d,"\mathcal P_{X^\vee}"']
		&
		\Omega^r(X)
		\arrow[d,"\mathcal P_X"]
		\\
		\mathcal C_{n-r}(X^\vee)
		\arrow[r,"\theta_X"']
		\arrow[ur,"\vartheta_X"]
		&
		\mathcal C_r(X).
	\end{tikzcd}
	\label{eq:structural-diagram}
\end{equation}
Equivalently,
\begin{equation}
	\theta_X^\sharp=\vartheta_X\circ \mathcal P_{X^\vee},
	\qquad
	\vartheta_X=\mathcal P_X^{-1}\circ \theta_X.
\end{equation}

The key point is that these maps intertwine the relevant differentials.
\begin{prop}[Intertwining differentials]\label{prop:intertwine}
	Degreewise, one has
	\begin{enumerate}[label=(\roman*)]
		\item \label{prop:vartheta-intertwines} $d_X\circ \vartheta_X = \vartheta_X\circ \partial_{X^\vee}$,
		\item \label{prop:theta-sharp-intertwine1} $\theta_X^\sharp\circ d_{X^\vee}
		=
		d_X^{\mathsf T}\circ \theta_X^\sharp$,
		\item \label{prop:theta-sharp-intertwine2} $d_X\circ \theta_X^\sharp
		=
		\theta_X^\sharp\circ d_{X^\vee}^{\mathsf T}$.
	\end{enumerate}
\end{prop}
\begin{proof}
	\begin{enumerate}
		\item Let
		\(c^\vee\in C_{n-r}(X^\vee)\), so that
		\(\vartheta_X(c^\vee)\in \Omega^r(X)\), and let \(a\in C_{r+1}(X)\). Then
		\begin{align*}
			\left(d_X\vartheta_X(c^\vee)\right)_a
			&=
			\sum_{b\in C_r(X)}
			\varepsilon_X(a,b)\,\vartheta_X(c^\vee)_b
			=
			\varepsilon_X\left(a,\theta_X(c^\vee)\right),
			\\
			\text{and,} \qquad \left(\vartheta_X(\partial c^\vee)\right)_a
			&=
			\varepsilon_{X^\vee}\left(c^\vee,\theta_X^{-1}(a)\right)
			=
			\varepsilon_X\left(a,\theta_X(c^\vee)\right),
		\end{align*}
		where the last equality is precisely the dual-incidence convention
		\eqref{eq:dual-incidence}. This proves Prop.~\ref{prop:intertwine}\ref{prop:vartheta-intertwines}.
		
		\item Let \(\alpha\in \Omega^{n-r}(X^\vee)\). Then $\theta_X^\sharp\alpha\in \Omega^r(X)$.
		For the first identity, evaluate both sides on a cell \(a\in C_{r-1}(X)\). We have
		\begin{align*}
			\left(\theta_X^\sharp d_{X^\vee}\alpha\right)_a
			&=
			\left(d_{X^\vee}\alpha\right)_{\theta_X^{-1}(a)}
			=
			\sum_{c^\vee\in C_{n-r}(X^\vee)}
			\varepsilon_{X^\vee}\left(\theta_X^{-1}(a),c^\vee\right)\,
			\alpha_{c^\vee}.
		\end{align*}
		Writing \(c=\theta_X(c^\vee)\in C_r(X)\), the dual-incidence convention \eqref{eq:dual-incidence} gives
		$
			\varepsilon_{X^\vee}\left(\theta_X^{-1}(a),c^\vee\right)
			=
			\varepsilon_X(c,a)
		$.
		Therefore
		\begin{align}
			\left(\theta_X^\sharp d_{X^\vee}\alpha\right)_a
			&=
			\sum_{c\in C_r(X)}
			\varepsilon_X(c,a)\,
			(\theta_X^\sharp\alpha)_c
			=
			\left(d_X^{\mathsf T}\theta_X^\sharp\alpha\right)_a.
		\end{align}
		This proves Prop.~\ref{prop:intertwine}\ref{prop:theta-sharp-intertwine1}.
		
		\item For the second identity, evaluate both sides on a cell \(a\in C_{r+1}(X)\). We get
		\begin{align}
			\left(d_X\theta_X^\sharp\alpha\right)_a
			&=
			\sum_{b\in C_r(X)}
			\varepsilon_X(a,b)\,
			(\theta_X^\sharp\alpha)_b
			=
			\sum_{b\in C_r(X)}
			\varepsilon_X(a,b)\,
			\alpha_{\theta_X^{-1}(b)}.
		\end{align}
		Using the dual-incidence convention again, $\varepsilon_X(a,b) = \varepsilon_{X^\vee}\left(\theta_X^{-1}(b),\theta_X^{-1}(a)\right)$.
		Thus
		\begin{align*}
			\left(d_X\theta_X^\sharp\alpha\right)_a
			&=
			\sum_{b\in C_r(X)}
			\varepsilon_{X^\vee}\left(\theta_X^{-1}(b),\theta_X^{-1}(a)\right)\,
			\alpha_{\theta_X^{-1}(b)}
			=
			\left(d_{X^\vee}^{\mathsf T}\alpha\right)_{\theta_X^{-1}(a)}
			=
			\left(\theta_X^\sharp d_{X^\vee}^{\mathsf T}\alpha\right)_a.
		\end{align*}
		This proves Prop.~\ref{prop:intertwine}\ref{prop:theta-sharp-intertwine2}.
	\end{enumerate}
\end{proof}

\begin{cor}\label{cor:complexIso}
	$\vartheta$ is an isomorphism of complexes.
\end{cor}
\begin{proof}
	It is a module isomorphism by definition \eqref{eq:vartheta}, since it is a composition of two isomorphisms \eqref{eq:theta-linear} and \eqref{eq:PX}. By Prop.~\ref{prop:intertwine}\ref{prop:vartheta-intertwines} it also preserves the differentials.
\end{proof}

\subsection{Intersection pairing}
\label{sec:intersection}

For each \(0\le r\le n\), define
\begin{equation}
	I_X^{(r)}:
	\mathcal C_{n-r}(X^\vee)\times \mathcal C_r(X)\to \mathbb Z_N
	\label{eq:intersection}
\end{equation}
by
\begin{equation}
	I_X^{(r)}(\Xi^\vee,\Sigma)
	:=
	\int_\Sigma \vartheta_X(\Xi^\vee)
	=
	\int_\Sigma \mathcal P_X^{-1} \left(\theta_X(\Xi^\vee)\right).
	\label{eq:intersection-def}
\end{equation}
On basis elements,
\begin{equation}
	I_X^{(r)}(c^\vee,b)=\delta_{\theta_X(c^\vee),\,b},
	\qquad
	c^\vee\in C_{n-r}(X^\vee),\quad b\in C_r(X).
	\label{eq:intersection-basis}
\end{equation}
Thus the pairing is nonzero precisely when the cells of the primal and dual lattice correspond to each other under dualization of cell decomposition.

The pairing is compatible with the boundary operators. If
\begin{equation}
	\Xi^\vee\in \mathcal C_{n-r}(X^\vee),
	\qquad
	\Sigma\in \mathcal C_{r+1}(X),
\end{equation}
then
\begin{equation}
	I_X^{(r+1)}(\partial \Xi^\vee,\Sigma)
	=
	I_X^{(r)}(\Xi^\vee,\partial \Sigma).
	\label{eq:intersection-Stokes}
\end{equation}
Indeed,
\begin{equation}
	I_X^{(r+1)}(\partial \Xi^\vee,\Sigma)
	=
	\int_\Sigma \vartheta_X(\partial \Xi^\vee)
	=
	\int_\Sigma d_X\vartheta_X(\Xi^\vee)
	=
	\int_{\partial\Sigma}\vartheta_X(\Xi^\vee)
	=
	I_X^{(r)}(\Xi^\vee,\partial\Sigma),
\end{equation}
where we used Prop.~\ref{prop:intertwine}\ref{prop:vartheta-intertwines} and the defining property of \(d_X\) \eqref{eq:d-def}.

\subsection{Linking pairing}
\label{app:linking}

For two cycles $\nu\in Z_p(X)$ and $\mu^\vee\in Z_{n-p-1}(X^\vee)$,
the pair of homology classes
\begin{equation}
	([\mu^\vee],[\nu])\in H_{n-p-1}(X^\vee)\times H_p(X)
\end{equation}
is canonical, but it contains no cycle-level geometric linking information. The purpose of
the present appendix is to fix a convenient bilinear pairing on cycles such that the pairing reduces to the standard filling-intersection expression as soon as one
of the two cycles bounds.

We assume in this section that $N$ is prime and all chain, cochain, homology, and cohomology groups are finite dimensional vector spaces over 
\begin{equation} 
	\mathbb F_N=\mathbb Z_N.
\end{equation}

\subsubsection{Definition and reduction to ordinary linking}

Since all homology groups are finite-dimensional \(\mathbb F_N\)-vector spaces, the short
exact sequences
\begin{equation}\begin{gathered}
		0\to B_p(X)\to Z_p(X)\to H_p(X)\to 0,
		\\
		0\to B_{n-p-1}(X^\vee)\to Z_{n-p-1}(X^\vee)\to H_{n-p-1}(X^\vee)\to 0
\end{gathered}\end{equation}
split. Likewise, because
\begin{equation}\begin{gathered}
		0\to Z_{p+1}(X)\to \mathcal C_{p+1}(X)\xrightarrow{\partial} B_p(X)\to 0,
		\\
		0\to Z_{n-p}(X^\vee)\to \mathcal C_{n-p}(X^\vee)\xrightarrow{\partial} B_{n-p-1}(X^\vee)\to 0
\end{gathered}\end{equation}
also split, we may choose once and for all:
\begin{enumerate}
	\item linear sections
	\begin{equation}
		s_p:H_p(X)\rightarrow Z_p(X),
		\qquad
		s^\vee_{n-p-1}:H_{n-p-1}(X^\vee)\rightarrow Z_{n-p-1}(X^\vee),
		\label{eq:sections}
	\end{equation}
	\item linear filling operators
	\begin{align}
		K_p:B_p(X)\rightarrow&\; \mathcal C_{p+1}(X),
		&\qquad
		\partial K_p=&\;\mathrm{id}_{B_p(X)},
		\label{eq:Kp}
	\\
		K^\vee_{n-p-1}:B_{n-p-1}(X^\vee)\rightarrow&\; \mathcal C_{n-p}(X^\vee),
		&
		\partial K^\vee_{n-p-1}=&\;\mathrm{id}_{B_{n-p-1}(X^\vee)},
		\label{eq:Kdual}
	\end{align}
	\item a bilinear pairing
	\begin{equation}
		B_H:H_{n-p-1}(X^\vee)\times H_p(X)\rightarrow \mathbb Z_N.
		\label{eq:BH}
	\end{equation}
\end{enumerate}

Using the chosen sections, decompose the cycles into homological and exact parts:
\begin{equation}
\begin{aligned}
	\nu_{\mathrm h}:=&\;s_p([\nu]),
	&\qquad
	\nu_{\mathrm ex}:=&\;\nu-\nu_{\mathrm h}\in B_p(X),
	\\
	\mu^\vee_{\mathrm h}:=&\;s^\vee_{n-p-1}([\mu^\vee]),
	&
	\mu^\vee_{\mathrm ex}:=&\;\mu^\vee-\mu^\vee_{\mathrm h}\in B_{n-p-1}(X^\vee).
\end{aligned}
\end{equation}
Thus
\begin{equation}
	\nu=\nu_{\mathrm h}+\nu_{\mathrm ex},
	\qquad
	\mu^\vee=\mu^\vee_{\mathrm h}+\mu^\vee_{\mathrm ex}.
	\label{eq:total-decomp}
\end{equation}

We then define the generalized linking pairing
\begin{equation}
	\Lk_X:
	Z_{n-p-1}(X^\vee)\times Z_p(X)\rightarrow \mathbb Z_N
	\label{eq:L-map}
\end{equation}
by
\begin{equation}
	\Lk_X(\mu^\vee,\nu)
	:=
	I_X^{(p+1)}\left(\mu^\vee,\,K_p(\nu_{\mathrm ex})\right)
	+
	I_X^{(p)}\left(K^\vee_{n-p-1}(\mu^\vee_{\mathrm ex}),\,\nu_{\mathrm h}\right)
	+
	B_H\left([\mu^\vee],[\nu]\right).
	\label{eq:generalized-linking}
\end{equation}
By construction, \(\Lk_X\) is bilinear.

The three terms in \eqref{eq:generalized-linking} have distinct roles. The first fills
the exact part of the primal cycle and intersects that filling with the full dual cycle.
The second fills the exact part of the dual cycle and intersects that filling with the
chosen homology representative of the primal cycle. The third depends only on the homology
classes and prescribes the residual purely homological coupling. One may rewrite
\eqref{eq:generalized-linking} in the equivalent form
\begin{equation}
	\Lk_X(\mu^\vee,\nu)
	=
	I_X^{(p+1)}\left(\mu^\vee_{\mathrm h},K_p(\nu_{\mathrm ex})\right)
	+
	I_X^{(p)}\left(K^\vee_{n-p-1}(\mu^\vee_{\mathrm ex}),\nu\right)
	+
	B_H\left([\mu^\vee],[\nu]\right).
	\label{eq:generalized-linking-equivalent}
\end{equation}
Indeed,
\begin{equation}
\begin{aligned}
	I_X^{(p+1)}\left(\mu^\vee_{\mathrm ex},K_p(\nu_{\mathrm ex})\right)
	&=
	I_X^{(p+1)}\left(\partial K^\vee_{n-p-1}(\mu^\vee_{\mathrm ex}),K_p(\nu_{\mathrm ex})\right)
	\qquad \text{using \eqref{eq:Kdual}}
	\\
	&=
	I_X^{(p)}\left(K^\vee_{n-p-1}(\mu^\vee_{\mathrm ex}),\partial K_p(\nu_{\mathrm ex})\right)
	\qquad \text{using \eqref{eq:intersection-Stokes}}
	\\
	&=
	I_X^{(p)}\left(K^\vee_{n-p-1}(\mu^\vee_{\mathrm ex}),\nu_{\mathrm ex}\right).
	\qquad \text{using \eqref{eq:Kp}}
\end{aligned}
\end{equation}
\eqref{eq:generalized-linking-equivalent} then follows from \eqref{eq:generalized-linking} by expanding $\mu^\vee$ and $\nu$ as in \eqref{eq:total-decomp}.

The importance of \eqref{eq:generalized-linking} is that it reduces to ordinary linking
as soon as one of the two cycles bounds. If \(\nu\in B_p(X)\), then \([\nu]=0\), hence \(\nu_{\mathrm h}=0\) and
\(\nu_{\mathrm ex}=\nu\), so
\begin{equation}
	\Lk_X(\mu^\vee,\nu)
	=
	I_X^{(p+1)}\left(\mu^\vee,K_p(\nu)\right).
	\label{eq:primal-exact}
\end{equation}
Thus one simply fills \(\nu\) and intersects that filling with \(\mu^\vee\). Similarly, if \(\mu^\vee\in B_{n-p-1}(X^\vee)\), then \([\mu^\vee]=0\), hence
\(\mu^\vee_{\mathrm h}=0\) and \(\mu^\vee_{\mathrm ex}=\mu^\vee\), so
\begin{equation}
	\Lk_X(\mu^\vee,\nu)
	=
	I_X^{(p)}\left(K^\vee_{n-p-1}(\mu^\vee),\nu\right).
	\label{eq:dual-exact}
\end{equation}
Thus one may instead fill \(\mu^\vee\) and intersect that filling with \(\nu\). In particular, if both \(\nu\in B_p(X)\) and \(\mu^\vee\in B_{n-p-1}(X^\vee)\), then
\begin{equation}
	\Lk_X(\mu^\vee,\nu)
	=
	I_X^{(p+1)}\left(\mu^\vee,K_p(\nu)\right)
	=
	I_X^{(p)}\left(K^\vee_{n-p-1}(\mu^\vee),\nu\right),
	\label{eq:boundary-linking}
\end{equation}
which is exactly the ordinary linking/intersection number obtained by filling either cycle
and intersecting with the other.

It is also useful to isolate the genuinely boundary-linking part:
\begin{equation}
	\Lk_X^{\mathrm{ex}}(\mu^\vee,\nu)
	:=
	I_X^{(p+1)}\left(\mu^\vee_{\mathrm ex},K_p(\nu_{\mathrm ex})\right)
	=
	I_X^{(p)}\left(K^\vee_{n-p-1}(\mu^\vee_{\mathrm ex}),\nu_{\mathrm ex}\right).
	\label{eq:exact-linking}
\end{equation}
Then
\begin{equation}
	\Lk_X(\mu^\vee,\nu)
	=
	\Lk_X^{\mathrm{ex}}(\mu^\vee,\nu)
	+
	I_X^{(p+1)}\left(\mu^\vee_{\mathrm h},K_p(\nu_{\mathrm ex})\right)
	+
	I_X^{(p)}\left(K^\vee_{n-p-1}(\mu^\vee_{\mathrm ex}),\nu_{\mathrm h}\right)
	+
	B_H\left([\mu^\vee],[\nu]\right).
	\label{eq:linking-decomposition}
\end{equation}
The first term depends only on the exact parts, the last only on the homology classes, and
the middle two terms measure the coupling between exact data and the chosen homology
representatives.

\subsubsection{Dependence on the auxiliary choices}

The pairing \(\Lk_X\) is not canonical. It depends on the chosen sections,
filling operators, and bilinear homological correction. For the applications in the main
text, the point is that once the filling operators \(K_p\) and \(K^\vee_{n-p-1}\) are
fixed, changing the other conventions changes \(\Lk_X\) only by a bilinear
function of the homology classes.

\begin{prop}
	\label{prop:choice-dependence}
	Fix a single choice of the filling operators \eqref{eq:Kp} and \eqref{eq:Kdual}. Let
	\begin{equation}
		(s_p,s^\vee_{n-p-1},B_H)
		\qquad\text{and}\qquad
		(\widetilde s_p,\widetilde s^\vee_{n-p-1},\widetilde B_H)
	\end{equation}
	be two choices of sections \eqref{eq:sections} and homological correction term
	\eqref{eq:BH}, and let \(\Lk_X\) and \(\widetilde{\Lk}_X\) be the corresponding
	generalized linking pairings \eqref{eq:generalized-linking}. Then:
	\begin{enumerate}[label=(\roman*)]
		\item\label{prop:choice-dependence-item1}
		There exists a bilinear map
		\begin{equation}
			C_H:H_{n-p-1}(X^\vee)\times H_p(X)\to \mathbb Z_N
			\label{eq:DeltaH}
		\end{equation}
		such that
		\begin{equation}
			\widetilde{\Lk}_X(\mu^\vee,\nu)-\Lk_X(\mu^\vee,\nu)
			=
			C_H\left([\mu^\vee],[\nu]\right)
			\label{eq:choice-dependence}
		\end{equation}
		for all
		\(\mu^\vee\in Z_{n-p-1}(X^\vee)\) and \(\nu\in Z_p(X)\). In particular, the two
		pairings agree whenever either \(\mu^\vee\) or \(\nu\) is a boundary.
		
		\item\label{propItem:bdryLink}
		If \(\mu^\vee\in B_{n-p-1}(X^\vee)\) and \(\nu\in B_p(X)\), then
		\begin{equation}
			\Lk_X(\mu^\vee,\nu)
			=
			I_X^{(p+1)}\left(\mu^\vee,\Sigma\right)
			=
			I_X^{(p)}\left(T^\vee,\nu\right),
			\qquad
			\partial\Sigma=\nu,\ \partial T^\vee=\mu^\vee.
			\label{eq:boundary-linking-canonical-merged}
		\end{equation}
		This value is independent of all auxiliary choices used in the definition of
		\(\Lk_X\), including the sections, the homological correction term, and the
		filling chains \(\Sigma\), \(T^\vee\). In particular, the linking number of two
		boundaries is canonical and does not require the assumption that \(N\) be
		prime.
	\end{enumerate}
\end{prop}

\begin{proof}
	Set
	\begin{equation}
		\begin{aligned}
			\alpha&:=\widetilde s_p-s_p:H_p(X)\to Z_p(X),
			\\
			\beta&:=\widetilde s^\vee_{n-p-1}-s^\vee_{n-p-1}:H_{n-p-1}(X^\vee)\to Z_{n-p-1}(X^\vee),
		\end{aligned}
		\label{eq:secDiffs}
	\end{equation}
	and
	\begin{equation}
		C_H^{0}:=\widetilde B_H-B_H:
		H_{n-p-1}(X^\vee)\times H_p(X)\to \mathbb Z_N.
	\end{equation}
	Since both \(s_p\) and \(\widetilde s_p\) are sections of
	\(Z_p(X)\twoheadrightarrow H_p(X)\), the image of \(\alpha\) lies in \(B_p(X)\).
	Likewise, the image of \(\beta\) lies in \(B_{n-p-1}(X^\vee)\).
	
	Writing the homological and exact parts defined by the two choices of sections, we have
	\begin{equation}
		\begin{aligned}
			\widetilde\nu_{\mathrm h}&=\nu_{\mathrm h}+\alpha([\nu]),
			&\qquad
			\widetilde\nu_{\mathrm ex}&=\nu_{\mathrm ex}-\alpha([\nu]),
			\\
			\widetilde\mu^\vee_{\mathrm h}&=\mu^\vee_{\mathrm h}+\beta([\mu^\vee]),
			&
			\widetilde\mu^\vee_{\mathrm ex}&=\mu^\vee_{\mathrm ex}-\beta([\mu^\vee]).
		\end{aligned}
	\end{equation}
	Substituting these into the definition of \(\widetilde{\Lk}_X\), expanding, and using
	\eqref{eq:intersection-Stokes} to move fillings of exact terms from one side to the
	other, we obtain
	\begin{equation}
		\begin{aligned}
			\widetilde{\Lk}_X(\mu^\vee,\nu)-\Lk_X(\mu^\vee,\nu)
			=&\;
			-I_X^{(p+1)}\left(\mu^\vee_{\mathrm h},K_p(\alpha([\nu]))\right)
			-I_X^{(p)}\left(K^\vee_{n-p-1}(\beta([\mu^\vee])),\nu_{\mathrm h}\right)
			\\
			&\;
			-I_X^{(p)}\left(K^\vee_{n-p-1}(\beta([\mu^\vee])),\alpha([\nu])\right)
			+C_H^{0}([\mu^\vee],[\nu]).
		\end{aligned}
		\label{eq:choice-proof}
	\end{equation}
	Now \(\mu^\vee_{\mathrm h}=s^\vee_{n-p-1}([\mu^\vee])\) and
	\(\nu_{\mathrm h}=s_p([\nu])\), so the right-hand side depends only on
	\([\mu^\vee]\) and \([\nu]\), and it is bilinear because all maps appearing in
	\eqref{eq:choice-proof} are linear or bilinear. This proves
	\eqref{eq:choice-dependence} and therefore
	\ref{prop:choice-dependence-item1}. If either \([\mu^\vee]=0\) or \([\nu]=0\),
	then all terms on the right-hand side of \eqref{eq:choice-proof} vanish, so the two
	pairings agree whenever either argument is a boundary.
	
	For \ref{propItem:bdryLink}, if
	\(\mu^\vee\in B_{n-p-1}(X^\vee)\) and \(\nu\in B_p(X)\), then
	\eqref{eq:boundary-linking} gives
	\begin{equation}
		\Lk_X(\mu^\vee,\nu)
		=
		I_X^{(p+1)}\left(\mu^\vee,\Sigma\right)
		=
		I_X^{(p)}\left(T^\vee,\nu\right),
		\qquad
		\partial\Sigma=\nu,\ \partial T^\vee=\mu^\vee.
	\end{equation}
	By \ref{prop:choice-dependence-item1}, this value is independent of the chosen
	sections and homological correction term, because both homology classes vanish.
	It remains only to check independence of the fillings. If
	\(\Sigma,\Sigma'\in\mathcal C_{p+1}(X)\) both satisfy
	\(\partial\Sigma=\partial\Sigma'=\nu\), then \(\Sigma-\Sigma'\in Z_{p+1}(X)\), so
	using \(\mu^\vee=\partial T^\vee\) and \eqref{eq:intersection-Stokes},
	\begin{equation}
		I_X^{(p+1)}\left(\mu^\vee,\Sigma-\Sigma'\right)
		=
		I_X^{(p+1)}\left(\partial T^\vee,\Sigma-\Sigma'\right)
		=
		I_X^{(p)}\left(T^\vee,\partial(\Sigma-\Sigma')\right)
		=
		0.
	\end{equation}
	Hence \(I_X^{(p+1)}(\mu^\vee,\Sigma)\) is independent of \(\Sigma\). The argument for
	\(I_X^{(p)}(T^\vee,\nu)\) is identical. Therefore the common value
	\eqref{eq:boundary-linking-canonical-merged} is canonical. Since its definition uses
	only ordinary filling-intersection of two boundaries, it does not depend on the
	prime-\(N\) hypothesis introduced to construct a generalized pairing on arbitrary
	cycles.
\end{proof}

Proposition~\ref{prop:choice-dependence} identifies the precise source of the ambiguity in
the generalized linking pairing. Changing the cycle-level convention does not alter the
ordinary linking content of boundary interactions; it changes only a purely homological
bilinear correction.

\section{Supporting computations}
\label{app:computeTr}

This appendix gives the finite-\(M\) derivation of the spacetime representation used in
\S\ref{sec:quantum-classical}.  The calculation is completely finite-dimensional.
Its purpose is to justify the local kernels, the spacetime \(P\)-cochain packaging, and the
effect of the spatial and Euclidean twist insertions.

Throughout this appendix we use the \(\hat Z\)-eigenbasis of the Hilbert space
\begin{equation}
	\mathcal H=\bigotimes_{b\in C_P(\Lambda)}\mathcal H_b,
	\qquad
	\ket{s}
	=
	\bigotimes_{b\in C_P(\Lambda)}\ket{s_b},
	\qquad
	s\in \Omega^P(\Lambda)\cong \mathbb Z_N^{C_P(\Lambda)}.
\end{equation}
In this basis,
\begin{equation}
	\hat Z_b\ket{s}=\omega^{s_b}\ket{s},
	\qquad
	\hat X_b\ket{s}
	=
	\ket{s+\delta_b},
	\qquad
	\omega=e^{2\pi i/N},
	\label{eq:Z-basis}
\end{equation}
where \(\delta_b\in \Omega^P(\Lambda)\) is the cochain whose value on \(b\) is \(1\) and
whose value on every other \(P\)-cell is \(0\).

\subsection{Proof of Prop.~\ref{prop:spatial-WT}}\label{app:WTalgebra}
\begin{enumerate}
	\item It is enough to show that \(\hat{\mathcal W}_\nu\) and
	\(\hat{\mathcal T}_{\mu^\vee}\) commute with all stabilizer generators
	\(\hat A_a\) and \(\hat B_c\).
	
	First, \(\hat{\mathcal W}_\nu\) is a product of \(\hat Z\)-operators, so it commutes
	with every \(\hat B_c\). For \(a\in C_{P-1}(\Lambda)\), using the Weyl relations \eqref{eq:Weyl} we get
	\begin{equation}
		\hat A_a\hat{\mathcal W}_\nu
		=
		\omega^{-(-1)^P\sum_{b\in C_P(\Lambda)}
			\varepsilon(b,a)\nu_b}
		\hat{\mathcal W}_\nu\hat A_a .
	\end{equation}
	The exponent is the coefficient of \(a\) in \(\partial\nu\) (see \eqref{eq:boundary}). Since
	\(\nu\in Z_P(\Lambda)\), we have \(\partial\nu=0\), and therefore $[\hat A_a,\hat{\mathcal W}_\nu]=0$.
	
	Similarly, \(\hat{\mathcal T}_{\mu^\vee}\) is a product of \(\hat X\)-operators, so it
	commutes with every \(\hat A_a\). Set $\eta:=\vartheta_\Lambda(\mu^\vee)\in \Omega^P(\Lambda)$. For \(c\in C_{P+1}(\Lambda)\), using the Weyl relations again, we find
	\begin{equation}
		\hat B_c\hat{\mathcal T}_{\mu^\vee}
		=
		\omega^{(-1)^P\sum_{b\in C_P(\Lambda)}
			\varepsilon(c,b)\eta_b}
		\hat{\mathcal T}_{\mu^\vee}\hat B_c
		=
		\omega^{(-1)^P(d\eta)_c}
		\hat{\mathcal T}_{\mu^\vee}\hat B_c .
	\end{equation}
	By the primal-dual intertwining relation (Prop.~\ref{prop:intertwine}\ref{prop:vartheta-intertwines}),
	\begin{equation}
		d\eta
		=
		d\,\vartheta_\Lambda(\mu^\vee)
		=
		\vartheta_\Lambda(\partial\mu^\vee)
		=
		0,
	\end{equation}
	because \(\mu^\vee\in Z_{d-P}(\Lambda^\vee)\). Thus $[\hat B_c,\hat{\mathcal T}_{\mu^\vee}]=0$.
	
	Therefore both \(\hat{\mathcal W}_\nu\) and \(\hat{\mathcal T}_{\mu^\vee}\) preserve the
	common \(+1\)-eigenspace of the stabilizers, namely
	\(\mathcal H_{\mathrm{code}}\), proving Prop.~\ref{prop:WTpreserves}.
	
	\item If \(\nu=\partial C\), with \(C=\sum_c C_c\,c\), then
	$
	\nu_b
	=
	\sum_{c\in C_{P+1}(\Lambda)}C_c\varepsilon(c,b)
	$.
	Therefore
	\begin{equation}
		\hat{\mathcal W}_\nu
		=
		\prod_{b\in C_P(\Lambda)}\hat Z_b^{\,\nu_b}
		=
		\prod_{c\in C_{P+1}(\Lambda)}\hat B_c^{\,C_c}.
	\end{equation}
	This proves Prop.~\ref{prop:spatial-WT}\ref{prop:Wtrivial}.
	
	\item Similarly, if \(\mu^\vee=\partial\Xi^\vee\) for some
	\((d-P+1)\)-chain \(\Xi^\vee\), set $\eta:=\vartheta_\Lambda(\Xi^\vee)\in\Omega^{P-1}(\Lambda)$. Then
	by the primal-dual intertwining relation, $\vartheta_\Lambda(\mu^\vee)=d\eta$.
	Hence
	\begin{equation}
		\hat{\mathcal T}_{\mu^\vee}
		=
		\prod_{b\in C_P(\Lambda)}
		\hat X_b^{\,(-1)^P(d\eta)_b}
		=
		\prod_{a\in C_{P-1}(\Lambda)}
		\hat A_a^{\,\eta_a},
	\end{equation}
	which proves Prop.~\ref{prop:spatial-WT}\ref{prop:Ttrivial}.
	
	\item Finally, using the Weyl relation \eqref{eq:Weyl} we get
	\begin{equation}
		\hat X_b^{\,(-1)^P I_\Lambda(\mu^\vee,b)}
		\hat Z_b^{\,\nu_b}
		=
		\omega^{-(-1)^P I_\Lambda(\mu^\vee,b)\nu_b}
		\hat Z_b^{\,\nu_b}
		\hat X_b^{\,(-1)^P I_\Lambda(\mu^\vee,b)}.
	\end{equation}
	Multiplying over \(b\in C_P(\Lambda)\) gives Prop.~\ref{prop:spatial-WT}\ref{prop:WTcommutator}
	\begin{equation}
		\hat{\mathcal T}_{\mu^\vee}\hat{\mathcal W}_\nu
		=
		\omega^{-(-1)^P\sum_b I_\Lambda(\mu^\vee,b)\nu_b}
		\hat{\mathcal W}_\nu\hat{\mathcal T}_{\mu^\vee}
		=
		\omega^{-(-1)^P I_\Lambda(\mu^\vee,\nu)}
		\hat{\mathcal W}_\nu\hat{\mathcal T}_{\mu^\vee}.
	\end{equation}
\end{enumerate}

\subsection{Untwisted Trotter kernels}
\label{app:vanillaTr}

The fixed-\(M\) Trotterized trace is \eqref{eq:ZM-def}
\begin{equation}
	Z_M
	=
	\tr_{\mathcal H}
	\left[
	\left(
	e^{-\beta\hat H_z/M}e^{-\beta\hat H_x/M}
	\right)^M
	\right].
\end{equation}
Inserting \(M\) resolutions of the identity in the \(\hat Z\)-basis and using the
periodic convention \(s(M)=s(0)\), we obtain
\begin{equation}
	Z_M
	=
	\sum_{s(0),\dots,s(M-1)\in\Omega^P(\Lambda)}
	\prod_{i=0}^{M-1}
	\matrixel{s(i)}{e^{-\beta\hat H_z/M}}{s(i)}
	\matrixel{s(i)}{e^{-\beta\hat H_x/M}}{s(i+1)}.
	\label{eq:ZM-slices}
\end{equation}
We now compute the two one-slice kernels.

\subsubsection{The diagonal \(Z\)-kernel}

For \(c\in C_{P+1}(\Lambda)\), the $Z$-stabilizer is $\hat B_c=\prod_{b\subset c}\hat Z_b^{\,\varepsilon(c,b)}$ (see \eqref{eq:AB}).
Therefore
\begin{equation}
	\hat B_c\ket{s}
	=
	\omega^{(ds)_c}\ket{s},
	\qquad
	(ds)_c
	=
	\sum_{b\in C_P(\Lambda)}\varepsilon(c,b)s_b.
	\label{eq:Bc-eigenvalue}
\end{equation}
It follows that the $Z$-projector from \eqref{eq:projectors} acts on the $Z$-eigenbasis \eqref{eq:Z-basis} as
\begin{equation}
	\hat{\mathcal B}_c\ket{s}
	=
	\frac1N\sum_{m\in\mathbb Z_N}\omega^{m(ds)_c}\ket{s}
	=
	\delta_N\left((ds)_c\right)\ket{s}.
	\label{eq:Bproj-eigenvalue}
\end{equation}
The \(Z\)-source term (see \eqref{eq:H1}) is also diagonal:
\begin{equation}
	\exp\left(
	\frac{\beta}{M}\sum_{n\in\mathbb Z_N}h_b^{(n)}\hat Z_b^n
	\right)\ket{s}
	=
	\exp\left(
	\frac{\beta}{M}\sum_{n\in\mathbb Z_N}h_b^{(n)}\omega^{n s_b}
	\right)\ket{s}.
	\label{eq:h-source-diagonal}
\end{equation}
Using the definition of \(\hat H_z\) \eqref{eq:Hz}, we obtain
\begin{equation}
	\matrixel{s(i)}{e^{-\beta\hat H_z/M}}{s(i)}
	=
	\prod_{c\in C_{P+1}(\Lambda)}
	W_{c(i)}^{\parallel}\left((ds(i))_c\right)
	\prod_{b\in C_P(\Lambda)}
	V_{b(i)}^{\parallel}\left(s(i)_b\right),
	\label{eq:Hz-kernel}
\end{equation}
where the weight functions are defined in \eqref{eq:parallel-weights} as
\begin{subequations}
	\begin{align}
		W_{c(i)}^{\parallel}(x)
		&=
		\exp\left(\frac{\beta K_c}{M}\delta_N(x)\right)
		=
		1+\left(e^{\beta K_c/M}-1\right)\delta_N(x),
		\\
		V_{b(i)}^{\parallel}(x)
		&=
		\exp\left(
		\frac{\beta}{M}\sum_{n\in\mathbb Z_N}h_b^{(n)}\omega^{nx}
		\right).
	\end{align}
\end{subequations}

\subsubsection{The off-diagonal \(X\)-kernel}

We next compute the kernel of \(e^{-\beta\hat H_x/M}\). Since all terms in \(\hat H_x\)
commute, we may separate the \(\hat{\mathcal A}\)-projector part from the local
\(\hat X\)-source part.

For each \(a\in C_{P-1}(\Lambda)\), the operator \(\hat{\mathcal A}_a\) \eqref{eq:projectors} is a projector, so
\begin{equation}
	e^{(\beta J_a/M)\hat{\mathcal A}_a}
	=
	1+\left(e^{\beta J_a/M}-1\right)\hat{\mathcal A}_a
	=
	\sum_{n_a\in\mathbb Z_N}
	\left[
	\delta_N(n_a)+\frac{e^{\beta J_a/M}-1}{N}
	\right]
	\hat A_a^{\,n_a}.
	\label{eq:projector-exp-sum}
\end{equation}
Using $\hat A_a=\prod_{b\supset a}\hat X_b^{\,(-1)^P\varepsilon(b,a)}$,
we get, for \(n\in\Omega^{P-1}(\Lambda)\),
\begin{equation}
	\prod_{a\in C_{P-1}(\Lambda)}\hat A_a^{\,n_a}
	=
	\prod_{b\in C_P(\Lambda)}
	\hat X_b^{\,(-1)^P(dn)_b},
	\qquad
	\prod_{a\in C_{P-1}(\Lambda)}\hat A_a^{\,n_a}\ket{s}
	=
	\ket{s+(-1)^Pdn}.
	\label{eq:A-product-dn}
\end{equation}
Therefore
\begin{equation}
	\begin{aligned}
		\matrixel{t}{
			e^{\frac{\beta}{M}\sum_{a\in C_{P-1}(\Lambda)}J_a\hat{\mathcal A}_a}
		}{s}
		=
		\sum_{n\in\Omega^{P-1}(\Lambda)}
		\prod_{a\in C_{P-1}(\Lambda)}
		\left[
		\delta_N(n_a)+\frac{e^{\beta J_a/M}-1}{N}
		\right]
		\\
		\times
		\prod_{b\in C_P(\Lambda)}
		\delta_N\left(t_b-s_b-(-1)^P(dn)_b\right).
	\end{aligned}
	\label{eq:HA-kernel}
\end{equation}

For the local \(X\)-source, introduce the spectral projectors of \(\hat X\):
\begin{equation}
	\hat\Pi_j^{(X)}
	:=
	\frac1N\sum_{n=0}^{N-1}\omega^{-jn}\hat X^n,
	\qquad
	j\in\mathbb Z_N.
	\label{eq:X-projectors}
\end{equation}
They satisfy
\begin{equation}
	\hat X\hat\Pi_j^{(X)}=\omega^j\hat\Pi_j^{(X)},
	\qquad
	\sum_{j\in\mathbb Z_N}\hat\Pi_j^{(X)}=\hat 1,
	\qquad
	\matrixel{s_b}{\hat\Pi_j^{(X)}}{t_b}
	=
	\frac1N\omega^{j(t_b-s_b)}.
	\label{eq:PiX-proj}
\end{equation}
Consequently,
\begin{align}
	&\matrixel{s}{
		e^{\frac{\beta}{M}\sum_{b\in C_P(\Lambda)}
			\sum_{n\in\mathbb Z_N}g_b^{(n)}\hat X_b^n}
	}{t}
	=
	\prod_{b\in C_P(\Lambda)}
	\sum_{j\in\mathbb Z_N}
	\exp\left(
	\frac{\beta}{M}\sum_{n\in\mathbb Z_N}g_b^{(n)}\omega^{nj}
	\right)
	\frac{\omega^{j(t_b-s_b)}}{N}.
	\label{eq:HXsource-kernel}
\end{align}

Convolving \eqref{eq:HA-kernel} and \eqref{eq:HXsource-kernel}, and using the
delta-functions to eliminate the intermediate field \(t\), gives
\begin{equation}
	\begin{aligned}
		\matrixel{s(i)}{e^{-\beta\hat H_x/M}}{s(i+1)}
		=&\;
		\sum_{n(i)\in\Omega^{P-1}(\Lambda)}
		\prod_{b\in C_P(\Lambda)}
		W_{\underline b(i)}^{\perp}
		\left(
		(dn(i))_b+(-1)^P\left(s(i+1)_b-s(i)_b\right)
		\right)
		\\
		&\; \hspace{5em} \times 
		\prod_{a\in C_{P-1}(\Lambda)}
		V_{\underline a(i)}^{\perp}\left(n(i)_a\right),
	\end{aligned}
	\label{eq:Hx-kernel}
\end{equation}
where the weight functions are defined in \eqref{eq:perp-weights} as
\begin{subequations}
	\begin{align}
		V_{\underline a(i)}^{\perp}(x)
		&=
		\delta_N(x)+\frac{e^{\beta J_a/M}-1}{N},
		\label{eq:Vperp}
		\\
		W_{\underline b(i)}^{\perp}(x)
		&=
		\sum_{j\in\mathbb Z_N}
		\exp\left(
		\frac{\beta}{M}\sum_{n\in\mathbb Z_N}g_b^{(n)}\omega^{nj}
		\right)
		\frac{\omega^{(-1)^Pjx}}{N}.
		\label{eq:Wperp}
	\end{align}
\end{subequations}

\subsubsection{Packaging into a spacetime \(P\)-cochain}
\label{app:ns2phi}

Define a spacetime \(P\)-cochain \(\phi\in\Omega^P(\underline\Lambda)\) by
\begin{equation}
	\phi_{b(i)}:=s(i)_b,
	\qquad
	\phi_{\underline a(i)}:=n(i)_a,
	\label{eq:phi-def}
\end{equation}
where \(b(i)=b\times\{i\}\) is a horizontal \(P\)-cell and
\(\underline a(i)=a\times[i,i+1]\) is a vertical \(P\)-cell, as defined in \eqref{eq:horizontalCell} and \eqref{eq:verticalCell} respectively.  Substituting \eqref{eq:Hz-kernel} and \eqref{eq:Hx-kernel} into
\eqref{eq:ZM-slices} and identifying the spacetime differential \eqref{eq:D-components} in the arguments of $W$ we get the untwisted spacetime partition function
\begin{equation}
	Z_M
	=
	\sum_{\phi\in\Omega^P(\underline\Lambda)}
	\prod_{c\in C_{P+1}(\underline\Lambda)}
	W_c\left((D\phi)_c\right)
	\prod_{u\in C_P(\underline\Lambda)}
	V_u(\phi_u).
	\label{eq:untwisted-spacetime}
\end{equation}
Here \(W_c\) denotes \(W^\parallel\) on horizontal \((P+1)\)-cells and \(W^\perp\) on
vertical \((P+1)\)-cells; similarly \(V_u\) denotes \(V^\parallel\) on horizontal
\(P\)-cells and \(V^\perp\) on vertical \(P\)-cells.

\subsection{Spatial Wilson and 't Hooft insertions}
\label{app:TrSp}

We next insert the spatial Wilson operator and 't Hooft operators \eqref{eq:quantumWT}.  The
decorated trace with only these spatial insertions is
\begin{equation}
	Z_M^{\mathrm{sp}}(\mu^\vee,\nu)
	:=
	\tr_{\mathcal H}
	\left[
	e^{-\beta\hat H_z/M}\,
	\hat{\mathcal W}_\nu\,
	e^{-\beta\hat H_x/M}\,
	\hat{\mathcal T}_{\mu^\vee}\,
	\left(e^{-\beta\hat H_z/M}e^{-\beta\hat H_x/M}\right)^{M-1}
	\right].
	\label{eq:spatially-decorated}
\end{equation}

The Wilson operator is diagonal:
\begin{equation}
	\matrixel{s(0)}{\hat{\mathcal W}_\nu}{s(0)}
	=
	\prod_{b\in C_P(\Lambda)}\omega^{\nu_b s(0)_b}
	=
	\chi\left(\int_\nu s(0)\right).
	\label{eq:spatial-W-factor}
\end{equation}
Equivalently, after packaging into \(\phi\), this is $\chi\left(\int_{\nu(0)}\phi\right)$,
where \(\nu(0)\in Z_P(\underline\Lambda)\) is the horizontal lift of \(\nu\) to the initial
time slice.

Using the definition of the 't Hooft operator $\hat{\mathcal T}_{\mu^\vee}
=
\prod_{b\in C_P(\Lambda)}
\hat X_b^{\,(-1)^PI_\Lambda(\mu^\vee,b)}$, and noting that $I_\Lambda(\mu^\vee, b) = \mu^\vee_{\theta_{\Lambda}^{-1}(b)} = \vartheta_{\Lambda}(\mu^\vee)_b$ (see \eqref{eq:intersection-basis}, \eqref{eq:vartheta}, and \eqref{eq:PX}) we find
\begin{equation}
	\hat{\mathcal T}_{\mu^\vee}\ket{s}
	=
	\ket{s+(-1)^P\vartheta_\Lambda(\mu^\vee)}.
	\label{eq:T-shift}
\end{equation}
Since the \(\hat X\)-kernel depends on \(s(i)\) and \(s(i+1)\) through the combination
\begin{equation}
	(dn(i))_b+(-1)^P\left(s(i+1)_b-s(i)_b\right),
\end{equation}
the insertion of \(\hat{\mathcal T}_{\mu^\vee}\) on the first slab changes the vertical
argument by
\begin{equation}
	(dn(i))_b+(-1)^P\left(s(i+1)_b-s(i)_b\right)
	\mapsto
	(dn(i))_b+(-1)^P\left(s(i+1)_b-s(i)_b\right)
	+\delta_{i,0}(\vartheta_\Lambda(\mu^\vee))_b.
	\label{eq:T-vertical-shift}
\end{equation}

Let \(\mu^\vee(0)\in Z_{d-P}(\underline\Lambda^\vee)\) denote the dual spacetime cycle
supported on the first dual slab and characterized by
\begin{equation}
	\left(\vartheta_{\underline\Lambda}(\mu^\vee(0))\right)_{\underline b(0)}
	=
	(\vartheta_\Lambda(\mu^\vee))_b,
	\qquad
	\left(\vartheta_{\underline\Lambda}(\mu^\vee(0))\right)_{c(i)}=0,
	\qquad
	\left(\vartheta_{\underline\Lambda}(\mu^\vee(0))\right)_{\underline b(i)}=0
	\quad(i\neq 0).
	\label{eq:mu0-cocycle}
\end{equation}
Then \eqref{eq:T-vertical-shift} is exactly the replacement
\begin{equation}
	D\phi
	\mapsto
	D\phi+\vartheta_{\underline\Lambda}(\mu^\vee(0)).
\end{equation}
Hence
\begin{equation}
	Z_M^{\mathrm{sp}}(\mu^\vee,\nu)
	=
	\sum_{\phi\in\Omega^P(\underline\Lambda)}
	\chi\left(\int_{\nu(0)}\phi\right)
	\prod_{c\in C_{P+1}(\underline\Lambda)}
	W_c\left((D\phi+\vartheta_{\underline\Lambda}(\mu^\vee(0)))_c\right)
	\prod_{u\in C_P(\underline\Lambda)}
	V_u(\phi_u).
	\label{eq:spatial-insertions-spacetime}
\end{equation}

\subsection{Euclidean twists}
\label{app:TrTwists}
\subsubsection{Electric twist}
\label{app:electric-twist}

For $\alpha=\sum_{a \in C_{P-1}(\Lambda)} \alpha_a a \in Z_{P-1}(\Lambda)$, using the orthogonality relations \eqref{eq:Aproj-relations} of the spectral projectors $\hat\Pi_{a,j}^{(A)} = \frac1N\sum_{m=0}^{N-1}\omega^{-jm}\hat A_a^{\,m}$ and the definition of the Euclidean twist operator \eqref{eq:electric-twist-operator} we get
\begin{align}
	\hat U_{e,a}(\alpha_a)e^{(\beta J_a/M)\hat{\mathcal A}_a}
	&=
	\sum_{j\in\mathbb Z_N}\lambda_{a,j+\alpha_a}\hat\Pi_{a,j}
	\nonumber\\
	&=
	\sum_{n_a\in\mathbb Z_N}
	\left[
	\delta_N(n_a)+\frac{e^{\beta J_a/M}-1}{N}
	\right]
	\omega^{\alpha_a n_a}\hat A_a^{\,n_a}.
	\label{eq:Ue-shift}
\end{align}
The factor \(\omega^{\alpha_a n_a}\) is harmless on the \(\delta_N(n_a)\) term, since
\(\delta_N(n_a)\omega^{\alpha_a n_a}=\delta_N(n_a)\). Therefore, relative to the untwisted
\(X\)-kernel, inserting \(\hat U_e(\alpha)\) on the \(i\)-th slab contributes the factor
\begin{equation}
	\prod_{a\in C_{P-1}(\Lambda)}\omega^{\alpha_a n(i)_a}
	=
	\chi\Biggl(\sum_{a\in C_{P-1}(\Lambda)}\alpha_a n(i)_a\Biggr).
	\label{eq:electric-slab-factor}
\end{equation}

Under the spacetime packaging \eqref{eq:phi-def}, this becomes $\chi\left(\int_{\underline\alpha(i)}\phi\right)$.
Multiplying over all \(M\) slabs gives
\begin{equation}
	\prod_{i=0}^{M-1}
	\chi\left(\int_{\underline\alpha(i)}\phi\right)
	=
	\chi\left(\int_{\Sus(\alpha)}\phi\right),
	\label{eq:Sus-alpha-factor}
\end{equation}
where $\Sus(\alpha)=\sum_{i=0}^{M-1}\underline\alpha(i)\in Z_P(\underline\Lambda)$ as in \eqref{eq:electric-suspension}.

\subsubsection{Magnetic twist}
\label{app:magnetic-twist}

For $c\in C_{P+1}(\Lambda)$ the spectral projectors $\hat\Pi_{c,j}^{(B)} = \frac{1}{N} \sum_{m=0}^{N-1} \omega^{-jm} \hat B_c^m$ act on the $Z$-eigenbasis \eqref{eq:Z-basis} as
\begin{equation}
	\hat\Pi_{c,j}^{(B)}\ket{s}
	=
	\delta_N\left((ds)_c-j\right)\ket{s}.
\end{equation}
Therefore, for $\beta^\vee\in Z_{d-P-1}(\Lambda^\vee)$ with $\beta_c = (\vartheta_\Lambda(\beta^\vee))_c$, the  magnetic twist operator \eqref{eq:magnetic-twist-operator} acts on the $Z$-eigenbasis as
\begin{equation}
	\hat U_{m,c}(\beta_c)e^{(\beta K_c/M)\hat{\mathcal B}_c}\ket{s}
	=
	\rho_{c,(ds)_c+\beta_c}\ket{s}
	=
	\exp\left(
	\frac{\beta K_c}{M}\delta_N\left((ds)_c+\beta_c\right)
	\right)\ket{s}.
	\label{eq:Um-shift}
\end{equation}
Hence
\begin{equation}
	\matrixel{s(i)}{\hat U_m(\beta^\vee)e^{-\beta\hat H_z/M}}{s(i)}
	=
	\prod_{c\in C_{P+1}(\Lambda)}
	W_{c(i)}^{\parallel}
	\left((ds(i))_c+(\vartheta_\Lambda(\beta^\vee))_c\right)
	\prod_{b\in C_P(\Lambda)}
	V_{b(i)}^{\parallel}\left(s(i)_b\right).
	\label{eq:magnetic-twist-zkernel}
\end{equation}

Thus the magnetic twist is precisely the horizontal cocycle shift
\begin{equation}
	D\phi
	\mapsto
	D\phi+\vartheta_{\underline\Lambda}(\Sus(\beta^\vee)),
\end{equation}
where the suspended dual cycle is defined as in \eqref{eq:eta-hor-beta-components}.

\subsection{Fully decorated trace}
\label{app:TrFull}

Consider the fully decorated trace \eqref{eq:decoratedTr}. All \(X\)-type operators appearing in
\(\hat{\mathcal T}_{\mu^\vee}\hat U_e(\alpha)e^{-\beta\hat H_x/M}\) commute with one
another, and all \(Z\)-type operators appearing in
\(\hat{\mathcal W}_\nu\hat U_m(\beta^\vee)e^{-\beta\hat H_z/M}\) commute with one another.
Thus the preceding kernel calculations apply directly.

The decorated \(X\)-kernel is
\begin{align}
	&\matrixel{s(i)}
	{
		\hat{\mathcal T}_{\mu^\vee}^{\,\delta_{i,0}}
		\hat U_e(\alpha)e^{-\beta\hat H_x/M}
	}
	{s(i+1)}
	\nonumber\\
	&\quad=
	\sum_{n(i)\in\Omega^{P-1}(\Lambda)}
	\chi\left(\sum_{a\in C_{P-1}(\Lambda)}\alpha_a n(i)_a\right)
	\prod_{a\in C_{P-1}(\Lambda)}
	V_{\underline a(i)}^{\perp}\left(n(i)_a\right)
	\nonumber\\
	&\qquad\qquad\times
	\prod_{b\in C_P(\Lambda)}
	W_{\underline b(i)}^{\perp}
	\left(
	(dn(i))_b
	+(-1)^P\left(s(i+1)_b-s(i)_b\right)
	+\delta_{i,0}(\vartheta_\Lambda(\mu^\vee))_b
	\right).
	\label{eq:decorated-x-kernel}
\end{align}
The decorated \(Z\)-kernel is
\begin{align}
	&\matrixel{s(i)}
	{
		\hat{\mathcal W}_\nu^{\,\delta_{i,0}}
		\hat U_m(\beta^\vee)e^{-\beta\hat H_z/M}
	}
	{s(i)}
	\nonumber\\
	&\quad=
	\chi\left(\int_\nu s(i)\right)^{\delta_{i,0}}
	\prod_{c\in C_{P+1}(\Lambda)}
	W_{c(i)}^{\parallel}
	\left((ds(i))_c+(\vartheta_\Lambda(\beta^\vee))_c\right)
	\prod_{b\in C_P(\Lambda)}
	V_{b(i)}^{\parallel}\left(s(i)_b\right).
	\label{eq:decorated-z-kernel}
\end{align}

Define the total electric and magnetic spacetime cycles
\begin{equation}
	q_e(\nu,\alpha)
	:=
	\nu(0)+\Sus(\alpha)\in Z_P(\underline\Lambda),
	\label{eq:total-electric-cycle}
\end{equation}
and
\begin{equation}
	q_m(\mu^\vee,\beta^\vee)
	:=
	\mu^\vee(0)+\Sus(\beta^\vee)\in Z_{d-P}(\underline\Lambda^\vee).
	\label{eq:total-magnetic-cycle}
\end{equation}
The closure of these cycles follows from
\(\partial\nu=0\), \(\partial\alpha=0\), \(\partial\mu^\vee=0\), and
\(\partial\beta^\vee=0\), together with periodicity in the Euclidean-time direction.

Combining \eqref{eq:decorated-x-kernel} and \eqref{eq:decorated-z-kernel}, and then
using the spacetime packaging \eqref{eq:phi-def}, gives the exact background formula
\begin{equation}
	Z_M(\mu^\vee,\nu;\beta^\vee,\alpha)
	=
	\sum_{\phi\in\Omega^P(\underline\Lambda)}
	\chi\left(\int_{q_e(\nu,\alpha)}\phi\right)
	\prod_{c\in C_{P+1}(\underline\Lambda)}
	W_c\left(
	(D\phi+\vartheta_{\underline\Lambda}(q_m(\mu^\vee,\beta^\vee)))_c
	\right)
	\prod_{u\in C_P(\underline\Lambda)}
	V_u(\phi_u).
	\label{eq:decorated-spacetime}
\end{equation}
This is the desired fused form.  The spatial Wilson insertion and the electric Euclidean
twist combine into the single spacetime Wilson factor supported on
\(q_e(\nu,\alpha)\), while the spatial 't Hooft insertion and the magnetic Euclidean twist
combine into the single cocycle shift
\(\vartheta_{\underline\Lambda}(q_m(\mu^\vee,\beta^\vee))\). Setting $\mu^\vee=\nu=\beta^\vee=\alpha=0$ we recover the partition function \eqref{eq:untwisted-spacetime} in the trivial background.

\printbibliography

\end{document}